\documentclass[a4paper,11pt]{article}
\pdfoutput=1 % if your are submitting a pdflatex (i.e. if you have
\usepackage{jheppub} % for details on the use of the package, please
\usepackage{amsthm}
\usepackage[T1]{fontenc} % if needed
\usepackage{dsfont}
\usepackage{amsmath}
\usepackage{bbm}
\usepackage[capitalize]{cleveref}
\usepackage{physics}
\usepackage[normalem]{ulem}
\usepackage{caption}
\usepackage{subcaption}
\usepackage{adjustbox}
\usepackage{comment}
\usepackage{cleveref}
\usepackage[section]{placeins}
\usepackage{mathrsfs}

\newcommand{\Gmc}{\mathcal{G}}
\newcommand{\Pmc}{\mathcal{P}}
\newcommand{\Dgv}{\int \prod_{v\in V_1}\dd{g_v}}
\newcommand{\DgvG}{\int \prod_{v\in V_\Gamma}\dd{g_v}}

\newtheorem{theorem}{Theorem}[section]

\newtheorem{definition}[theorem]{Definition}
\newtheorem{example}[theorem]{Example}

\newtheorem{lemma}[theorem]{Lemma}

\newtheorem{remark}[theorem]{Remark}

\DeclareMathOperator{\image}{Image}
\DeclareMathOperator{\interior}{Int}
\DeclareMathOperator{\exterior}{Ext}
\makeatletter
\newcommand{\crefnames}[3]{%
  \@for\next:=#1\do{%
    \expandafter\crefname\expandafter{\next}{#2}{#3}%
  }%
}
\makeatother
\crefnames{section}{Section}{Sections}
\newcommand{\defi}{{:=}}

\newcommand{\linearops}[1]{{\mathcal{L}({#1})}}

\newcommand{\id}{\mathds{1}}
\newcommand{\idV}{\mathds{1}_{\mathcal{H}_V}}

\title{\boldmath Gauging the bulk: generalized gauging maps and holographic codes}
\author[a]{Kfir Dolev}
\author[b]{Vladimir Calvera}
\author[a]{Samuel S. Cree}
\author[a]{Dominic J. Williamson}

\affiliation[a]{Stanford Institute for Theoretical Physics, Stanford University, 382 Via Pueblo Mall, Stanford, CA 94305-4060, U.S.A.}
\affiliation[b]{Department of Physics, Stanford University, Stanford, CA 94305, USA}

\emailAdd{dolev@stanford.edu}
\emailAdd{fvcalver@stanford.edu}
\emailAdd{scree@stanford.edu}
\emailAdd{domjw@stanford.edu}

\abstract{
Gauging is a general procedure for mapping a quantum many-body system with a global symmetry to one with a local gauge symmetry. 
We consider a generalized gauging map that does not enforce gauge symmetry at all lattice sites, and show that it is an isometry on the full input space including all charged sectors. 
We apply this generalized gauging map to convert global-symmetric bulk systems of holographic codes to gauge-symmetric bulk systems, and vice versa, while preserving duality with a global-symmetric boundary. 
We separately construct holographic codes with gauge-symmetric bulk systems by directly imposing gauge-invariance constraints onto existing holographic codes, and show that the resulting bulk gauge symmetries are dual to boundary global symmetries. Combining these ideas produces a toy model that captures several interesting features of holography -- it exhibits a rudimentary sort of dynamical duality, can be modified to demonstrate the relationship between metric fluctuations and approximate error-correction, and serves as an illustration for certain no-go theorems concerning symmetries in holography.
Finally, we apply the generalized gauging map to construct codes with arbitrary transversal gate sets -- for any compact Lie group, we use a symmetry-preserving truncation scheme to construct covariant finite-dimensional approximate holographic codes.
}

\begin{document} 
\maketitle
\flushbottom

\section{Introduction}
\label{sec:intro}

The anti-de Sitter/conformal field theory (AdS/CFT) correspondence~\cite{Maldacena1997,Witten1998} is a general pattern of holographic dualities~\cite{Hooft1993,Susskind1995} that relate a Bulk theory of quantum gravity in asymptotically AdS spacetime to a Boundary\footnote{Note: we capitalize Boundary throughout this paper when referring to the system that arises as one side of the AdS/CFT duality or holographic codes, to distinguish from other mathematical uses of the word boundary.} CFT in a space with one fewer dimension.
Recent years have seen the development of quantum error correcting code toy models~\cite{HaPPY,haydenHolographicDualityRandom2016,Marolf,Yang2016,harlow,preskillpastawski,caoApproximateBaconShorCode2020,farrellyParallelDecodingMultiple2020,harrisCalderbankSteaneShorHolographicQuantum2018,harrisMaximumLikelihoodDecoder2020,Kohler2019,Apel2021,Cao2021,Jahn2021,noncliffords} that capture some of the most puzzling aspects~\cite{almheiriBulkLocalityQuantum2015} of the AdS/CFT correspondence.
These toy models have provided valuable insight on holography by providing an explicit, finite-dimensional theory amenable to concrete calculations~\cite{Jahn2021}.
Although these models are useful for understanding static features like entanglement and geometry, they have had limited success at replicating aspects of dynamics and the inclusion of symmetry~\cite{noncliffords,Faist_2020}, particularly gauge symmetries~\cite{Kogut:1974ag,Kogut1979}.
%However, such toy models fail to capture some aspects of holographic duality, and there is an ongoing effort to better understand their deficiencies and how to improve them~. 

%Symmetries play a vital role in holography, and constitute a significant part of the 
The interplay of symmetry and quantum gravity has been the subject of much investigation~\cite{Misner1957,Polchinski2004,Banks2011,Harlow2016,harlow-ooguri,Harlow2021,Chen2021,Hsin2020,Belin2020}.
In Ref.~\cite{harlow-ooguri} it was shown that AdS/CFT can never exhibit a Boundary global symmetry that is dual to a Bulk global symmetry, and that this eliminates the possibility of the latter's existence entirely.
Instead, all Boundary global symmetries must be dual to Bulk gauge symmetries, and vice versa.
%Such dualities are clearly of interest for a range of applications -- namely, whenever one wishes to study a gauge theory in AdS space or the internal global symmetries of a CFT.
These statements hold even for spacetime symmetries, and in particular time evolution. Models of time evolution in holographic codes are of interest because they may be used to practically implement ideas in Ref.~\cite{May_2020}, such as the efficient implementation of non-local quantum computations. 
Existing toy models are not developed enough for this purpose as they either do not implement a proper symmetry duality \cite{Kohler2019}, or they do not achieve entangling dynamics \cite{Osborne_2020}. 
This motivates the construction of toy models that implement gauge symmetries in the Bulk as a means for concretely studying possible avenues towards time evolution, as well as studying symmetries in holography more generally.
%which is the main focus of this paper.

Some of the major issues that remain for the incorporation of symmetries into toy models are: 
\begin{enumerate}
    \item the construction of a toy model with an explicit realization of Bulk gauge symmetry dual to a Boundary global symmetry\footnote{while the code in \cite{Marolf}, referred to in this paper as the ``LOTE'' code, can accommodate a gauge invariant Hilbert space, an explicit choice of gauge symmetry was not studied there.}
    \item a deeper investigation of the conditions under which current holographic codes, such as HaPPY, do allow for Bulk global symmetries that are dual to Boundary global symmetries despite the results of \cite{harlow-ooguri}, and finally
    \item the construction of a toy model which properly incorporates a Boundary global time evolution symmetry, along with a dual Bulk gauge symmetry.
\end{enumerate}

In this work we resolve the first issue by imposing Bulk gauge symmetry on the holographic code introduced in \cite{Marolf}. We then adapt the proof used for AdS/CFT in \cite{harlow-ooguri} to show that any holographic code whose Bulk exhibits such a symmetry must have a dual Boundary global symmetry. Thus the construction exhibits the desired duality. 

We suggest a resolution to the second issue by showing that an apparent global symmetry in the Bulk can emerge from a gauge symmetry when one restricts to a specific subspace of the Bulk gauged system.
This can explain why global symmetries appear in a context in which only gauge symmetries are allowed; the symmetries merely appear to be global due to this subspace restriction.
To be more concrete, we generalize the ``gauging map'' defined in Ref.~\cite{Cirac_2015}, so that it isometrically maps systems with global symmetries into systems with gauge symmetries.
We show that the image of this generalized gauging map is a ``fixed-flux sector'' -- a subspace characterized by specific eigenvalues of all relevant Wilson operators \cite{Wilson1974}.
Thus, it can be inverted to map this subspace of a gauged system into a system with global symmetry, a process we refer to as ``ungauging''.
Then an explicit way to construct a holographic code with Bulk global symmetry is to take one with Bulk gauge symmetry, ungauge it in this way to project out all but a single fixed-flux sector, and obtain a system in which the original gauge symmetry now appears as a global one.
%The gauging map is capable of isometrically gauging local quantum many body systems in such a way that global symmetries are converted to gauge symmetries.
%After restricting a Bulk with gauge symmetry to a fixed-flux sector, we can use the gauging map to ``ungauge'' it to realize a Bulk global symmetry which is dual to a Boundary global symmetry.
This circumvents the ``no Bulk global symmetries'' theorem of Ref.~\cite{harlow-ooguri} which requires the global symmetry to appear in any Bulk subspace.
We argue that this loophole explains the presence of global symmetries in HaPPY, as its Bulk can be understood as having been restricted to a fixed-flux sector.

Despite applying only to internal i.e.\ non-spacetime symmetries, the gauging map also offers insight into the third issue, how Bulk time evolution appears. Drawing an analogy to diffeomorphism invariance (the gauge symmetry of gravity) we interpret the fixed-flux subspace as a fixed geometry subspace, and the application of the ungauging map as the emergence of Bulk locality in such a subspace. The duality thus generated between Bulk and Boundary global symmetries can then be interpreted as a trivial kind of time evolution under a non-interacting Hamiltonian, in which all information is stuck in place. 

There is an interesting consequence of this work that may have applications in the study of pure quantum error correction. We can generalize the construction of a code with Bulk gauge symmetry to allow for an arbitrary finite or continuous compact Lie group. By using the map to ungauge the Bulk, the Bulk gauge symmetry is converted to a global symmetry, as before. In the language of quantum error correction, we have uncovered a new method to construct codes with arbitrary transversal gate sets. For continuous (compact Lie) groups the code is infinite dimensional~\cite{Kang2021,Gesteau2020a,Gesteau2020b}, but we show how to use a truncation scheme which produces finite dimensional approximations to the exact continuous code while preserving the covariance property, as must be the case to respect the results of Ref.~\cite{Faist_2020}. 

The structure of the paper is as follows. In \cref{sec:preliminaries} we introduce notation, definitions, and background that is used throughout the paper. Of particular importance are the definitions of global symmetries, gauge symmetries, and dualities between them. In \cref{sec:generalized-gauging-map} we define the gauging map and prove its relevant properties, although some more detailed proofs are left to \cref{app:gauge-proofs}. In \cref{sec: holographic-codes} we give a working definition of holographic codes and provide illustrative examples of the two main types of holographic code we are interested in, those with unconstrained local Bulk degrees of freedom, and those with Bulk constraints due to gauge symmetries. We also prove here that any holographic code whose Bulk exhibits a gauge symmetry has a dual Boundary global symmetry.
As part of this construction, we derive new results for quantum error-correcting codes to do with the implementation of a logical unitary representation via a physical unitary representation on the complement of a correctable subsystem, which may be of independent interest -- see \cref{sec:qecc} for details.
Finally, in \cref{sec:applications-to-holography} we illustrate a number of applications of the gauging map to holographic codes, including gauging and ungauging holographic codes, and commentary on the resulting models of holography.

\section{Setup}
\label{sec:preliminaries}

In this section we review and introduce notation for the background concepts needed to understand the results of this work. Readers with a background in quantum error correction and/or gauge theory will likely have an easier time,  though neither of these is prerequisite. In \ref{subsec: math-notation} we introduce notation for basic definitions of linear algebra (for discussing quantum systems) and graph theory (for giving these systems geometrical structure). In \ref{subsec: QMBS} we give definitions for quantum many body systems that keep track of their local tensor product structure and any constraints they may have, which are helpful in keeping track of the variety of systems being mapped to each other throughout the rest of the paper. In \ref{subsec: global-sym} we define what it means to equip a system with a global symmetry, and what it means for a map between systems to exhibit a global/global symmetry duality. Finally in \ref{subsec: gauge-sym} we define what it means for a system to have a gauge symmetry, show how the Hilbert space of such systems naturally breaks up into ``fixed-flux sectors'', and what it means for a map between a gauged system and an unconstrained system to exhibit a gauge/global duality.

As part of our conventions we capitalize Boundary and Bulk throughout this paper when referring to the two systems dual to each other in AdS/CFT or holographic codes. This is to avoid confusion with other meanings of the word boundary. 

\subsection{Mathematical notation} \label{subsec: math-notation}

We use the following notation for the basic mathematical objects that appear in the work that follows.
We encourage the reader to skip this section on a first read and use it merely for reference where necessary.

\subsubsection{Linear algebra}

\begin{itemize}
    \item Hilbert spaces are denoted by $\mathcal{H}_x$ with $x$ some label.
    \item For a Hilbert space $\mathcal{H}_x$, we denote the set of all linear operators acting on $\mathcal{H}_x$ by $\linearops{\mathcal{H}_x}$.
    \item For a Hilbert space $\mathcal{H}_x$ we denote the identity element of $\linearops{\mathcal{H}_x}$ by $\id_x$, or $\id$ when the space is clear from context. We reserve the symbol $I$ for the identity of a group.
    \item For a set $V$, and a set of Hilbert spaces $\{\mathcal{H}_v|v\in V\}$, we define $\mathcal{H}_V\equiv \bigotimes_{v\in V}\mathcal{H}_v$. We refer to the labels $v$ as
    \textbf{subsystems}.
    \item For a Hilbert space $\mathcal{H}$ and a projector $\Pi:\mathcal{H}\rightarrow\mathcal{H}$ we denote by $\Pi\mathcal{H}$ the Hilbert space which is the image of $\Pi$.
	\item For a qubit Hilbert space with some choice of basis $\left\{ \ket{0},\ket{1} \right\}$, we define the states $\ket{\pm} \defi \frac{1}{\sqrt{2}}\left( \ket{0} \pm \ket{1} \right)$.
		We also define Pauli operators $X = \ketbra{+} - \ketbra{-}$, $Z = \ketbra{0} - \ketbra{1} $, and $Y = iXZ$.
	\item We use the word \emph{support} in the conventional, slightly ambiguous way.
	\begin{itemize}
	    \item  An operator $O\in \linearops{\mathcal{H}}$ with $\mathcal{H} = \mathcal{H}_A \otimes \mathcal{H}_B$ is supported on \emph{subsystem} $\mathcal{H}_A$ if it takes the form $O_A\otimes \id_B$ with $O_A \in \linearops{\mathcal{H}_A}$.
	    \item  An operator $O\in \linearops{\mathcal{H}}$ with $\mathcal{H} = \mathcal{H}_1 \oplus \mathcal{H}_2$ is supported on \emph{subspace} $\mathcal{H}_1$ if it takes the form $O_1\oplus 0_2$ with $O_1 \in \linearops{\mathcal{H}_1}$ and $0_2$ the zero operator in $\linearops{\mathcal{H}_2}$.
	\end{itemize}
	Which of these two meanings is intended should be clear from context.
\end{itemize}

\subsubsection{Graph theory}
\label{subsubsec: graph-theory-background}
\begin{itemize}
    \item A \textbf{directed graph} $\Lambda=(V,E)$ consists of a set $V$ also called \textbf{the vertices} and set $E\subseteq  V\times V$ also called \textbf{the edges}. 
    %\item An \textbf{undirected graph} $\lambda=(V',E')$ consists of an $V$ also called \textbf{the vertices} and set $E'\subseteq  \{\{e,e'\}|e,e'\in V' \text{ and }e\neq e'\}$ also called \textbf{the edges}. Note that undirected graphs are used only in this sub-subsection to aid in defining notions for directed graphs. Thus there ought not be later confusion about the types of elements $E$ contains.
    
    \item A directed graph $\Lambda=(V,E)$ is said to be \textbf{oriented} if for any two vertices $v,u\in V$, there is at most one edge connecting the two, i.e. $(v,u)\in E\rightarrow (u,v)\not\in E$. We refer to such a graph simply as an \textbf{oriented graph}, and every $\Lambda$ from here on refers to an oriented graph.
    
    %\item Let $\Lambda=(V,E)$ be an oriented graph and $\lambda=(V',E')$ an undirected graph. We say $\lambda$ is the unoriented version of $\Lambda$ if $V'=V$ and $E'=\{\{e,f\}|(e,f)\in E\}$.
    
    \item A \textbf{path} in $\Lambda$ is a tuple of vertices $(v_1,v_2,\ldots ,v_n)$ such that for all $i\in\{1,\ldots ,n\},$ $v_i\in V$ and for all $i\in\{1,\ldots ,n-1\}$ $(v_i,v_{i+1})\in E$ or $(v_{i+1},v_{i})\in E$.
    We emphasize that although the path itself has a specific direction (from $v_1$ to $v_n$), there is no restriction on the direction of each individual edge.
    
    %\item Let $\lambda=(V',E')$ be an undirected graph. A \textbf{path} in $\lambda$ is a tuple of vertices $(v'_1,v'_2,\ldots ,v'_n)$ such that for all $i\in\{1,\ldots ,n\}$ $v'_i\in V'$ and for all $i\in\{1,\ldots ,n-1\}$ $\{v'_i,v'_{i+1}\}\in E'$.
    
    \item  A \textbf{cycle} in $\Lambda$ is a path $(v_1,v_2,\ldots ,v_n)$ such that $v_1=v_n$.
    
    \item We say that $\Lambda$ is connected if any two vertices are connected by a path, noting again that there is no constraint on the directions of the edges comprising the path\footnote{This notion of connectedness for a directed graph is sometimes called ``weakly connected'', as opposed to ``strongly connected'' when one requires $(v_i,v_{i+1}) \in E$ for each path.}.

    \item For any edge $e$ in $\Lambda$ we define the vertices $e^+$ and $e^-$ such that $e=(e^+,e^-)$. For a set of edges $F\subseteq E$ and a vertex $v\in V$ we define 
    $F(v)=\{e\in F | v \in e\}$,
    $F^+(v)=\{e\in F(v) | e^+=v\}$, and $F^-(v)=\{e\in F(v) | e^-=v\}$.
    
    \item A \textbf{bit labeling} is a map $\alpha:V\rightarrow \{0,1\}$. For an element $\sigma\in \{0,1\}$, we denote \textbf{the set of vertices labeled $\sigma$} as $V_\sigma\equiv \{v\in V|\alpha(v)=\sigma\}$. 
    When $V$ are the vertices of an oriented graph $(V,E)$, we also define $E_1 = \{e \in E | e \subseteq  V_1\times V_1 \} $ and $E_0 = E\setminus E_1 $.
    
    \item A \textbf{planar graph} is a graph which can be embedded in a plane without edges intersecting. When we speak of planar graphs, we implicitly assume that a specific embedding has already been chosen.
\end{itemize}

\subsection{Quantum many body systems} \label{subsec: QMBS}

We now define the two fundamental objects of interest that we deal with: a quantum many body system with or without constraints. An unconstrained quantum many body system is a labeled collection of Hilbert spaces. A constrained quantum many body system additionally involves considering only a subspace of the combined Hilbert space to be physical. These two types of systems are of interest because they describe states of condensed matter systems, lattice ultraviolet completions of quantum field theories, and quantum error correcting codes.
% The addition of constraints allows for the non-uniqueness of physical operators, a property which should be familiar to readers who have studied either quantum error correction or gauge symmetry. 
For brevity we replace the phrase ``quantum many body system'' with simply ``system'' for the rest of this work.\\
\begin{definition}
An \textbf{unconstrained system} $\mathcal{S}=(V,\{\mathcal{H}_{v\in V}\})$ consists of a set $V$ labeling individual quantum systems living in Hilbert spaces $\{\mathcal{H}_v\}$. 
For a subset $W\subseteq V$, we define the \textbf{algebra of $\mathcal{S}$-operators localized to $W$} as 
\begin{align*}
    \mathcal{A}_\mathcal{S}(W)=\linearops{\mathcal{H}_W}\otimes I_{V\setminus W}\subseteq \linearops{\mathcal{H}_V}.
\end{align*}
When $W=V$ we simply say that $\mathcal{A}_S(V)$ is the \textbf{algebra of $\mathcal{S}$-operators.} 
\end{definition}

\begin{definition}
A \textbf{constrained system} $\mathcal{T}=(\mathcal{S},\Pi)$ consists of an unconstrained system $\mathcal{S}=(V,\{\mathcal{H}_{v\in V}\})$ and a projector $\Pi:\mathcal{H}_V\rightarrow \mathcal{H}_V$. For any $W\subseteq V$, we define the \textbf{algebra of physical $\mathcal{T}$-operators localized to $W$} as
\begin{align}
    \mathcal{A}_\mathcal{T}(W)=\{O\in \mathcal{A}_\mathcal{S}(W)|[O,\Pi]=0\}. \label{eq:localized-physical-operator}
\end{align}
When $W=V$ we simply say that $\mathcal{A}_\mathcal{T}(V)$ is the \textbf{algebra of physical $\mathcal{T}$-operators}. We also call $\Pi\mathcal{H}_V$ the \textbf{physical Hilbert space of $\mathcal{T}$}.
\end{definition}
We emphasize that this definition of physical is relative to the system being considered.
In this work, we consider error-correcting codes in which the logical system is constrained in some way, meaning we can sensibly talk about a ``physical logical operator'', which may sound strange to readers familiar with error-correction.
In our context, this is just a logical operator which is physical with respect to the logical system, meaning it commutes with the relevant projector (i.e.\ preserves the constraint).

Note that a constrained system is a generalization of an unconstrained system. Note also that, unlike for an unconstrained system, for a constrained system $\mathcal{T}$ with non-trivial projector we can find two distinct elements in the algebra of physical $\mathcal{T}$-operators that have the same action on states in the physical Hilbert space of $\mathcal{T}$. This is a key property which the reader may recognize from either quantum error correction or gauge symmetry. In a quantum error correction context, such a non-uniqueness appears because information which is stored in the logical Hilbert space may be manipulated by acting on any set of physical systems from which it may be recovered. In the context of gauge symmetry, it appears because a physical operator may be multiplied (on either side) by an operator implementing a gauge transformation, and its action on physical states would remain the same.

\subsection{Global symmetries \&  global/global dualities}
\label{subsec: global-sym}
In this subsection we define what it means for a system to have a \textit{global symmetry}, and what it means for an isometry between two systems to exhibit a \textit{global/global duality}.

One of the most fundamental features unconstrained\footnote{One could also consider constrained systems with global symmetries; however we stick with unconstrained systems for simplicity and because it is sufficient for our discussions on holographic codes.} systems may exhibit is an internal\footnote{This is in contrast with symmetries that do allow information to spread locally, such as spacetime symmetries. We do not consider such symmetries in this paper, again for simplicity.}, also known as ``on-site'', global symmetry. This usually means that 1) there exists a unitary representation of a group acting on the whole system which is a product of unitary representations acting on individual subsystems, and 2) all elements of this ``collective'' representation commute with the Hamiltonian. We do not consider systems with dynamics in this work, so we do not require 2). With this in mind we now provide a formal definition:

\begin{definition}\label{def:global-symmetry-transformation}
An \textbf{unconstrained system transforming under a global symmetry}, $(\mathcal{S}, G, \{U_v\})$ consists of a system $\mathcal{S}=(V, \{\mathcal{H}_v\})$, a group $G$, and a set of unitary representations of $G$
\begin{align*}
    U_v:G\rightarrow \mathcal{A}_\mathcal{S}(v)
\end{align*}
where we drop the set brackets around $v$ for brevity. For any $W\subseteq V$ we define $U_W(g)\equiv\prod_{w\in W}U_w(g)$. We refer to $U_W(g)$ as a \textbf{global symmetry transformation restricted to $W$}, and when $W=V$ we simply refer to it as a \textbf{global symmetry transformation}. 
\end{definition}

To genuinely be a global symmetry, we should also require that the representations acting on each site be faithful. This is not necessary for our results to apply except when we need to use the results of \cite{harlow-ooguri} in \cref{subsec:connection-to-harlow-ooguri}, since they explicitly make use of this.  

In this paper we often examine isometries between the Hilbert spaces of two systems. When these two systems are unconstrained and transform under a global symmetry with the same group, then we may inquire whether the isometry possesses a \textit{global/global duality} -- that is, whether the isometry is a group isomorphism between the two sets of global symmetry transformations acting on the two different systems. This can be more explicitly stated as follows.

\begin{definition} \label{def:global-global}
For a group $G$ consider two systems of the following form
\begin{enumerate}
    \item A system transforming under a global symmetry $(\mathcal{S},G,\{U_v\})$ with $\mathcal{S}=(V,\{\mathcal{H}_{v\in V}\})$
    \item A system transforming under a global symmetry $(\tilde{\mathcal{S}},G,\{\tilde{U}_{\tilde{v}}\})$ with $\tilde{\mathcal{S}}=(\tilde{V},\{\tilde{\mathcal{H}}_{\tilde{v}\in\tilde{V}}\})$
\end{enumerate}
%Assume without loss of generality that $dim(\mathcal{H}_V)\leq dim(\tilde{\mathcal{H}}_{\tilde{V}})$.
Consider an isometry $\mathcal{V}:\mathcal{H}_V\rightarrow\tilde{\mathcal{H}}_{\tilde{V}}$. We say that $\mathcal{V}$ has a \textbf{global/global symmetry duality} if for all $g\in G$ the global symmetry transformation $\tilde{U}_{\tilde{V}}(g)$ acting on $\tilde{\mathcal{S}}$  implements the global symmetry transformation $U_{V}(g)$ acting on $\mathcal{S}$, i.e.
\begin{align*}
    \forall g\in G\quad \mathcal{V}U_V(g)&=\tilde{U}_{\tilde{V}}(g)\mathcal{V}.
\end{align*}
\end{definition}

Notice that $\tilde{U}_{\tilde{V}}$ automatically
preserves the image of $\mathcal{V}$, i.e. $ [\tilde{U}_{\tilde{V}}(g),\mathcal{V}\mathcal{V}^\dagger]=0$.
The object $\mathcal{V}$ is also sometimes called an ``intertwiner'' or a ``covariant isometry''.

\subsection{Gauge symmetries}
\label{subsec: gauge-sym}
We now review the standard construction of a lattice gauge theory originally introduced in Ref.~\cite{Kogut:1974ag}, but stripped of all dynamics. In particular, we describe how to build a constrained system ``with gauge symmetry''. We begin by introducing some basic definitions in \ref{subsubsec: gauge-sym-basic-definitions}, and then the formal definition in \ref{subsubsec: gauge-sym-definition}. In \ref{subsubsec: gauge-sym-fixed-flux-sectors} we show how the ``gauge-invariant Hilbert space'' can be written as a direct sum of ``fixed-flux sectors''. Finally in \ref{subsubsec: gauge-global-dualities} we define what it means for an isometry to have ``gauge/global symmetry duality''.

Gauge symmetries are different in nature than global symmetries. A global symmetry involves a unitary representation acting on a system, generically changing its state. In contrast, a system exhibits gauge symmetry if it is constrained to be invariant under the action of local unitary representations. Such a system can be constructed by adding new degrees of freedom to an unconstrained system and restricting the resulting Hilbert space (thus making a constrained system) to be invariant under the action of ``gauge transformations'', defined in depth below. It is not necessary to include gauge transformations at all locations in this restriction, a fact which is of vital importance to our results, and which is encoded in an assignment of a bit to each subsystem of the original system, i.e. a ``bit labeling''. \\

\subsubsection{Basic definitions}
\label{subsubsec: gauge-sym-basic-definitions}

One of the ingredients of a system with gauge symmetry is the presence of degrees of freedom which encode group elements. More specifically, these degrees of freedom live in the Hilbert space $L^2(G)$ for some group $G$. For finite groups this is a Hilbert space with $\{|g\rangle|g\in G\}$ as an orthonormal basis. For continuous compact Lie groups it is essentially the same along with the usual subtleties that we address in \cref{subsubsec: continuous-groups}. A natural representation of $G$ which acts on $L^2(G)$ is ``left-multiplication'', $U^L:G\rightarrow \linearops{L^2(G)}$\footnote{In the literature $U^L$ and $U^R$ are more commonly denoted as $L$ and $R$. We deviate from this notation as we use other meanings for both $L$ and $R$.}, which acts on the basis vector $\ket{h}$ as
\begin{align} \label{eq:UL}
    U^L(g)\ket{h}=\ket{gh}
\end{align}
Another representation is ``right-multiplication'', $U^R:G\rightarrow \linearops{L^2(G)}$, which acts on the basis vector $\ket{h}$ as
\begin{align}\label{eq:UR}
    U^R(g)\ket{h}=\ket{hg^{-1}}.
\end{align}
It is easy to check that these representations commute with each other, i.e.
\begin{align*}
    [U^R(g),U^L(h)]=0
\end{align*}
for any $g,h\in G$.

We often need to take the averaged sum over all group elements of an expression. We denote this as $\int dg$, which for finite groups means $\frac{1}{|G|}\sum_{g\in G}$, and for continuous compact Lie  groups means integration using the Haar measure\footnote{More explicitly, we use the unique left and right invariant measure normalized so that $\int_{G}1\dd{g}= 1 $.}. In the remainder of the paper we restrict to finite groups and continuous groups that are also compact Lie groups\footnote{Note that a finite group is trivially a compact Lie group as it has the discrete topology so every map is trivially smooth and the compactness follows from the finiteness of the group.}. Therefore, when we say continuous group we really mean continuous compact Lie group.

\subsubsection{Systems with gauge symmetry}
\label{subsubsec: gauge-sym-definition}

The definition of gauge symmetry proceeds as follows.
Starting with an ``ungauged'' system, we first introduce new gauge degrees of freedom to obtain a ``pre-gauging'' system.
On this system, we can then define gauge transformations, which couple the gauge degrees of freedom with the original degrees of freedom.
The gauge-symmetric theory is then obtained by constraining the pre-gauging system into a subspace invariant under such gauge transformations.
We formalize this as follows.

\begin{definition}\label{def:pre-gauging}
Given a system transforming under a global symmetry, $(\mathcal{S}, G, \{U_v\})$, with $\mathcal{S}=(V, \{\mathcal{H}_v\})$, referred to here as the \textbf{ungauged system}, and an oriented graph $\Lambda=(V,E)$, we define the \textbf{pre-gauging system} to be the system
\begin{align*}
\mathcal{S}'\equiv(V\cup E,\{\mathcal{H}_{v\in V}\}\cup\{\mathcal{H}_{e\in E}\})  
\end{align*}
with $\mathcal{H}_{v\in V}$ the same for the ungauged and pre-gauged systems. We also require $\mathcal{H}_{e\in E}\cong L^2(G)$, which we sometimes describe as ``gauge degrees of freedom'' that ``live on the edges''.
%We will use $U_v(g)$ to mean $U_v(g)\otimes I_{E}$ when using it as an $\mathcal{S}'$ operator.
% We refer to $\mathcal{H}_V$ as the \textbf{ungauged Hilbert space}, and operators acting on it as \textbf{ungauged operators}.
% We also refer to $$\mathcal{H}_\Lambda\equiv\mathcal{H}_{V\cup E}$$ as the \textbf{pre-gauging Hilbert space}. 
%Furthermore, for any subgraph $\Gamma=(\tilde{V},\tilde{E})$ of $\Lambda$ we define $\mathcal{A}_{\mathcal{S}'}(\Gamma)\equiv\mathcal{A}_{\mathcal{S}'}(\tilde{V}\cup\tilde{E})$.
\end{definition}

\begin{definition}
\label{def: gauge-trans-operator}
Given a system transforming under a global symmetry, $(\mathcal{S}, G, \{U_v\})$, with $\mathcal{S}=(V, \{\mathcal{H}_v\})$ and an oriented graph $\Lambda=(V,E)$, let $\mathcal{S}'$ be the pre-gauged system. For a vertex $v\in V$ we define a \textbf{gauge transformation (operator) localized to $v$} as an operator of the form

\begin{align} \label{eq:Avg}
    A_v(g)=U_v(g)\prod_{e\in E^+(v)}U^R_e(g)\prod_{e\in E^-(v)}U^L_e(g),
\end{align}
with $U^R$ and $U^L$ defined as in \cref{eq:UL,eq:UR}.
Note that $A_v(g)\in\mathcal{A}_{\mathcal{S}'}(\{v\}\cup E(v))$ is a unitary representation of $G$. The Hilbert space invariant under (or stabilized by) gauge transformations localized to $v$ is the image of the projector, 
\begin{align*}
    \Pi_v\equiv \int dg A_v(g)
\end{align*}
called the \textbf{local gauge projector}.
\end{definition}

\begin{definition}
\label{def: gauging-definition}
Given a system transforming under a global symmetry, $(\mathcal{S}, G, \{U_v\})$, with $\mathcal{S}=(V, \{\mathcal{H}_v\})$, and an oriented graph $\Lambda=(V,E)$ with bit labeling $\alpha:V\rightarrow\{0,1\}$, let $\mathcal{S}'$ be the pre-gauged system. We define the system obtained by \emph{gauging} $(\mathcal{S}, G, \{U_v\})$ to be the constrained system $\mathcal{T}=(\mathcal{S}',\Pi_{GI})$ with
\begin{align*}
    \Pi_{GI} \equiv \prod_{v\in V}(\Pi_{v})^{\alpha(v)}=\prod_{v\in V_{1}}\Pi_{v},
\end{align*}
where a projector to the power of zero is understood to mean the identity and as above, $V_{\beta} = \{ v\in V | \alpha(v) = \beta \}$ with $\beta \in \{0,1\}$.
We call $\Pi_{GI}\mathcal{H}_\Lambda$ the \textbf{gauge-invariant Hilbert space}, and operators in $\mathcal{A}_{\mathcal{T}}(\Lambda)$, \textbf{gauge-invariant operators}.
We use the convention that edges connecting vertices in $V_0$ to those in $V_1$ always point from the former to the latter. 
\end{definition}

The label $\alpha(v)$ determines whether or not the Hilbert space is restricted to be invariant under a gauge transformation localized to a vertex $v$. 
We sometimes refer to the vertices in $V_0$ as \emph{NGC} (not gauge-constrained) vertices, and those in $V_1$ as \emph{GC} (gauge-constrained) vertices.
Because gauge transformations localized to NGC vertices are not constrained, they continue to have a non-trivial action on the gauge-invariant subspace, which is a key feature for the rest of the paper.
When we introduce holographic codes, NGC vertices are associated with the Boundary system and GC vertices with the Bulk system.

\begin{definition} \label{def:asymptotic-transformation}
Given a system transforming under a global symmetry, $(V, \{\mathcal{H}_v\}, G, \{U_v\})$ gauged with respect to an oriented graph $\Lambda=(V,E)$ with bit labeling $\alpha$, for any $W\subseteq V_0$ we define the \textbf{NGC} or \textbf{asymptotic symmetry transformation restricted to $W$} as
\begin{align*}
    A_{W}(g)\equiv\prod_{v \in W} A_v(g)
\end{align*}
For the special case $W=V_0$, we refer to $A_{V_0}(g)$ simply as an
\textbf{NGC} or \textbf{asymptotic symmetry transformation}.
\end{definition}

It is straightforward to show that NGC symmetry transformations are gauge-invariant. We do not include vertices in $V_1$ in this transformation. Doing so would make no difference to how the NGC symmetry transformation would act on states in the gauge-invariant Hilbert space, though it would change which subsystems it acts on for states in the pre-gauged Hilbert space. In particular, the $\mathcal{T}$-operator $\prod_{v\in V} A_v(g)$ has the same action as the NGC symmetry transformation $A_{V_0}(g)$.

We use the term \emph{asymptotic symmetry transformation} in holographic contexts because it is used in Ref.~\cite{harlow-ooguri} in the specific case where the system has a geometric structure and the $V_0$ vertices correspond to vertices on the Boundary.
For quantum field theories these transformations take place at asymptotically spacial infinity, hence the name. 
In more application-agnostic settings, such as here and in \cref{sec:generalized-gauging-map}, we use the term NGC simply to emphasize the generality of the construction.

\subsubsection{Fixed-flux sectors}
\label{subsubsec: gauge-sym-fixed-flux-sectors}

We can decompose the gauge-invariant Hilbert space into a direct sum of \textit{fixed-flux sectors}, or simultaneous eigenspaces of all ``NGC-to-NGC'' Wilson lines and all Wilson loops, defined below. This decomposition is useful in understanding the image of the gauging isometry.

The name ``fixed-flux sector'' is inspired by quantum electrodynamics, which can be described as a continuum version of a gauged system where $G=U(1)$. In this context, the Wilson loop operators measure the magnetic flux going through them, and thus a ``fixed-flux sector'' is the portion of the state space for which all magnetic fluxes have fixed given value.

\begin{definition} \label{def:flux-free}
Consider a system transforming under a global symmetry, $(\mathcal{S}, G, \{U_v\})$, with $\mathcal{S}=(V, \{\mathcal{H}_v\})$ and a graph $\Lambda=(V,E)$ with bit labeling $\alpha$.
Let $\mathcal{S}'$ be the pre-gauged system, and $\mathcal{T}=(\mathcal{S}',\Pi_{GI})$ be the gauged system.
Let $r:G\rightarrow\mathds{C}_{d\times d}$ with $d$ some positive integer be a faithful unitary representation of $G$, and let $r_{ij}(g)$ denote its components in some basis.
Define the \textbf{Wilson link operator} acting on edge $e$, $W^{e}_{ij}\in\mathcal{A}_{\mathcal{S}'}(e)$, as
\begin{align*}
    W^e_{ij}|g\rangle_e%\otimes|\psi\rangle_{\Lambda\setminus e}
    =r_{ij}(g)|g\rangle_e. %\otimes|\psi\rangle_{\Lambda\setminus e}
\end{align*}
%for all $|\psi\rangle_{\Lambda\setminus e}$.  
For $e=(u,v)\in V\times V$  define $\bar{e}\equiv(v,u)$%. For $e\in E$ define
, and $W^{\bar{e}}_{ij}\equiv (W^{e}_{ji})^\dagger$. To each cycle $l=(v_1,\ldots ,v_n)$ such that $v_1,\ldots ,v_n\in V_1$, we associate the \textbf{Wilson loop operator} 
    $$
    W^l\equiv Tr({W^{\overline{(v_1,v_2)}}}\ldots {W^{\overline{(v_{n-1},v_{n})}}} {W^{\overline{(v_n,v_{1})}}}), \footnote{It is possible to remove the overlines by switching left and right multiplication in the definition of gauge transformations. We choose this convention instead to maintain consistency with \cite{Cirac_2015}.}
    $$
where both the multiplication of operators on the right hand side and the trace are understood to be with respect to the representation indices.

For each undirected path $p=(u,v_1,\ldots ,v_n,w)$ with $u,w\in V_0$ and $v_1,\ldots ,v_n\in V_1$ we define the $\textbf{NGC-to-NGC Wilson line operator}$

\begin{align*}
    W^p_{ij}\equiv (W^{\overline{(u,v_1)}}\ldots W^{\overline{(v_{n-1},v_{n})}} W^{\overline{(v_n,w)}})_{ij}.
\end{align*}

All Wilson loop and NGC-to-NGC Wilson line operators commute with each other and are gauge-invariant (Shown in the appendices under \cref{app:ProofWilsonLoopsI}), and thus their eigenspaces can be used to partition the gauge-invariant Hilbert space. We call a subspace of $\Pi_{GI}\mathcal{H}_\Lambda$ which is a simultaneous eigenspace of $W^l$ and $W^p_{ij}$ for all loops $l$ and NGC-to-NGC paths $p$, for all choices of $r$, a \textbf{fixed-flux sector}. 
There is a special sector for which it is sufficient to use a single representation $r$, called  the \textbf{flux-free sector}, which has eigenvalue $Tr(r(I))$ for all $W^l$ and $r_{ij}(I)$ for all $W^p_{ij}$. We denote its projector by $\Pi_{FF}$.
Because these lines and loops measure only edge degrees of freedom, $\mathcal{H}_E$, this projector can also be expressed in the form $\qty( \Pi_{FF}^E \otimes \id_V ) \Pi_{GI}$, with $\Pi_{FF}^E$ projecting onto the eigenspaces described above.
\end{definition}

To understand this definition, notice that $W^e$ essentially measures the group element at site $e$. Thus an NGC-to-NGC Wilson line measures the product of group elements along its path.
The reason that we only consider Wilson lines ending on NGC vertices at both ends is that only these Wilson lines are gauge invariant on their own. 
As we discuss below Wilson lines ending at GC vertices can be gauge-invariant when in combination with vertex operators.

The Wilson loop is a bit more subtle. We can not say that it measures the product of group elements, $g$, along the loop because $g$ would depend on which link we started with. The Wilson loop actually measures the conjugacy class of $g$, defined by $C(g)=\{hgh^{-1}|h\in G\}$. 
Using the cyclic property of the trace, it is clear that $\Tr[r(g)]$ depends only on $C(g)$. 
Character theory ensures that the values of $\Tr[r(g)]$ for all irreducible representations $r$ uniquely determines conjugacy class (see e.g.\ Ref.~\cite{martin_1963}). However, it is sufficient to use a single faithful (finite-dimensional) unitary representation to distinguish the conjugacy class $C(I)$ from the others. This, and the fact that the identity is the only member of its conjugacy class, shows that the flux-free sector as defined is the unique subspace of $\Pi_{GI}\mathcal{H}_\Lambda$ for which the product of group elements along any path or loop is the identity.

What degrees of freedom remain within each fixed-flux sector? Even if all elements of $\{\mathcal{H}_{v\in V}\}$ are trivial, there may still remain some pure gauge degrees of freedom (see Refs.~\cite{Cui_2020} and \cite{Sengupta}). As we uncover in greater detail in \cref{sec:generalized-gauging-map}, if some of the $\{\mathcal{H}_{v\in V}\}$ are non-trivial, then they contribute additional degrees of freedom to each flux sector. These are completely fixed by Wilson lines that, on at least one side, end with an action on a subsystem corresponding to an element of $V_1$. We call these ``NGC-to-GC'' (when the other end is in $V_0$, i.e.\ an NGC vertex) and ``GC-to-GC'' (when both are in $V_1$) Wilson lines.

We note that, alongside Wilson loops, NGC-to-NGC Wilson lines, and NGC-to-charge Wilson lines, there is another gauge invariant operator $U_e^L(g)=U_e^R(g)$ for $g$ in the center of $G$ and any $e$ is also gauge invariant. We refer to these as $\textbf{central flux operators}$. Note that left and right multiplication are equivalent only for central elements, and that only for central elements are they gauge-invariant.

\subsubsection{Gauge/global dualities}
\label{subsubsec: gauge-global-dualities}
We now introduce a definition for an isometry with a \textit{gauge/global duality}. The definition is motivated by AdS/CFT, where global symmetry transformations acting on the Boundary CFT implement asymptotic gauge transformations in the Bulk AdS \cite{harlow-ooguri}. 

\begin{definition}\label{def:gauge-global}
Consider the following two systems involving the group $G$:
\begin{enumerate}
    \item A constrained system $\mathcal{T}=(\mathcal{S}',\Pi_{GI})$ which is obtained by gauging a system with global symmetry $(\mathcal{S},G,\{U_v\})$, where $\mathcal{S}=(V,\{\mathcal{H}_{v\in V}\})$,  with respect to an oriented graph $\Lambda=(V,E)$ with bit labeling $\alpha:V\rightarrow\{0,1\}$.
    \item A system transforming under a global symmetry $(\tilde{\mathcal{S}},G,\{\tilde{U}_{\tilde{v}}\})$ with $\tilde{\mathcal{S}}=(\tilde{V},\{\tilde{\mathcal{H}}_{\tilde{v}\in \tilde{V}}\})$.
\end{enumerate}

Consider an isometry $\mathcal{V}:\Pi_{GI}\mathcal{H}_\Lambda\rightarrow\tilde{\mathcal{H}}_{\tilde{V}}$. We say that $\mathcal{V}$ has a \textbf{gauge/global symmetry duality} if for all $g\in G$ the global symmetry transformation $\tilde{U}_{\tilde{V}}(g)$ on $\tilde{\mathcal{S}}$ implements the NGC symmetry transformation $A_{V_0}(g)$ on $\mathcal{T}$, i.e.
\begin{align*}
   \forall g \in G, \quad \mathcal{V}A_{V_0}(g)&=\tilde{U}_{\tilde{V}}(g)\mathcal{V}.
\end{align*}
If instead $\mathcal{V}$ is an isometry of the form $\mathcal{V}:\tilde{\mathcal{H}}_{\tilde{V}}\rightarrow\Pi_{GI}\mathcal{H}_\Lambda$ then we say it has a \textbf{global/gauge symmetry duality} if for all $g\in G$ the global symmetry transformation $\tilde{U}_{\tilde{V}}(g)$ on $\tilde{\mathcal{S}}$ is implemented by the NGC symmetry transformation $A_{V_0}(g)$ on $\mathcal{T}$, i.e.
\begin{align*}
    \forall g \in G, \quad  \mathcal{V}\tilde{U}_{\tilde{V}}(g)&=A_{V_0}(g)\mathcal{V}.
\end{align*}
Lastly, if instead we replace $\Pi_{GI}\mathcal{H}_\Lambda$ with one of its fixed-flux sectors, we instead say that $\mathcal{V}$ exhibits a \textbf{fixed sector gauge/global (or global/gauge) symmetry duality}. 
\end{definition}

Again we automatically have that
 $[\tilde{U}_{\tilde{V}}(g),\mathcal{V}\mathcal{V}^\dagger] =0$ for the gauge/global case and 
  $[A_{V_0}(g),\mathcal{V}\mathcal{V}^\dagger] =0$ for the global/gauge case. 

\section{The gauging isometry}
\label{sec:generalized-gauging-map}

The procedure for defining a system with gauge symmetry, as we have presented it in Sec.~\ref{subsec: gauge-sym}, appears at first glance an involved modification of a system with global symmetry. There is, however, an explicit relationship between these two systems in the form of a linear map between their two overall Hilbert spaces. The properties of this map in the context of lattice spin systems have been studied in various guises since the work of Refs.~\cite{Kramers1941,wegner1971duality}. We follow the recent treatment in Ref.~\cite{Cirac_2015}, which focused on the case where gauge constraints are imposed over all the local gauge transformations, i.e. $V_0=\varnothing$.
 In this case, the map has as its support the space of states that transform trivially under the global symmetry, where it acts as an isometry.

In the Bulk of AdS/CFT however, the gauge constraint is not enforced at the asymptotic boundary (see Ref.~\cite{harlow-ooguri}). To model additional features of AdS/CFT we would like to study holographic codes that do likewise. For this reason we now generalize this map by allowing for some local gauge transformations to be excluded from the gauge constraint. We then prove the relevant properties of this map, namely that it can be an isometry on the whole system, that it preserves the locality of operators up to a ``dressing to the boundary'', that it displays a global/gauge duality, and that its image is the full flux-free sector. We then describe a generalization of this map to other flux sectors. 

\subsection{Definition of the map}
We now define a map that ``gauges'' states of a system transforming under a global symmetry by encoding them into gauge-invariant states. Our definition differs from the one in Ref.~\cite{Cirac_2015} only in that it allows for the possibility that some gauge transformations are not included in the gauge constraint. 

\begin{definition}
\label{def: gauging-map}
Consider an unconstrained system transforming under a  global symmetry, $(\mathcal{S}, G, \{U_v\})$  with $\mathcal{S}=(V, \{\mathcal{H}_v\})$ and an oriented graph $\Lambda=(V,E)$ with $\alpha$ labeling. Let $\mathcal{T}=(\mathcal{S}',\Pi_{GI})$ be the constrained system obtained by gauging $\mathcal{S}$. Then we define the \textbf{gauging map} $\mathcal{G}$ as
\begin{equation}
    \begin{split}
        \mathcal{G}:\mathcal{H}_V&\rightarrow\Pi_{GI}\mathcal{H}_{\Lambda} \\
    |\psi\rangle_V&\mapsto \Pi_{GI}\left (|\psi\rangle_V\bigotimes_{e\in E}|I\rangle_e\right )
    \end{split}
\end{equation}
where $I$ is the identity element of $G$. 
\end{definition}

There is a subtle terminology distinction that should be clarified, as it may cause confusion.
The phrase ``the system obtained by gauging $(\mathcal{S}, G, \{U_v\})$'', as defined in \cref{def: gauging-definition}, refers only to a system constructed from an ungauged system, with appropriate edge degrees of freedom added and a gauge-invariance constraint imposed.
It is not to be confused with the application of the gauging map $\mathcal{G}$ to an ungauged system, because its image, as we explain shortly is additionally constrained beyond the gauge-invariance requirement.
The former is a prescription for constructing a Hilbert space, whereas the latter is a specific isometry that embeds the ungauged system into the former.
This potential cause for confusion is an unfortunate consequence of the former phrase being standard in the literature.

\subsection{The map can be an isometry}

We now show that the gauging map can be made an isometry when the graph is connected and some gauge transformations are not included in the gauge constraint.

\begin{theorem}
\label{thm: gauging-map-is-isometry}
If $V_0\neq\varnothing$ and $\Lambda$ is connected, then $\mathcal{G}$ is an isometry up to an overall normalization factor.
\end{theorem}

\begin{proof}
Consider the map $\Pmc=\mathcal{G}^{\dagger}\mathcal{G}^{}: \mathcal{H}_{V} \rightarrow \mathcal{H}_{V}$ given by 

\begin{align*}
    \Pmc &= \bra{\{I_e\}}_{E} \left(\Pi_{GI}\right)^{\dagger}\Pi_{GI} \ket{\{I_e\}}_{E} = \bra{\{I_e\}}_{E} \Pi_{GI} \ket{\{I_e\}}_{E} \\
    &=\int\prod_{v\in V_1}\dd{g_v} \prod_{u\in V_1}U_{u}(g_u) \prod_{e\in E} \mel{I}{
    U_{e}^{L}(g_{e^{-}}^{})
    U_{e}^{R}(g_{e^{+}}^{})}{I}_e  \\
    &=\int\prod_{v\in V_1}\dd{g_v} \prod_{u\in V_1}U_{u}(g_u) \prod_{e\in E} \braket{I}{g_{e^{-}}^{}g_{e^{+}}^{-1}}_e.   
\end{align*}
where in the last line if $e^{+}\in V_0$ we have $g_{e^{+}} = I$ since the gauge transformation acting on $e^{+}$ is not included in the gauge constraint. The last factor is non-zero only when $g_{e^{-}}^{}=g_{e^{+}}^{}$. When $V_0 \neq \varnothing$ and $\Lambda$ is connected, this causes all $g$'s to be set to $I$\footnote{In the case when $V_0=\varnothing$ then all $g$'s are equal to each other but there remains a sum over all elements of the group, making the map a projector onto the trivial representation. 
}%
.

Thus $\mathcal{P} = \kappa^{2}\id_{\mathcal{H}_V}$ with 
\begin{align*}
    \kappa^{2} &= \int\prod_{v\in V_1}\dd{g_v} \prod_{e\in E}\braket{I}{g_{e^{-}}^{}g_{e^{+}}^{-1}}.
\end{align*}
In particular, for $G$ finite the above formula is   
\[ \kappa^2 = \abs{G}^{-\abs{V_1}}.\]
Then one could make $\Gmc$ an isometry by dividing by $\kappa$.
\end{proof}

In \cref{sec:Image-Isometry-Is-Flux-Free}, we find that the image of this map is precisely the flux-free sector introduced in \cref{def:flux-free}.
We can use that here result to provide a heuristic degrees-of-freedom-counting argument for why the existence of at least one NGC vertex makes allows the map become an isometry.
When there are no NGC vertices (i.e.\ the setting studied in Ref.~\cite{Cirac_2015}), the gauging map enforces too many constraints, causing all states other than the symmetric subspace of the input system to be annihilated.
Although the gauging map provides additional degrees of freedom in the form of edge states, these are counteracted by the gauge constraints and the flux-free Wilson loop constraints.
If a GC vertex is changed to be an NGC vertex, then these constraints are replaced by the trivial NGC-to-NGC Wilson line constraints, but only so long as there is already at least one other NGC vertex, otherwise no such lines can exist.
Thus the first NGC vertex added provides additional degrees of freedom not counteracted by constraints, allowing the map to no longer annihilate any states.

\subsection{Local operators are implemented by their dressed versions}
\label{subsec: dressing-local-operators}

We now show that local ungauged operators are implemented by ``dressed'' gauged operators. The dressed operator may be chosen to have support on any path from the original site to a point on the boundary, as portrayed in \cref{fig:dressing-local-ops}. The non-uniqueness can be understood in analogy with error correcting codes\footnote{In general, our discussions of error-correcting codes throughout this work are in the context of correcting for erasure errors.}, where a logical operation can be implemented on any subset of physical degrees of freedom on which the logical information can be recovered. This is true for isometries in general, even if they are not good at correcting errors. 

\begin{figure}[t]
  \centering
    \includegraphics[width=1.0\textwidth]{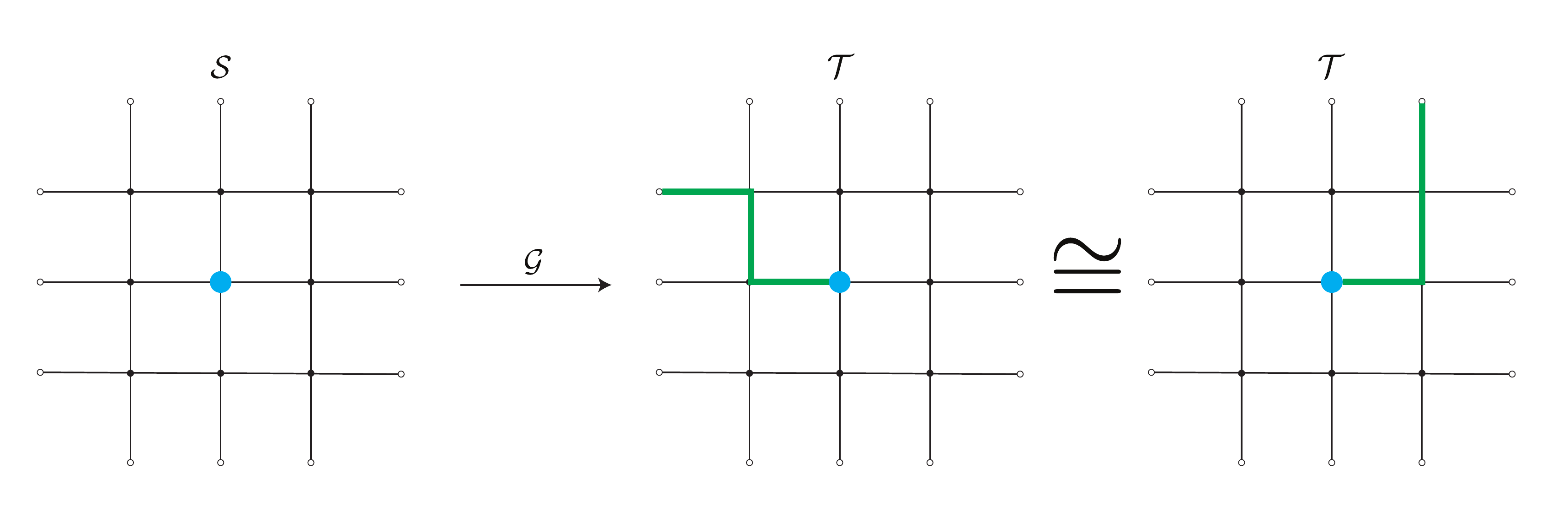}
    \caption{Mapping of an $\mathcal{S}$ operator under $\mathcal{G}$. The black and white dots are vertices labeled $1$ and $0$ respectively. The blue dot represents an operator acting on a vertex $v\in V_1$, while the green line represents an operator acting on all the edges it lies on. The localized $\mathcal{S}$ operator can be implemented on $\mathcal{T}$ via an operator supported on $v$, as well as all the edges lying along a path from $v$ to any $v'\in V_0$.}
  \label{fig:dressing-local-ops}
\end{figure}

\newcommand{\subgraphGaugingLemma}{%
Consider a system transforming under a global symmetry, $(\mathcal{S}, G, \{U_v\})$ with $\mathcal{S}=(V, \{\mathcal{H}_v\})$ and an oriented graph $\Lambda=(V,E)$ with $\alpha$ labeling. Let $\mathcal{S}'$ be the pre-gauged system and $\mathcal{T}$ the gauged system. Consider a subgraph of $\Lambda$, $\Gamma=(V_\Gamma,E_\Gamma)$, such that $E_{\Gamma}^{+}\cup E_{\Gamma}^{-}\subseteq V_{\Gamma}\cup V_{0} $, i.e.\ all the vertices of the edges on $E_{\Gamma}$ are in $V_0$ or $V_{\Gamma}$. Define, for any $v\in V_\Gamma$ and $g\in G$, the following $\mathcal{S}'$ operator:
\begin{align*}
    A^\Gamma_{v}(g)=U_v(g) \prod_{e\in E^+(v)}U^R_e(g)\prod_{e\in E^-(v)}U^L_e(g).
\end{align*}
For any $Q_{\Gamma}\in \mathcal{A}_{\mathcal{S}'}(\Gamma)\equiv\mathcal{A}_{\mathcal{S}'}(V_\Gamma\cup E_\Gamma)$, the operator

\begin{align*}
O_\Gamma= {\DgvG} \quad \left (\prod_{v\in V_\Gamma}A^\Gamma_v(g_{v})\right)
Q_{\Gamma}
\left (\prod_{v\in V_\Gamma}A^\Gamma_v(g_{v})\right)^\dagger    
\end{align*}
is an element of $\mathcal{A}_\mathcal{T}(\Gamma)$.
}

\begin{lemma}
\label{lemma:subgraph-gauging-lemma}
\subgraphGaugingLemma
\end{lemma}

\begin{proof} The proof is given in \cref{proof:subgraph-gauging-lemma}. 
\end{proof}

We now prove our claim that a local ungauged operator supported on some vertex $u\in V_1$ can be implemented by a gauged operator with the same support plus an arbitrary dressing to the boundary.
Specifically, this dressing lies on a subgraph $\Gamma = (V_\Gamma, E_\Gamma)$ consisting of a path from $u$ to $v_0 \in V_0$ and its edges. However, it does not have support on $v_0$ itself.

\begin{theorem}
\label{thm: operator-dressing-theorem}
Consider a system transforming under a global symmetry, $(\mathcal{S}, G, \{U_v\})$  with $\mathcal{S}=(V, \{\mathcal{H}_v\})$ with $\Lambda=(V,E)$ connected and labeled by $\alpha$ such that $V_0 \neq \varnothing $. Let $\mathcal{T}$ be obtained from gauging $\mathcal{S}$. Let $\mathcal{G}$ be the corresponding gauging isometry.
Consider $u \in V$ and a subgraph of $\Lambda$, $\Gamma=(V_\Gamma,E_\Gamma)$ with $V_\Gamma \subseteq  V_1$, such that the edges $E_\Gamma$ and the vertices $V_\Gamma \cup \{v_0\}$ form a path starting at $u$ and ending with a vertex $v_0 \in V_0$.
For any $O_{u}\in \mathcal{A}_{\mathcal{S}}(u)$, there exists an operator $O_\Gamma\in \mathcal{A}_\mathcal{T}(\Gamma)$ called a \textbf{dressed operator} that preserves the image of the gauging map and implements $O_u$, i.e.\ 
\begin{align*}
     \mathcal{G}O_{u}&=O_\Gamma\mathcal{G}, \\
     [O_\Gamma,\mathcal{G}\mathcal{G}^\dagger]&=0.
\end{align*}
\end{theorem}

\begin{proof}
The proof is by construction. The candidate $O_\Gamma$ is the one defined in \cref{lemma:subgraph-gauging-lemma}
\begin{align*}
    O_\Gamma=\int \prod_{v\in V_\Gamma}\dd g_v \quad \left (\prod_{v\in V_\Gamma}A^\Gamma_v(g_{v})\right) Q
\left (\prod_{v\in V_\Gamma}A^\Gamma_v(g_{v})\right)^\dagger
\end{align*}
with 
\begin{align}
    Q=O_{u} \otimes \bigotimes_{e\in E_\Gamma} \ketbra{I}{I}_e
\end{align}
By lemma \ref{lemma:subgraph-gauging-lemma} we have that $O_\Gamma\in \mathcal{A}_\mathcal{T}(\Gamma)$.
To prove that it implements $O_u$, we show that for all $|\psi\rangle_V\in\mathcal{H}_V$ we have 
\[\Gmc O_{u}\ket{\psi}_V=O_\Gamma\Gmc\ket{\psi}_V.\]
Looking at the right hand side, expanding $\mathcal{G}$ and using the gauge-invariance of $O_\Gamma$ to commute it with $\Pi_{GI}$ gives 
\begin{align*}
    RHS&=\Pi_{GI} O_\Gamma\left (|\psi\rangle_V\bigotimes_{e\in E}|I\rangle_e\right) \\
    &=\Pi_{GI} \int\prod_{v\in V_\Gamma} d{g_v} \left (\prod_{v\in V_\Gamma}A^\Gamma_v(g_{v})\right) Q
\left (\prod_{v\in V_\Gamma}A^\Gamma_v(g_{v})\right)^\dagger \left (|\psi\rangle_V\bigotimes_{e\in E}|I\rangle_e\right) \\
&=\Pi_{GI} \int\prod_{v\in V_\Gamma} d{g_v} \left (\prod_{v\in V_\Gamma}A^\Gamma_v(g_{v})\right) Q
\left (\prod_{v\in V_\Gamma}U^\dagger_v(g_v)|\psi\rangle_V  \bigotimes_{e\in E_{\Gamma}}|g_{e^-}^{-1}g^{}_{e^+}\rangle_e \bigotimes_{f \in E \setminus E^{\Gamma}}\ket{I}_f\right) \\
&=\Pi_{GI} \int\prod_{v\in V_\Gamma} d{g_v} \left (\prod_{v\in V_\Gamma}A^\Gamma_v(g_{v})\right) O_{u}
\left (\prod_{v\in V_\Gamma}U^\dagger_v(g_v)|\psi\rangle_V\bigotimes_{e\in E_{\Gamma}}|I\rangle\langle I|g^{-1}_{e^-}g^{}_{e^+}\rangle_e \bigotimes_{f \in E \setminus E^{\Gamma}}\ket{I}_f\right),
\end{align*}
where we must define $g_{v_0} = I$ for the third line to be valid.
For the edge $e \in E_\Gamma$ such that $e^+ = v_0$ (i.e.\ the edge $e \in E_\Gamma \cap E_0$ with $E_0$ as defined in \cref{subsubsec: graph-theory-background}), the matrix element $\braket{I}{g^{}_{e^-}g^{-1}_{e^+}}_e$ imposes $g_{e^{-}} = I$. Then after integration over $g_v$ and using the constraints from the other matrix elements, we can set $g_v = I$ for all $v\in V_\Gamma$ to obtain 
\begin{align*}
    RHS & = \Pi_{GI}\Dgv  O_{u}
\left (|\psi\rangle_V\bigotimes_{e\in E_{}}\ket{I}_e \prod_{e\in E_{\Gamma}} \braket{I}{g^{-1}_{e^-}g^{}_{e^+}}\right) \\
& = \kappa_{\Gamma}\Pi_{GI} O_{v}\left( \ket{\psi}_{V} \bigotimes_{e\in E} \ket{I}_e\right) \\
& = \kappa_{\Gamma} \mathcal{G} O_v \ket{\psi}_{V}
\end{align*}
with 
\[ \kappa_{\Gamma} = \int \prod_{v\in V_{\Gamma}} \dd{g_v} \prod_{e\in E_{\Gamma}} \braket{I}{g_{e^{-}}^{} g_{e^{+}}^{-1}}=\abs{G}^{-\abs{V_{\Gamma}}}.\]
To prove the implementation condition we need only then to redefine $O_\Gamma\rightarrow \frac{O_\Gamma}{\kappa_\Gamma}$.

Finally, we just need to show that $O_\Gamma$ commutes with the projector onto the image of $\mathcal{G}$, that is, it is block diagonal with respect to a decomposition into the image and its orthogonal complement.
In this decomposition, it follows from $O_\Gamma \mathcal{G} = \mathcal{G} O_u$ that the lower left block of $O_\Gamma$ is zero.
It is straightforward to verify that $O_\Gamma^\dagger \mathcal{G} = \mathcal{G} O_u^\dagger$, which implies the upper right block is zero, and thus $[O_\Gamma,\mathcal{G}\mathcal{G}^\dagger]=0$.
\end{proof}

It is possible to map some operators to their gauged counterparts without a boundary attached dressing. In particular, for a $G$-symmetric $O$, i.e. one  satisfying 
\begin{align*}
    O=\left(\prod_{v\in V}U_v(g)\right)O\left(\prod_{v\in V}U_v(g)\right)^\dagger
\end{align*}
for every $g\in G$ and having support on some subset of vertices $W\subseteq  V_1$, it is possible to use as $\Gamma$ a subgraph which includes all of $W$ and no edges connecting to vertices in $V_0$. This fact was originally shown in Ref.~\cite{Cirac_2015} using essentially the same proof. The idea is that without the constraint imposed by the edge connecting to $V_0$, in the last part of the proof, all $g_v$ are set equal to each other, but not necessarily to $I$, leaving a projection of $O$ onto the space of $G$ symmetric operators. But this does nothing on $G$ symmetric operators. For all other operators a boundary attached dressing is required.

In the appendices under \cref{thm: undressing-local-ops}, we show a sort of converse to the theorem proving a ``reverse dressing property'', namely that any gauged operator that preserves the flux-free sector and is supported on some subgraph $\Gamma$ of $\Lambda$, with certain conditions on $\Gamma$, implements an ungauged operator supported on only the vertices of $\Gamma$.

\subsection{The isometry exhibits a global/gauge duality}

We now show that global symmetry transformations are implemented by NGC symmetry transformations (see \cref{fig:SymmetryBulkToBoundary}).

\begin{figure}[t]
  \centering
    \includegraphics[width=1.0\textwidth]{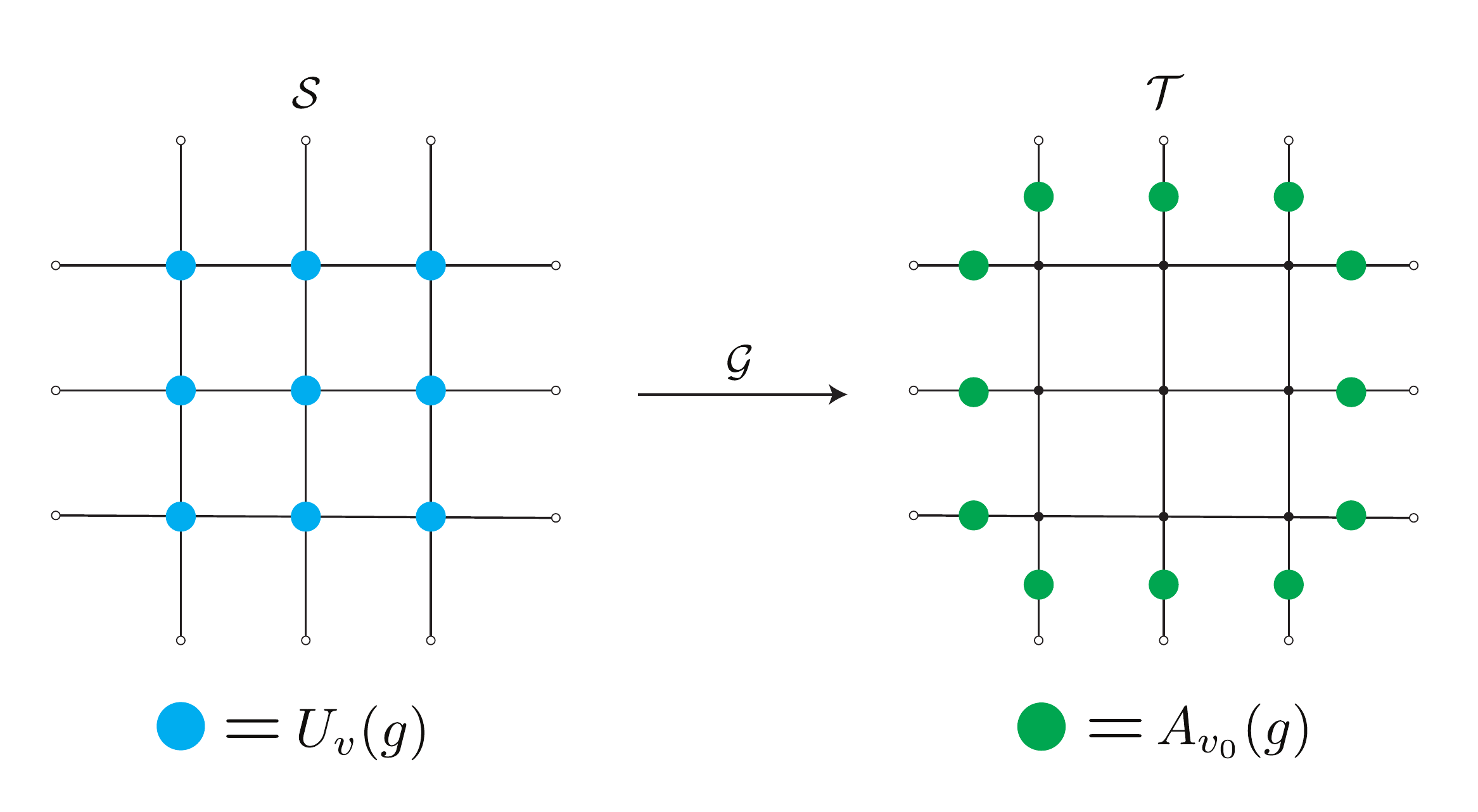}
    \caption{A global symmetry action acting on the input to the gauging map acts as an NGC symmetry transformation acting on its output. Note that we are omitting the action on NGC vertices (represented by white dots) for consistency with future diagrams. This is because for holographic codes we require in the Bulk that all $V_0$ Hilbert spaces are trivial.}
  \label{fig:SymmetryBulkToBoundary}
\end{figure}

\begin{theorem}
 \label{thm:gauge-map-duality}
Consider a system transforming under a global symmetry, $(\mathcal{S}, G, \{U_v\})$  with $\mathcal{S}=(V, \{\mathcal{H}_v\})$ with $\Lambda=(V,E)$ and labeled by $\alpha$ such that $V_0 \neq \varnothing $. Let $\mathcal{T}$ be obtained from gauging $\mathcal{S}$. Let $\mathcal{G}$ be the corresponding gauging isometry. NGC symmetry transformations on the output implement global symmetry transformations on the input, i.e.\ 
% \textbf{Recall what $A_{V_0}$ is }
\begin{align*}
    \forall h\in G, \quad \mathcal{G}U_{V}(h) &= A_{V_0}(h)\mathcal{G}.
\end{align*}
%with $A_{V_0}$ defined as in \cref{def:asymptotic-transformation} and $U_V(h)$ as in \cref{def:global-symmetry-transformation}.
%\[A_{V_0}(h) = \prod_{v \in V_0} A_{v}(h) \quad
%\text{and} \quad
% U_{V}(h) = \prod_{v\in V} U_{v}(h) .\]
\end{theorem}
\begin{proof}
In order to prove the equality, we show that for every $\ket{\psi}_V\in \mathcal{H}_V$ we have 
\[ \mathcal{G}U_{V}(h)\ket{\psi}_V = A_{V_0}(h)\mathcal{G}\ket{\psi}_V.\]
% where for any $W\subseteq  V$ we define 
% \[ U_{W}(h) = \prod_{v\in W} U_{v}(h).\]
% 
Expanding the left hand side we get
\begin{align*}
    LHS=\int \prod_{v\in V_{1}}\dd{g_v} \left(\prod_{v\in V_1}U_v(g_vh)\right)
    U_{V_0}(h)
    |\psi\rangle_V\bigotimes_{e\in E_1 }|g^{}_{e^-}g^{-1}_{e^+}\rangle_e\bigotimes_{e\in E_0 }|g_{e^-}\rangle_{e}
\end{align*}
with $E_0$, $E_1$ as defined in \cref{subsubsec: graph-theory-background}. % and $U_{V_0}(h) = \prod_{v\in V_{0}} U_{v}(h)$.
Changing variables to $g'_v=g_vh$ for $v\in V_1$ gives
\begin{align*}
    LHS&=\int \prod_{v\in V_{1}}\dd{g'_v} \left(\prod_{v\in V_1}U_v(g_v')\right)
    U_{V_0}(h)
    |\psi\rangle_V\bigotimes_{e\in E_1 }|(g')^{}_{e^-}(g')^{-1}_{e^+}\rangle_e\bigotimes_{e\in E_0 }|g'_{e^-}h^{-1}\rangle_{e} \\
    &=U_{V_0}(h) \prod_{e\in E_0 } U^R_{e}(h) \mathcal{G}\ket{\psi}_V = A_{V_0}(h)\mathcal{G}\ket{\psi}_{V}.
\end{align*}
\end{proof}

\subsection{The image of the isometry is the flux-free sector}
\label{sec:Image-Isometry-Is-Flux-Free}
We now show that the image of the gauging isometry is the flux-free sector of the gauge-invariant Hilbert space, when the $V_1$ induced subgraph is connected.

\begin{lemma}
\label{lem: gauge-map-wilson-link}
For a gauging map $\mathcal{G}$ and a Wilson link operator $W_{ij}^{e}$
\begin{align}
    W^{{e}}_{ij}\mathcal{G} &= \delta_{ij} \mathcal{G}\\
    W^{\overline{e}}_{ij}\mathcal{G} &= \delta_{ij} \mathcal{G}
\end{align}
    
\begin{proof}
Consider an arbitrary state $\ket{\psi}_{V}$, then $\mathcal{G}\ket{\psi}_V = \Pi_{GI} \left( \ket{\psi}_V \otimes \bigotimes_{e\in E} \ket{I}_e\right)$. As the Wilson link operators are gauge-invariant, we can move them past the projector and act them on $\ket{I}_e$ to give the desired result.
\end{proof}
\end{lemma}

\newcommand{\fluxfreegauge}{%
Consider a system transforming under a global symmetry, $(\mathcal{S}, G, \{U_v\})$  with $\mathcal{S}=(V, \{\mathcal{H}_v\})$ and a connected graph $\Lambda=(V,E)$ with $\alpha$ labeling such that $\Lambda\setminus(V_0\cup E_0)$ is still connected. Let $\mathcal{T}$ be obtained from gauging $\mathcal{S}$. Let $\mathcal{G}$ be the corresponding gauging isometry. The image of $\mathcal{G}$ is $\Pi_{FF} \mathcal{H}_\Lambda$, the entire flux-free sector of the gauge-invariant Hilbert space.
}
\begin{theorem}
\label{thm: gauging-map-image-is-flux-free}
\fluxfreegauge
\end{theorem}

\begin{proof}
By \cref{lem: gauge-map-wilson-link} and the definition of flux-free sector, any state in the image of $\mathcal{G}$ is in the flux-free sector, i.e. 
\[ \image[\mathcal{G}] \subseteq  \image[\Pi_{FF}]. \]

The other inclusion is more difficult to prove, thus we sketch it here but relegate a full proof to the appendices under \cref{proof:gauging-map-image-is-flux-free}.
By the comments at the end of \cref{def:flux-free}, note that 
% we can write any state in the flux-free sector $\Pi_{FF}\ket{\Psi}_{GI}\in \Pi_{FF}\mathcal{H}_{\Lambda}$ as
% \[ \Pi_{FF}\ket{\Psi}_{GI}=\Pi_{GI}\left(\sum_{i=1}^{N} c_{i} \ket{\psi_i}_V\otimes \ket{\{g^{(i)}_e\}}_{E}\right),\]
% with $\ket{ \{g_e\} }_E $ some product state of group elements on each edge Hilbert space, living in the support of $\Pi^E_{FF}$ (i.e.\ the trivial eigenspace of the relevant Wilson operators).
% In other words, 
states of the form $$\Pi_{GI} \left(  \ket{\psi}_V \otimes \bigotimes_{e\in E} \ket{ g_e } \right)$$ form a basis for the flux-free sector%
, with $\bigotimes_{e\in E} \ket{ g_e } $ living in the support of $\Pi^E_{FF}$ (i.e.\ the flux-free sector when all vertex degrees of freedom are trivial).
Thus we just need to show that every such state lies in the image of $\mathcal{G}$.

If the edge configuration $\bigotimes_{e\in E} \ket{g_e} $ were the trivial one, $\bigotimes_{e\in E} \ket{I_e} $, then we would be done -- this would simply be equal to $\mathcal{G} \ket\psi_V$.
Lemma 3.3 of Ref.~\cite{Cui_2020} allows us to relate the above configuration to the trivial one.
It shows that because these edge configurations are both in the flux-free sector, they can be related by some product of gauge transformations, up to corrections that act just on vertices and not edges (and thus can just be absorbed into a redefinition of $\ket\psi_V$). %a product state of edge configurations in the flux free sector $\bigotimes_{e\in E} \ket{g_e} $ is related by gauge transformations to the trivial edge configuration, $\bigotimes_{e\in E} \ket{I_e}$.
The gauge transformations associated with GC vertices, by the way that the gauge invariant sector was constructed, are simply absorbed into the gauge-invariant constraint $\Pi_{GI}$.

The gauge transformations associated with NGC vertices are slightly more complicated.
Because the edge configuration in question is in the flux-free sector, it turns out that all gauge transformations on NGC vertices used to relate it to the trivial configuration must actually correspond to the same group element. %, i.e.\ $\forall v,v'\in V_0, h_v = h_{v'}$.
Thus the product of all such gauge transformations is just an NGC symmetry transformation, which via \cref{thm:gauge-map-duality} is dual to a global symmetry on the input system.
This global symmetry can also just be absorbed into a redefinition of the vertex configuration, $\ket\psi_V$.
Thus we can effectively replace the edge configuration above with the trivial one, and conclude that the state $\Pi_{GI} \left(  \ket{\psi}_V \otimes \bigotimes_{e\in E} \ket{ g_e } \right)$ is in the image of the gauging map.

For more details, see the proof of \cref{proof:gauging-map-image-is-flux-free} in the appendices.

% Lastly, notice that any state $\Pi_{FF}\ket{\Psi}_{GI}\in \Pi_{FF}\mathcal{H}_{\Lambda}$ is of the form 
% \[ \Pi_{FF}\ket{\Psi}_{GI}=\Pi_{GI}\left(\sum_{i=1}^{N} c_{i} \ket{\psi_i}_V\otimes \ket{\{g^{(i)}_e\}}_{E}\right).\]
% with $\ket{\{g^{(i)}_e\}}_{E}$ being flux-free. After some manipulations we find
% \begin{align*}
%      \Pi_{FF}\ket{\Psi}_{GI} &= \Pi_{GI}\left(\sum_{i=1}^{N} c_{i} \left(\mathcal{U}_{(i)}\ket{\psi_i}_V\right)\otimes \ket{\{I_e\}}_{E}\right) \\
%      & = \Pi_{GI}\left(\sum_{i=1}^{N} c_{i} \left(\mathcal{U}_{(i)}\ket{\psi_i}_V\right)\right) \otimes \ket{\{I_e\}}_{E} \\
%      & = \mathcal{G}\left(\sum_{i=1}^{N} c_{i} \mathcal{U}_{(i)}\ket{\psi_i}_V\right).
% \end{align*}
% where $\mathcal{U}_{(i)}$ is the unitary that satisfies $\Gmc \mathcal{U}_{(i)}\ket{\psi} = \Pi_{GI} \ket{\psi}\otimes \ket{\{g^{(i)}_{e}\}}_{E} $.
\end{proof}

\subsection{NGC flux sectors and beyond}

\label{subsec: NGC-flux-maps}

In this section, we consider a family of gauging maps related to the original by NGC gauge transformations.
We show that each such map displays the four major properties shown above for the standard gauging map, all of which are relevant for later sections.
Namely, it is an isometry, it implements local operators via dressed ones, it exhibits a gauge/global duality, and it is surjective on the relevant flux sector. 
In \cref{subsec:connection-to-harlow-ooguri} we find that this family of generalized maps is related to acting with restricted global symmetry actions on the Boundary of a holographic code. 
% This family is still restricted to have trivial values for closed Wilson loops; 
We state some results for even more general gauging maps but leave their proofs to \cref{subsec: other-flux-sectors-proofs}. 

\begin{definition}
\label{def: generalized-gauging-map}
Consider an unconstrained system transforming under a  global symmetry, $(\mathcal{S}, G, \{U_v\})$  with $\mathcal{S}=(V, \{\mathcal{H}_v\})$ and an oriented graph $\Lambda=(V,E)$ with $\alpha$ labeling. Let $\mathcal{T}=(\mathcal{S},\Pi_{GI})$ be the constrained system obtained by gauging $\mathcal{S}$. For any function $\phi: E \rightarrow G$, we define the twisted gauging map as 
\begin{equation}\label{eq: twisted-gauing-map}
    \mathcal{G}_{\phi}: \ket{\psi}_{V} \longmapsto \Pi_{GI} \left( \ket{\psi}_{V}\bigotimes_{e\in E} \ket{\phi_{e}}_{e}\right)
\end{equation}
\end{definition}
\begin{definition}
    We call the subset of such maps that are related to the original gauging map by NGC gauge transformations \textbf{NGC flux gauging maps}. More explicitly such maps can be written as 
    \begin{align*}
        \mathcal{G}_\phi=U_0\mathcal{G}=\prod_{v\in V_0} A_{v}(h_v)\mathcal{G}
    \end{align*}
    for some set of group elements $\{h_v\}_{v\in{V_0}}$. Notice that any such choice of group elements gives a legitimate twisted gauging map, since they can be commuted past $\Pi_{GI}$ to act on the edge degrees of freedom.
\end{definition}

% NGC flux gauging maps with $\phi_e=I$ for all $e\in E_1$ can be rewritten in terms of the original ``flux free'' gauging map by rewriting
% \begin{align*}
%     \Pi_{GI} \left( \ket{\psi}_{V}\bigotimes_{e\in E} \ket{\phi_{e}}_{e}\right)=\prod_{e\in E_0}U_e^R(\phi_e^{-1})\Pi_{GI} \left( \ket{\psi}_{V}\bigotimes_{e\in E} \ket{I}_{e}\right).
% \end{align*}
% %
% Where we used that $U^R_e(g)$ is gauge invariant for $e\in E_0$.\footnote{We must use $U^R$ and not $U^L$ because of the convention that edges connecting vertices in $V_0$ to those in $V_1$ always point from the former to the latter. Gauge transformations in the gauge constraint involving $e\in E_0$ have terms acting on it only of the $U^L_e(h)$, which commutes with $U^R_e(g)$.}  Thus we have
% %
% \newcommand{\Uphi}{U_{0}}
% \begin{align*}
%     \mathcal{G}_\phi=\Uphi^{\dagger}\mathcal{G}
% \end{align*}
% %
% for $\phi$ with only NGC fluxes and $\Uphi=\prod_{e\in E_0}U_e^R(\phi_e)$. For NGC flux gauging maps with $\phi'$ related to $\phi$ by gauge transformations there is an additional local action on the input vertices. This does not affect the proofs below, so we omit these cases.
Showing that an NGC flux map $\mathcal{G}_\phi$ has the same properties as $\mathcal{G}$ is essentially just a matter of stating that they are unaffected by a local basis change. 
Note that, we can reach an even larger class of maps by acting acting on the output of the flux free map $\mathcal{G}$ with local unitaries that act as central elements of $G$ on the $E_1$ edge degrees of freedom.
Because they act as central elements, they commute with the gauge projector and are thus equivalent to changing the flux configuration $\phi_e$ for $e\in E_1$ before the gauging map was applied.%, up to gauge transformations.

That the map $\mathcal{G}_\phi$ is an isometry (when $\Lambda$ is connected and $V_0\neq\varnothing$) for NGC flux maps is clear. In \cref{thm: general-twisted-map-is-an-isometry} we show that it holds for any choice of $\phi$.

Local operators in the ungauged system also continue to be implemented by dressed operators in the gauged system, which can be seen as follows. Let $O_u$ be an $\mathcal{A}_\mathcal{S}(u)$ operator, and let $O_\Gamma$ be the $\mathcal{A}_\mathcal{T}(\Gamma)$ operator with arbitrary dressing along the path $\Gamma$ which one may use to implement $O_u$ using the original gauging map, i.e. $\mathcal{G}O_u =O_{\Gamma}\mathcal{G}$. Then 
\begin{align*}
    \mathcal{G}_\phi O_u=U_0 O_\Gamma U_0^\dagger\mathcal{G}_\phi=O'_\Gamma \mathcal{G}_\phi,
\end{align*}
where $O'_\Gamma = U_0 O_\Gamma U_0^\dagger$ is supported on $\Gamma$ because $U_0$ is a product of local unitaries which preserves the support of an operator. 

The reverse dressing property also holds. More explicitly, for any subgraph $\Gamma$ satisfying the conditions in \cref{thm: undressing-local-ops}, an operator $O_\Gamma\in\mathcal{A}_\mathcal{T}(\Gamma)$ which preserves the image of $\mathcal{G}_\phi$\footnote{For the original gauging map we required that it preserve the flux-free sector. It is easier for the sake of the proof to require only the image is preserved, and as we soon see these are the same anyway.} implements an operator $O_{V_\Gamma}$ supported on the vertices of $\Gamma$. This can be proven as follows. First, notice that
\begin{align*}
    U^\dagger_0 O_{\Gamma} \mathcal{G}_{\phi} = U^\dagger_0 O_{\Gamma} U_0\mathcal{G}.
\end{align*}

This means $U^\dagger_0O_\Gamma U_0$ preserves the image of $\mathcal{G}$. Furthermore this operator is also supported only on $\Gamma$. Thus by the reverse dressing property for the original gauging map, it implements an operator $O_{V_\Gamma}$ supported on only the vertices of $\Gamma$, i.e. 

\begin{align*}
    U^\dagger_0 O_\Gamma\mathcal{G}_\phi=\mathcal{G}O_{V_\Gamma}.
\end{align*}
Multiplying the left by $U_0$ gives the desired result.\footnote{We speculate that the dressing property and reverse dressing property break down in general, but we leave concrete results in this direction to future work.}

The global/gauge property holds for NGC flux maps up to a local basis change of the symmetry actions. More explicitly,
\begin{align*}
    \mathcal{G}_\phi U_V(h)
    =\prod_{v\in V_0}U_0 A_{v}(h_v)U_0^\dagger\mathcal{G}_\phi
    =\prod_{v\in V_0} \tilde{A}_{v}(h_v)\mathcal{G}_\phi
\end{align*}
with $\tilde{A}_{v}$ still representations of $G$.
In \cref{thm: general-twisted-gauge-global-duality} we show that for general $\phi$ we get 
\begin{align*}
    \mathcal{G}_{\phi}U_{V}(h) = A_{V_0}(h)\mathcal{G}_{^{h}\!\phi}.
\end{align*}
where $^{h}\!\phi: e \longmapsto  h^{-1}\cdot \phi_{e} \cdot h^{}$. Interestingly, 
% this means that for an arbitrary $\phi$, we at least have that the centralizer of $\phi(E_1)$ be a subgroup of $G$ for which the gauge/global duality still holds (up to a local basis transformation). 
this means that the gauge/global duality still holds for $h$ in the centralizer of $\phi(E)$.

We provide the proof that the image of an NGC flux map is the full relevant flux sector in \cref{thm: NGC-map-image-is-full-sector-sometimes}, where we additionally assume that $\Lambda\setminus(V_0\cup E_0)$ is still connected. For general $\phi$ we no longer have that the image is the full flux sector, because as discussed in the end of \cref{subsubsec: gauge-sym-fixed-flux-sectors}, for groups with outer class automorphisms there are additional pure gauge degrees of freedom besides Wilson loops. This means that it is not possible in general to reach every configuration of gauge degrees of freedom in a given nontrivial fixed flux sector by acting with gauge transformations. Since the gauging map essentially takes the superposition of all gauge transformations, there are some states in such a flux sector outside its image. 
Although the image of a single gauging map may not be surjective in a fixed nonabelian flux sector, there are many gauging maps with image in the same flux sector and the union of all their images coincides with the full fixed flux sector. 

\section{Holographic codes} \label{sec: holographic-codes}

A holographic code is a type of error correcting code which encodes a logical ``Bulk'' system living on discretized $d$ dimensional geometry into a physical ``Boundary'' system which lives on a discretized $d-1$ dimensional geometry, with the geometry of the Boundary inherited from the boundary of the Bulk geometry~\cite{almheiriBulkLocalityQuantum2015,HaPPY}. 

%A \textit{holographic code} is a toy model designed to capture aspects of holographic theories such as AdS/CFT, and generally consist of a pair of systems -- the Bulk, whose geometry is the interior of a manifold, and the Boundary, whose geometry is its exterior -- with an isometry from Bulk to Boundary.
We define more specific requirements for a holographic code in this section, and give two examples: 1) the HaPPY code~\cite{HaPPY}, which has an unconstrained Bulk system and exhibits a $\mathds{Z}_2$ global/global duality, and 2) what we have dubbed the gauged LOTE code~\cite{Marolf}, whose Bulk system is a full $\mathds{Z}_2$ gauge constrained system. 
We then show that any holographic code whose Bulk system is a gauge-constrained system with group $G$, and whose constrained subspace is the full gauge-invariant sector (i.e.\ the constraint is $\Pi_{GI}$), exhibits a gauge/global duality. 
The definition of holographic codes, as well as this result on the existence of gauge/global dualities, is used in section~\ref{sec:applications-to-holography} to demonstrate an equivalence between holographic codes with various symmetry dualities. 

\subsection{Working definition}

For clarity we restrict our discussion to Bulk geometries that are two dimensional. This geometry determines which subsystems of the Bulk system are encoded into which subsystems of the Boundary system~\cite{almheiriBulkLocalityQuantum2015}. This is usually done via a Ryu-Takayanagi (RT)-like prescription~\cite{Ryu2006a,Ryu2006b,Hubeny2007}, but for the sake of both clarity and generality we only require prescriptions that obey a property we call \textit{near-boundary probing}. For the remainder of the work we describe the Bulk geometry with a planar graph with quantum degrees of freedom living on its nodes and edges. In higher dimensions, it would be more appropriate to use something like a CW complex \cite{Hatcher}. 

\begin{figure}[t]
  \centering
    \includegraphics[width=1.0\textwidth]{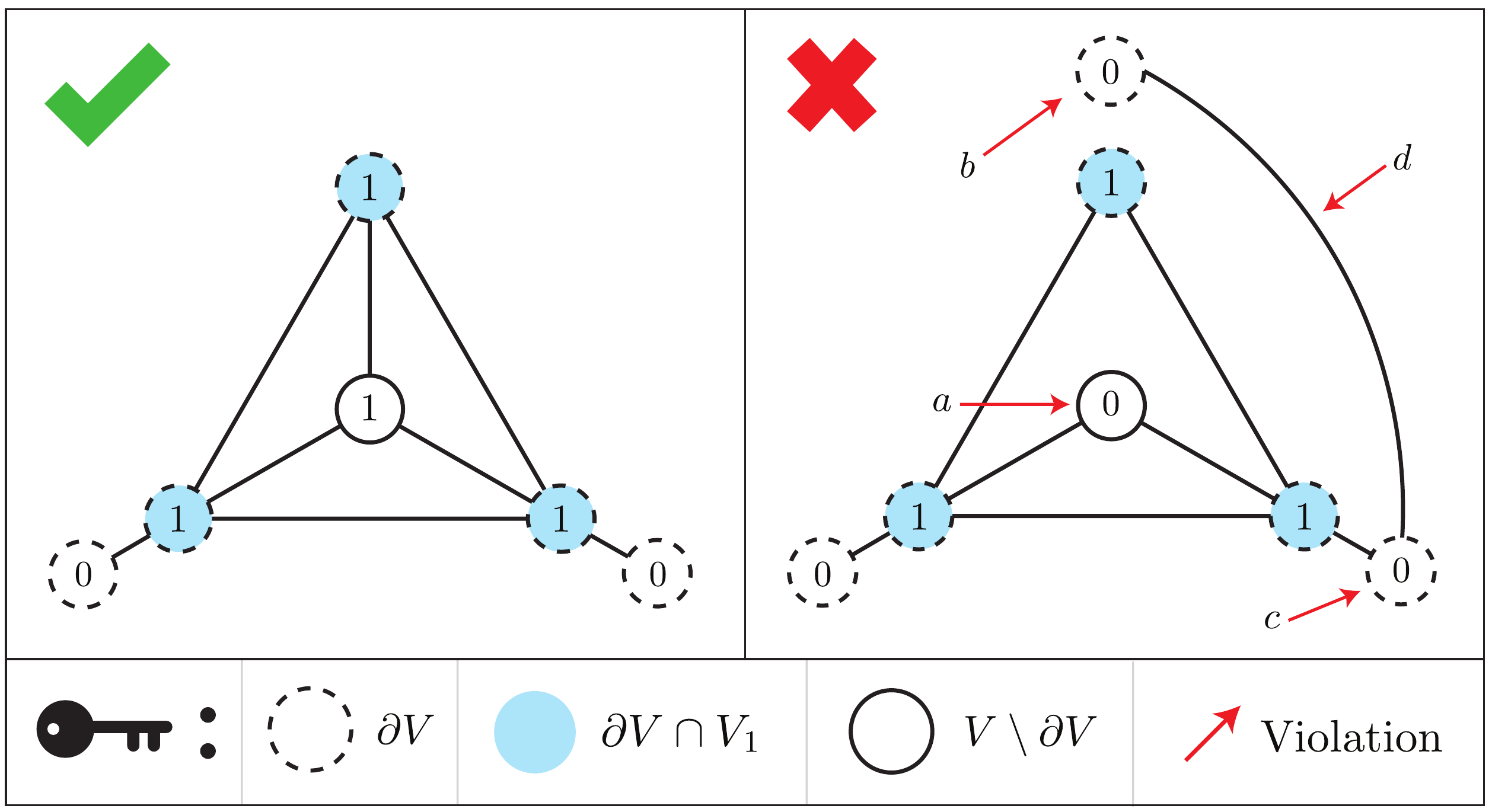}
    \caption{(Left) a planar graph with dangling edges. 
    %The directionality of the edges is suppressed for clarity. 
    (Right) a graph which does not meet the requirements for a planar graph with dangling edges because of the following issues: a) there is a node in $V_0$ which is not in $\partial V$, b) there is a node in $V_0$ which is not connected to a node in $\partial V\cap V_1$, c) there is a node in $V_0$ which participates in more than one edge, and d) there is an edge connecting two nodes in $V_0$.}
  \label{fig:dangling-edges}
\end{figure}

\begin{definition}\label{def:dangling-edges}
(illustrated in \cref{fig:dangling-edges})
Given a planar graph $\Lambda=(V,E)$ with boundary vertices $\partial V\subseteq V$\footnote{A boundary vertex is one from which it is possible to draw a path reaching to infinity on the embedding plane without intersecting another part of the graph.} and bit labeling $\alpha:V\rightarrow\{0,1\}$, we say that $(\Lambda,\alpha)$ is a \textbf{planar graph with dangling edges} if the following conditions hold
\begin{enumerate}
        \item $V_0 \subseteq \partial V$
        %\item For all $v\in\partial V\cap V_1$ there exists $u\in V_0$ such that $(u,v)\in E$.  
        \item For all $u\in V_0$ there exists exactly one vertex $v\in V$ connected by an edge to $u$. Furthermore this edge is $(u,v)$ and $v\in V_1 \cap \partial V$.
\end{enumerate}
In the context of holographic codes, we sometimes refer to $E_0$ (that is, the set of edges containing one vertex in $V_0$) as the set of \emph{dangling edges}.
\end{definition}
For holographic codes, the $V_0$ vertices comprise the Boundary system\footnote{%
Recall that Boundary is capitalized when referring to the system and its vertices $V_0$ to distinguish from other uses of ``boundary'', such as the boundary vertices of the planar graph, $\partial V$.%
}, so these conditions can be phrased as follows: all Boundary vertices are on the boundary, and connected to exactly one other vertex which is a Bulk vertex also on the boundary (with the edge pointing from Boundary to Bulk).

\begin{figure}[t]
  \centering
    \includegraphics[width=1.0\textwidth]{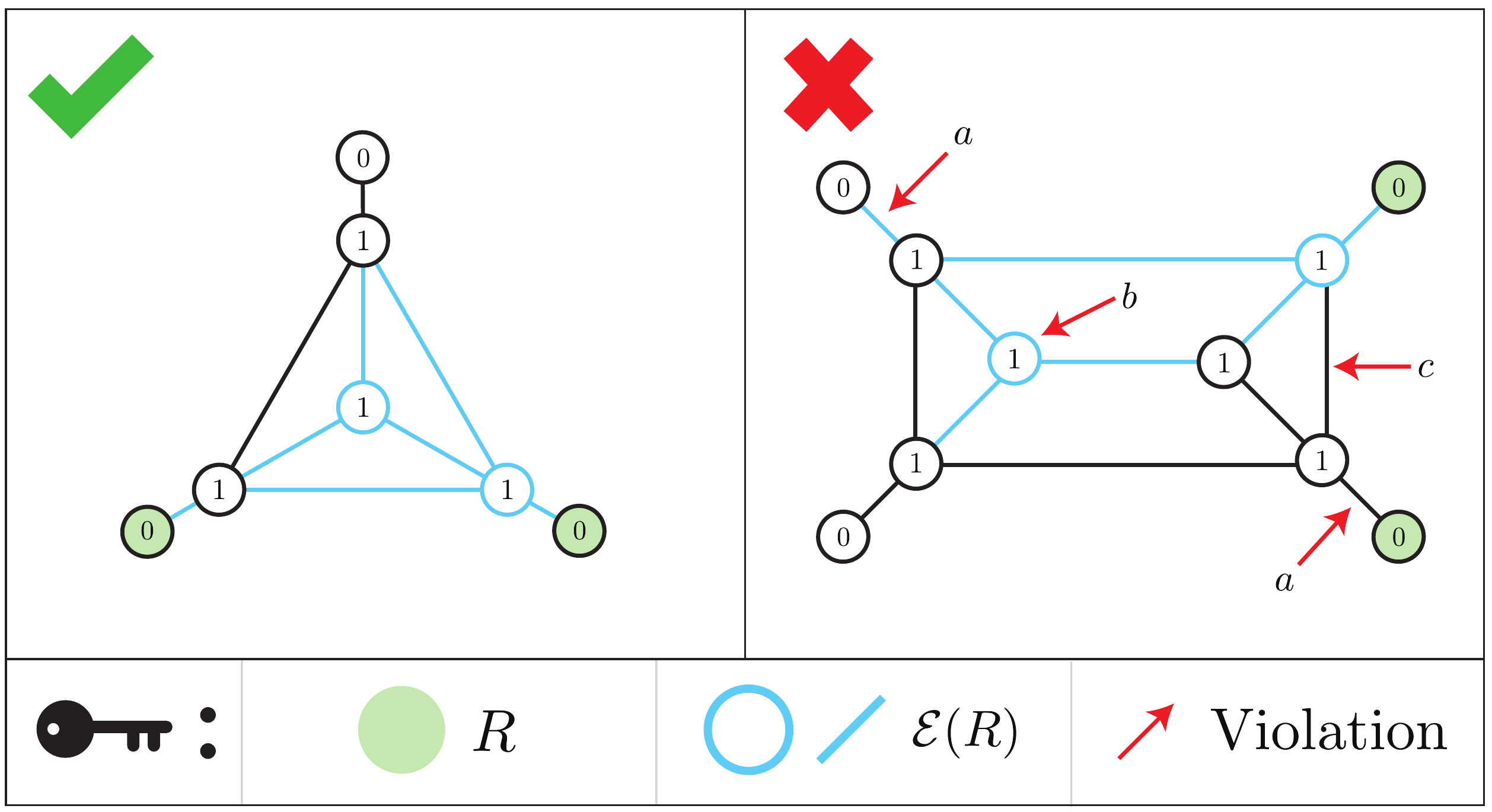}
    \caption{(Left) An example of a map $\mathcal{E}$ which, for the $R$ shown, has a $\mathcal{E}(R)$ which is allowed for a near-boundary probing map. (Right) an example of an $\mathcal{E}(R)$ which is prohibited for a near-boundary probing map because of the following issues: a) There are elements of $E_0\cap\mathcal{E}(R)$ not connected to vertices in $R$ and elements of $E_0$ connected to edges in $R$ that are not in $\mathcal{E}(R)$, both of which violate property 1 of the definition. b) there are vertices in $\mathcal{E}(R)$ not connected to an element of $R$ by a path in $\mathcal{E}(R)\cup R$, and c) there are non-dangling edges $e\notin \mathcal{E}(R)$ that are incident on vertices in $\mathcal{E}(R)$ which are themselves not in $\mathcal{E}(R)$.}
  \label{fig:near-boundary-probing}
\end{figure}

% \begin{definition} \label{def:near-boundary-probing}
% (illustrated in \cref{fig:near-boundary-probing})
% Consider a connected planar graph with dangling edges $(\Lambda=(V,E),\alpha)$. We say that a map $\mathcal{E}:2^{V_0}\rightarrow 2^{V\cup E}$ \footnote{recall that for a set $S$ its power set $2^S$ is the set of all of its subsets.} is $\textbf{near-boundary probing}$ if it possesses the following two properties. First, $\mathcal{E}(V_0)=V\cup E$. Second, for any $R\in2^{V_0}$, $\mathcal{E}(R)$ can be obtained in the following way. Let $\bar{\Lambda}=(\bar{V},\bar{E})$ be the dual graph of $\Lambda$. There should exist a set of cycles in $\bar{\Lambda}$ that each enclose at least one element of $R$, that collectively enclose all elements of $R$, but which enclose no elements of $V_0\setminus R$. $\mathcal{E}(R)$ is the union of the following sets:
%     \begin{enumerate}
%         \item The set of all vertices in $V$ enclosed by these cycles,
%         \item The set of all edges in $E$ enclosed by or intersecting these cycles.
%     \end{enumerate}
% \end{definition}

We now define the necessary requirements for the \emph{entanglement wedge map}, which plays a crucial role in the definition of a holographic code.
{{
\begin{definition} \label{def:near-boundary-probing - B}
(illustrated in \cref{fig:near-boundary-probing})
Consider a connected planar graph with dangling edges $(\Lambda=(V,E),\alpha)$. We say that a map $\mathcal{E}:2^{V_0}\rightarrow 2^{V_1\cup E}$ \footnote{Recall that for a set $S$ its power set $2^S$ is the set of all of its subsets.} is $\textbf{near-boundary probing}$ if it possesses the following three properties. 
\begin{enumerate}
    % \item $\mathcal{E}(R)$ contains $R$, i.e. $R\subseteq  \mathcal{E}(R)$.
    % \item $\mathcal{E}(R)$ does not overlap with the boundary complement of $R$, i.e. $(\partial V - R)\cap \mathcal{E}(R).= \{\}$.
    %\item  $V_0\cap \mathcal{E}(R) = R$.
    \item The dangling edges in $\mathcal{E}(R)$ are exactly those that connect to vertices in $R$; that is, $\mathcal{E}(R) \cap E_0 = \{ (u,v) \in E_0 | u \in R \}$ .
    \item Every vertex in $\mathcal{E}(R)$ is connected by a path in $\mathcal{E}(R)\cup R$ to an element of $R$.
    \item A non-dangling edge is in $\mathcal{E}(R)$ if and only if it is incident on at least one vertex in $\mathcal{E}(R)$.
\end{enumerate}
Furthermore, we define the \emph{exterior} of $\mathcal{E}(R)$ to be $\exterior{\mathcal{E}(R)} = \{ e \in \mathcal{E}(R) | \exists v\in e | v \notin R \cup \mathcal{E}(R) \}$, i.e.\ those edges in $\mathcal{E}(R)$ that connect to only one vertex in $R\cup \mathcal{E}(R)$.
We also define its \emph{interior} $\interior{ \mathcal{E}(R)} $ to be all other edges and vertices.
\end{definition}

%The definition essentially means that $\mathcal{E}(R)$ can be retracted to $R$ and its restriction to $V_0$ is exactly $R$.
}}

\begin{definition} \label{def:holographic-code}

A \textbf{holographic code} $(\mathcal{T},\Lambda,\alpha,\mathcal{B},\mathcal{E},\mathcal{V})$ consists of
\begin{enumerate}
    \item A logical constrained system $\mathcal{T}=(\mathcal{U},\Pi)$ called the \textbf{Bulk} and a connected planar graph with dangling edges $(\Lambda=(V,E),\alpha)$ such that  $\mathcal{U}=(V_1\cup E,\{\mathcal{H}_{v\in V_1}\}\cup\{\mathcal{H}_{e\in E}\})$, where $V_1$ as usual refers to the subset of $V$ labeled $1$ by $\alpha$. 
    
    \item A physical system $\mathcal{B}=(V_0,\{\tilde{\mathcal{H}}_{v\in V_0}\})$ called the \textbf{Boundary}.%where $V_0$ as usual refers to the subset of $V$ labeled $0$ by $\alpha$.
    \item A map $\mathcal{E}:2^{V_0}\rightarrow 2^{V_1\cup E}$ called the \textbf{entanglement wedge} which is near-boundary probing; and such that the map $R \mapsto \mathcal{E}(R^c)^c \cup \exterior{\mathcal{E}(R^c)}$ is also near-boundary probing.
    \item An algebra $\mathcal{A}_{\mathcal{E}(R)}$ for each Boundary region $R\subseteq V_0$ called the \textbf{entanglement wedge algebra}, which is a subalgebra of the physical algebra restricted to $\mathcal{E}(R)$; formally, $\{\mathcal{A}_{\mathcal{E}(R)} \subseteq \mathcal{A}_\mathcal{T} (\mathcal{E}(R)) \}_{R\subseteq V_0}$.
    Furthermore, we require it to contain the full algebra of physical operators supported on the interior of $\mathcal{E}(R)$, i.e.\ $\mathcal{A}_{\mathcal{T}} (\interior{\mathcal{E}(R)}) \subseteq \mathcal{A}_{\mathcal{E}(R)} $.
    \item An \textbf{encoding isometry} $\mathcal{V}:\Pi\mathcal{H}_{V_1 \cup E} \rightarrow \tilde{\mathcal{H}}_{V_0}$ from the Bulk to the Boundary.
\end{enumerate}
Furthermore it obeys a condition known as ``entanglement wedge reconstruction''~\cite{dongReconstructionBulkOperators2016,Cotler2019}, that is for any $R\subseteq V_0$, any Bulk operator localized to the entanglement wedge of $R$ which is in the corresponding entanglement wedge algebra, i.e.\ every $O\in\mathcal{A}_{\mathcal{E}(R)}$ must be implementable by a Boundary operator localized to $R$, $\tilde{O}\in\mathcal{A}_\mathcal{B}(R)$, i.e. 
    \begin{align*}
        \mathcal{V}O=\tilde{O}\mathcal{V}
    \end{align*} 
and $[\tilde{O},\mathcal{V}\mathcal{V}^\dagger]=0$.\footnote{The entanglement wedge reconstruction condition for a given Boundary region $R$ and the algebra associated to its entanglement wedge $\mathcal{A}_\mathcal{T} (\mathcal{E}(R))$ is equivalent, via the so called ``cleaning lemma'', to the condition that $V_0\setminus R$ is ``correctable'' in the error correction sense with respect to $\mathcal{A}_\mathcal{T} (\mathcal{E}(R))$. See \cref{sec:qecc} for further details.}
%\textsuperscript{,}\footnote{In general one might only require a sub-algebra of $\mathcal{A}_\mathcal{T} (\mathcal{E}(R))$ to be reconstructable, but this is not necessary for the central points of this paper and would serve to only further clutter the notation. \label{footnote: more-general-holocode-def}}
\end{definition}

%Note that strictly speaking, this definition could be satisfied for essentially any code given a trivial choice of entanglement wedge algebra.
%Notable holographic codes have entanglement wedge algebras that are large subalgebras of the corresponding $\mathcal{A}_{\mathcal{T}}(\mathcal{E}(R))$; for example, as we will see in HaPPY this is the full algebra, and the entanglement wedge algebra of ungauged LOTE is as large as it could be without violating the no-cloning theorem.

Though our definition of holographic codes and results about them naturally generalize to higher dimensions, in that case there are additional types of gauging prescriptions corresponding to placing degrees of freedom on simplices with dimension greater than one, which we have not explored. In the language of field theory: though our results apply to 1-form gauge fields in any dimension, we expect similarly interesting results for higher-form gauge fields~\cite{Gaiotto2015}. 

% \FloatBarrier

\subsection{Visualizing many body systems with tensor networks}

\begin{figure}[t]
     \centering
     \begin{subfigure}[t]{0.30\textwidth}
         \centering
\includegraphics[width=\textwidth]{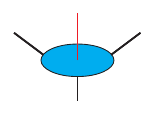}
         \caption{Depiction of a four legged tensor. The legs may be colored differently in order to emphasize some grouping among them.}
         \label{fig:example-tensor}
     \end{subfigure}
     \hfill
     \begin{subfigure}[t]{0.30\textwidth}
         \centering
         \includegraphics[width=\textwidth]{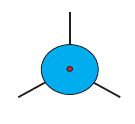}
         \caption{The same tensor ``viewed from the top'' so that the red leg is coming out of the page.}
         \label{fig:example-tensor-top}
     \end{subfigure}
     \hfill
     \begin{subfigure}[t]{0.30\textwidth}
         \centering
         \includegraphics[width=0.5\textwidth]{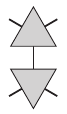}
         \caption{Two three legged tensors with one leg from each contracted together.}
         \label{fig:contraction-example}
     \end{subfigure}
        \caption{Some basic diagrammatic representations of tensor networks.}
        \label{fig:tensor-networks}
\end{figure}

Before introducing examples of holographic codes, we first introduce the tools used to build and visualize all examples known to date, namely tensor networks.
We only do so briefly, and readers unfamiliar with tensor networks are encouraged to refer to Refs.~\cite{HaPPY,handwaving} for more details.

The basic building block of a tensor network is, unsurprisingly, a tensor. A tensor with $n$ indices is represented diagrammatically by some shape with lines, or ``legs'' coming out of it, with each leg representing one index of the tensor, see figures \ref{fig:example-tensor} and \ref{fig:example-tensor-top}. 
A tensor network consists of a set of tensors, each with some number of indices, with some indices contracted\footnote{A contraction between two indices refers to a summation over that index in both tensors, e.g.\ for two two-legged tensors $M_{ij}$ and $N_{kl}$, contracting the first index of the former with the second of the latter would give $\sum_i M_{ij} N_{ki}$.} together among the various tensors. Diagrammatically, this contraction is represented by the two corresponding legs joining together into a single line -- see \cref{fig:contraction-example}.
Legs can also remain unconnected, in which case they represent an uncontracted index of the resultant tensor.\footnote{Such as $j$ and $k$ in the example of the previous footnote.}

In our context, legs are generally identified with corresponding Hilbert spaces.
For example, a code encoding a single qubit into five qubits can be represented as a six-legged tensor $M_{i}^{o_1o_2o_3o_4o_5}$, with one input leg and five output legs.
This is defined implicitly by the encoding isometry $V:\mathcal{H}_2 \to \mathcal{H}_2^{\otimes 5}$, as
\begin{align}
    V \ket{i} = \sum_{o_1,o_2,o_3,o_4,o_5} M_{i}^{o_1o_2o_3o_4o_5} \ket{o_1} \otimes\ket{o_2} \otimes\ket{o_3} \otimes\ket{o_4} \otimes\ket{o_5},
\end{align}
according to some choice of basis for each individual Hilbert space.
In this way, we can identify the leg of $M$ indexed $o_2$ with the Hilbert space of the second physical qubit, and the leg indexed by $i$ with the logical Hilbert space. 
%Whereas the systems defined in \cref{subsec: QMBS} consist purely of Hilbert spaces, holographic codes additionally require an isometry between two different systems as part 

\subsection{Example with unconstrained Bulk: the HaPPY code}\label{subsubsec: holocodes-example-happy}

The first and most common example of a holographic code is the celebrated HaPPY code introduced in \cite{HaPPY}.
We describe it here, although we do not go in depth into its structure except as necessary to see how it fits into our definition of a holographic code.
We then discuss the global/global symmetry duality it exhibits and the mystery it presents. 

The specific family of HaPPY codes we discuss here requires two main ingredients: 1) a uniform pentagonal tiling of hyperbolic space with a radial cutoff, see \cref{fig: hyperbolic-tiling}, and 2) a $6$-legged so-called \textit{perfect tensor}. The code is constructed by contracting many copies of the perfect tensor in a manner which mimics the tiling, see \cref{fig:HaPPY-code}.

\begin{figure}[t]
  \centering
    \includegraphics[width=0.7\textwidth]{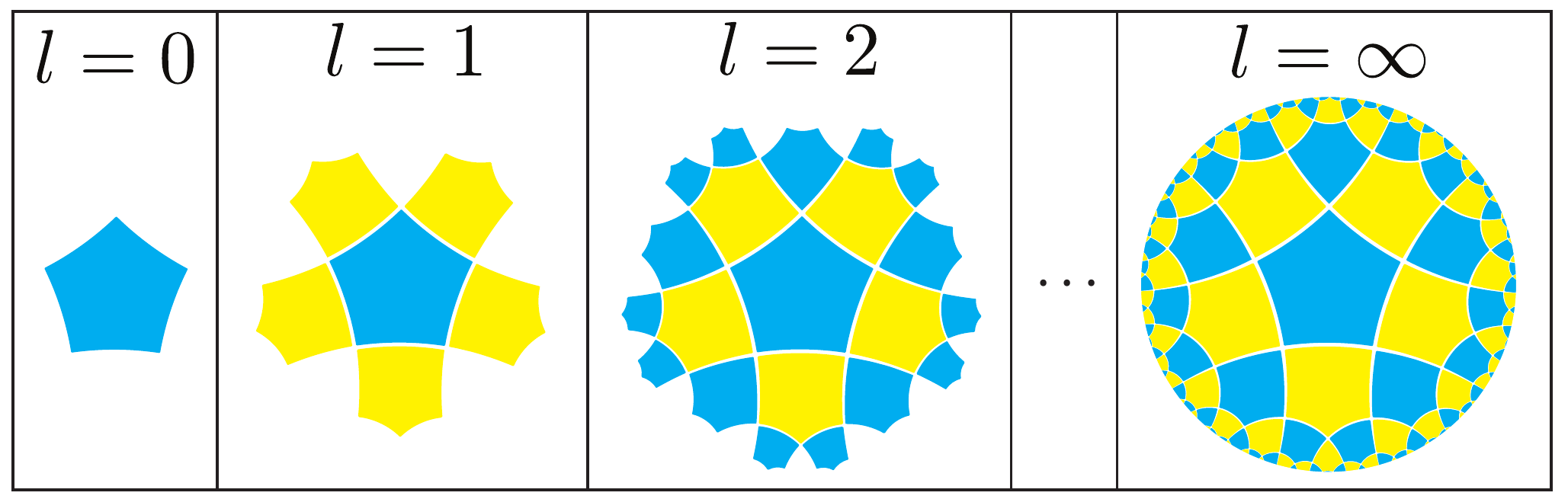}
    \caption{A tiling of the hyperbolic plane with a cutoff labeled by $l$. Increasing $l$ by one amounts to adding an additional ``layer'' of pentagons. }
  \label{fig: hyperbolic-tiling}
\end{figure}

To show that such a HaPPY code satisfies our definition, we first construct the corresponding planar graph with labelled dangling edges $(\Lambda_{\textrm{HaPPY}}=(V,E),\alpha)$.
Each input leg (e.g.\ those in red in \cref{fig:HaPPY-code}) is assigned a vertex labeled $1$, while each output leg (e.g.\ those in white in \cref{fig:HaPPY-code}) is assigned a vertex%
\footnote{Because we now assign output legs to vertices, but they originally arose as the faces of the tiling in \cref{fig: hyperbolic-tiling}, the graph associated with the Bulk system ends up being the dual of the tiling shown there.} 
 labeled $0$
.
Two vertices labeled $1$ are connected by an edge if the tensors of the legs they are assigned are contracted together.
A vertex labeled $1$ and a vertex labeled $0$ are connected by an edge if the legs they are assigned belong to the same tensor.
%In subsequent subsections, we will require directionality for the edges -- this can be chosen arbitrarily, but for ease of later calculations we take the convention that dangling edges point from $V_0$ to $V_1$ vertices.

In this case, the Bulk system is unconstrained, so $\mathcal{T}=(\mathcal{U},\id)$. In the $\mathcal{U}=(V_1\cup E,\{\mathcal{H}_{v\in V_1}\}\cup\{\mathcal{H}_{e\in E}\})$ system, for all $e\in E$ we have that $\mathcal{H}_e$ is trivial since no degrees of freedom are associated to edges in the HaPPY code. For $v\in V_1$ we have that $\mathcal{H}_v$ is taken to be the Hilbert space of the input leg that $v$ is assigned to, in this case a qubit. For the Boundary system, it is left only to specify $\tilde{\mathcal{H}}_v$ for $v\in V_0$, which is identified with the Hilbert space of the \emph{output} leg that $v$ is assigned to. The isometry $\mathcal{V}$ is the HaPPY code itself. \Cref{fig:flan-HaPPY-code} depicts $\mathcal{V}$ as a map from the Bulk to the Boundary while abstracting away the internal structure of the tensor network.

\begin{figure}[t]
  \centering
    \includegraphics[width=0.5\textwidth]{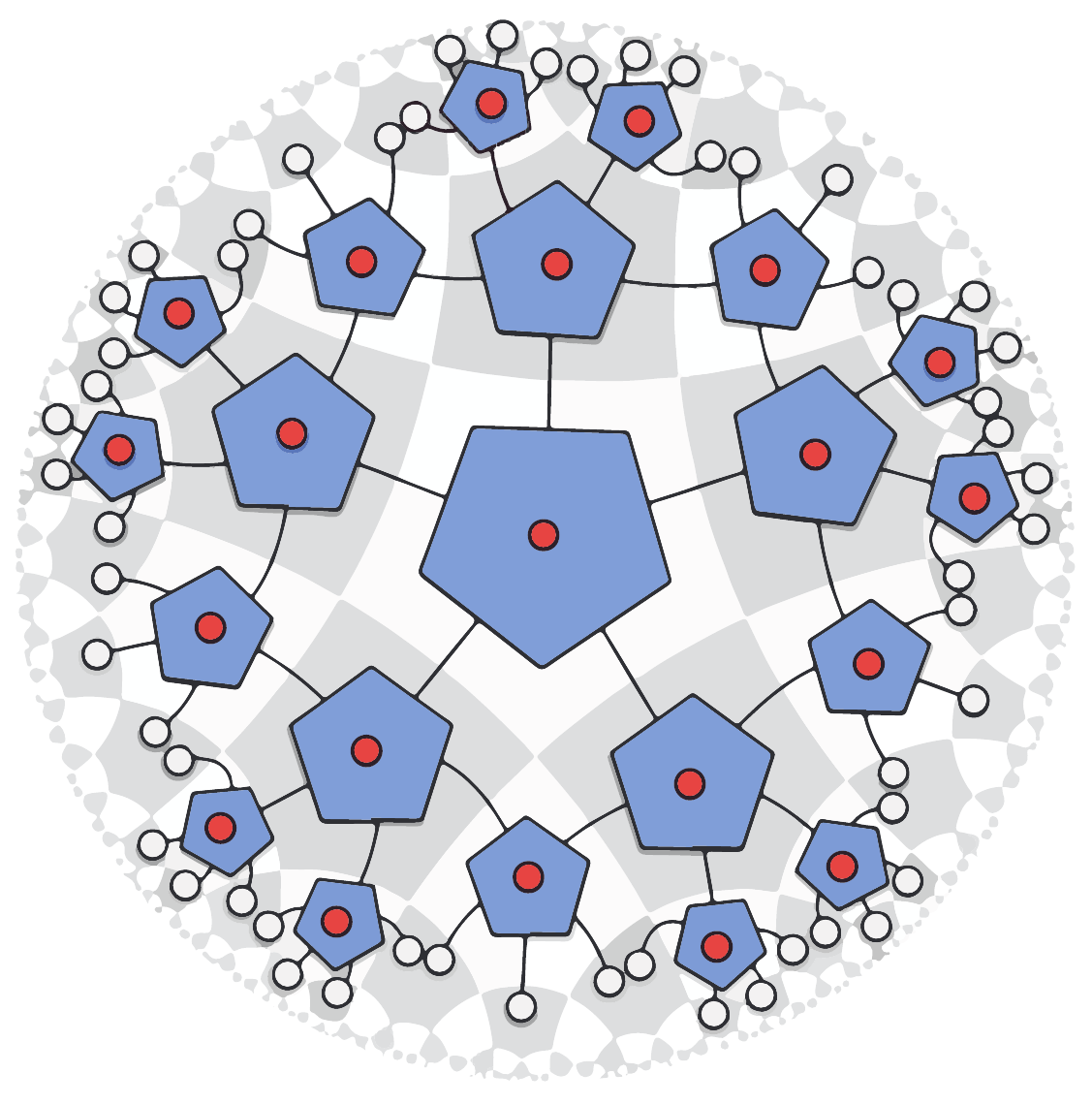}
    \caption{An example of a HaPPY code with cut off at $l=2$.
    A blue pentagon represents an individual tensor.
    Each tensor's six legs, each of which represents a single qubit Hilbert space, are represented by the red dot on top of it and the five black lines coming out of it. The red dot legs are treated as inputs, and the legs terminating in a white dot as outputs. Legs which terminate at tensors on both sides are refereed to as ``virtual legs''.
    (This and similar subsequent figures are adapted from those in Ref.~\cite{HaPPY}).} 
  \label{fig:HaPPY-code}
\end{figure}

\begin{figure}[t]
  \centering
    \includegraphics[width=0.7\textwidth]{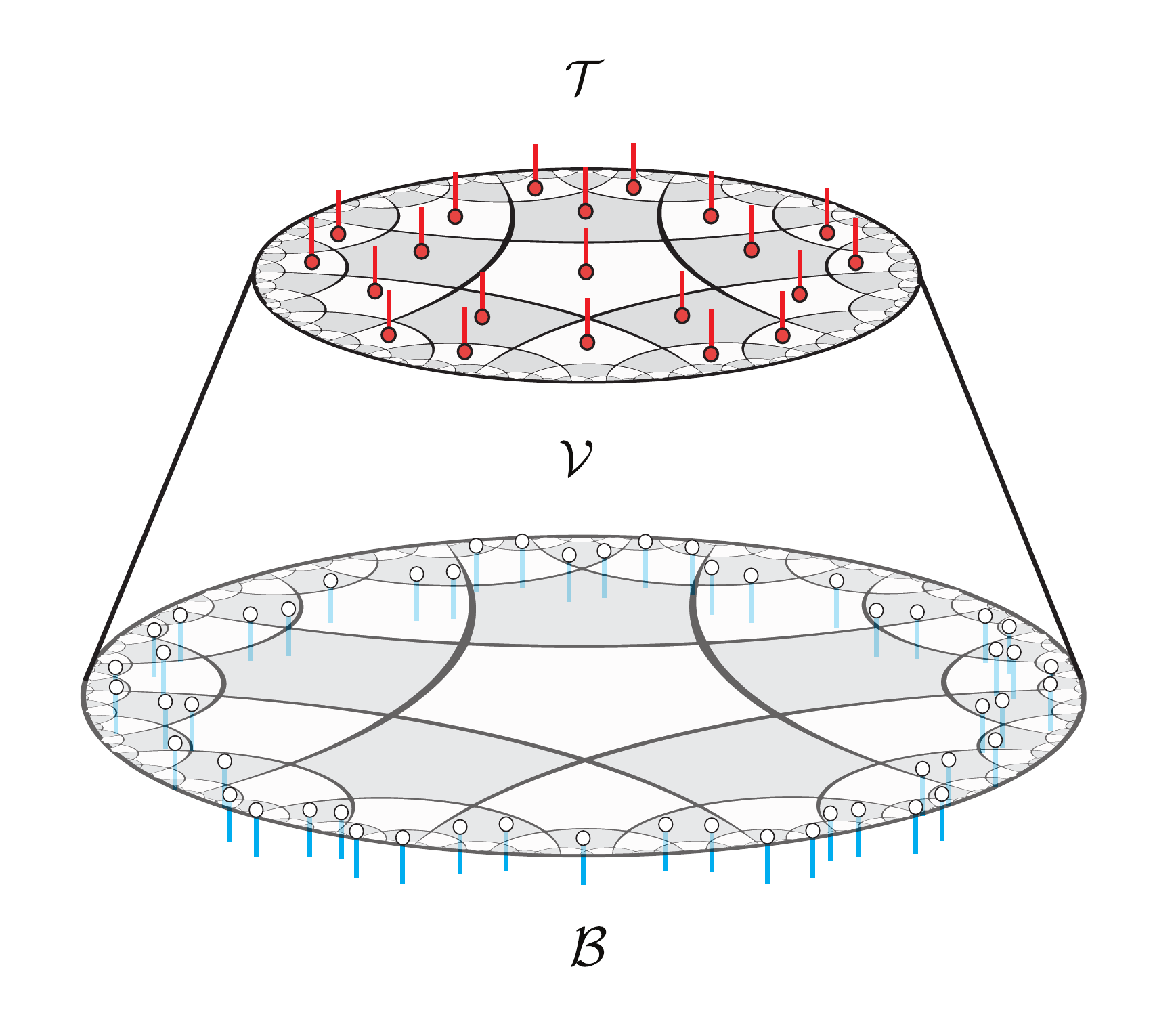}
    \caption{The isometry for the same HaPPY code as in \cref{fig:HaPPY-code} but without the internal structure of the tensor network. This pictorial depiction is useful in section~\ref{sec:applications-to-holography} when visualizing the composition of the HaPPY isometry with other maps. The red legs correspond to the Bulk/input legs, denoted by red circles in fig \ref{fig:HaPPY-code}, while the blue legs correspond to the Boundary/output legs, depicted there as white circles.}
  \label{fig:flan-HaPPY-code}
\end{figure}

The entanglement wedge $\mathcal{E}(R)\subseteq V_1\cup E$ of a Boundary region $R\subseteq V_0$ is given by the greedy entanglement wedge algorithm, defined in Ref.~\cite{HaPPY}, and which we review here in terms of the language we have defined.
To construct it, one initializes $\mathcal{E}(R)$ to contain just the dangling edges connecting to the Boundary region $R$.
We then iterate the following procedure.
For each vertex labeled 1 not included in the entanglement wedge, check the number of incident edges already included in the entanglement wedge.
If it is least half of its total number of incident edges (at least three in this case), then add that vertex and every non-dangling edge it participates in into $\mathcal{E}(R)$.
Repeat this process until none of the remaining vertices labeled 1 satisfy this criterion.
It is argued in Ref.~\cite{HaPPY} that the order of adding vertices does not matter, and so this gives a uniquely defined algorithm.
Furthermore, by the properties of perfect tensors, it is guaranteed that an operator acting in the entanglement wedge can be reconstructed on the starting Boundary region as a codespace-preserving operator (i.e.\ one commuting with the projector onto the image of the isometry), so $\mathcal{A}_{\mathcal{E}(R)}=\mathcal{A}_\mathcal{T}(\mathcal{E}(R))$.

This map qualifies as near-boundary probing, which we argue as follows.
When adding new vertices, only non-dangling edges are added; so the set of dangling edges in $\mathcal{E}(R)$ does not change from the initial conditions, thus the first point is satisfied.
Secondly, any included vertex must be connected to $R$ by some path of in $\mathcal{E}(R)\cup R$, which can be argued as follows: start at the vertex in question.
If it is connected to a dangling edge, we are done. Otherwise move to the neighboring vertex which was earliest to be included by the algorithm.
Since every included vertex either is connected to a dangling edge or neighbors a vertex that was included earlier, this eventually reaches a vertex connected by a dangling edge to $R$.
Finally, the third requirement is satisfied because the set of non-dangling edges included contains exactly those that connect to included vertices.

Furthermore, the ``complement'' map $R \mapsto \mathcal{E}(R^c)^c \cup \exterior{\mathcal{E}(R^c)}$ is also near-boundary probing.
The first requirement follows from the same reasoning as above (noting that $\exterior{\mathcal{E}(R^c)}$ does not contain any dangling edges).
The third requirement follows from the inclusion of $\exterior{\mathcal{E}(R^c)}$.
The second requirement is more subtle, but can be argued by contradiction.
Suppose it does not hold; then there is a site $v\in V_1$ that is excluded from the entanglement wedge of $R^c$ but has no path to the boundary contained in $\mathcal{E}(R^c)^c$.
Then $v$ is contained in some region $ S\subseteq \mathcal{E}(R^c)^c$ that is fully surrounded by vertices in $\mathcal{E}(R^c)$.
But any region $S$ of a hyperbolic pentagonal tiling that doesn't extend to the Boundary must contain at least one vertex with three edges on the boundary of $S$ due to the negative curvature of the tiling.
That vertex should have been included in the entanglement wedge $\mathcal{E}(R^c)$ by the greedy algorithm; thus we have a contradiction.

Besides satisfying the definition of a holographic code, the HaPPY code exhibits another feature: a global/global symmetry duality. A Bulk operator that is a product of $X$ operators on every input leg is implemented on the Boundary system by an operator that is a product of $X$ operators on every output leg.\footnote{This is also true of $Z$ operators and together they generate a unitary representation of a larger group (that is projective on each site). There are also further symmetries coming from Clifford operators, see Ref.~\cite{noncliffords}. Here we ignore these additional symmetries so we can deal with the simpler group $\mathds{Z}_2$.} 
%This is also true of $Z$ operators and the Clifford operator $K$, see Ref.~\cite{}, for simplicity we focus on the $\mathbb{Z}_2$ symmetry generated by $X$ here.} 
This can be seen from the fact that each perfect tensor itself has the same global/global duality with respect to its input and outgoing legs, or in other words each perfect tensor is stabilized by the action of $X$ on all legs (see Ref.~\cite{noncliffords} for a pictorial argument). This together with the identity operators make a $\mathds{Z}_2$ global/global symmetry duality.
This global symmetry is somewhat mysterious, because, as we discuss in \cref{subsec:connection-to-harlow-ooguri}, arguments from Ref.~\cite{harlow-ooguri} appear to imply it is impossible, and the resolution of this paradox in Ref.~\cite{Faist_2020} leaves an unexplained difference between AdS/CFT and the HaPPY code.
In the same section we discuss where this difference comes from.

% =====================

\subsection{Example with gauge constrained Bulk: the LOTE code}\label{subsubsec: holocodes-example-lote}

We now discuss an example of a holographic code with a $\mathbb{Z}_2$ gauge constrained Bulk, which we dub the ``gauged LOTE'' code\footnote{As we soon explain, it is modified from the original ``LOTE'' code~\cite{Marolf}, which we refer to by an acronym for the title of that paper.}. We give a summary of the code, show explicitly how it is constructed, and finally, show how it satisfies our definition of a holographic code. We then comment about the symmetry dualities it exhibits. In \cref{subsec: building_holocodes_with_arbitrary_dualities} we give a generalization to finite groups and an approximate version for compact Lie groups.

Despite its great success, the HaPPY code fails to capture many of the features of AdS/CFT.
One such feature is the fact that some Bulk degrees of freedom can be recovered on both a Boundary region and its complement.
This circumvents the no-cloning theorem because these degrees of freedom are classical, or more precisely only a central (i.e.\ abelian) algebra is associated to them.
To remedy this, Donnelley et.\ al.\ introduced in Ref.~\cite{Marolf} a modification of the HaPPY code which we refer to as the LOTE code, that 1) gives the edges non-trivial Hilbert spaces and 2) has the property that an edge sitting on the boundary of some entanglement wedge (i.e.\ only one of its vertices is included in the wedge) has a central sub-algebra that can be reconstructed on the corresponding Boundary region and (often)\footnote{This is not always true because the HaPPY code also fails to perfectly capture the ``complementary recovery'' property of AdS/CFT, which is $\mathcal{E}(R^c)=\mathcal{E}(R)^c$ for any choice of $R$. However for many choices of $R$ this still holds.} its complement as well. Although the Bulk system clearly has the structure of a pre-gauging Hilbert space, the authors stopped short of considering gauge symmetries. We now take a step further and impose gauge constraints on the Bulk, resulting in the gauged LOTE code. As we see in section~\ref{sec:applications-to-holography}, this code exhibits additional properties that illuminate aspects of symmetries in AdS/CFT and offer a new method to build covariant codes. 

First, we provide an explicit construction of the (ungauged) LOTE code introduced in Ref.~\cite{Marolf}.
%In what follows we will often say that we assigned an input or output leg to a vertex or edge, by which we mean we assign the Hilbert space corresponding to that leg.
 The first step is to ``stack'' 6 copies of the HaPPY code, and to use the same graph $\Lambda_\textrm{HaPPY}$ as defined in \ref{subsubsec: holocodes-example-happy}.
Initially, we identify the Hilbert spaces for each component of this graph to just be six copies of the corresponding Hilbert spaces from HaPPY -- each Bulk vertex $v\in V_1$ and Boundary vertex $v\in V_0$ is associated with Hilbert spaces consisting of six qubits, while edge Hilbert spaces are taken to be trivial for now.
%For the Bulk system, assign (tentatively) input legs in the stack to $V_1$ vertices by grouping together all input legs which line up ``vertically'', so that each code in the stack contributes one input leg per vertex.  For the boundary system, do the same but with output legs and $V_0$ vertices. Bulk $V_0$ vertices are again assigned trivial Hilbert spaces.
%The Hilbert spaces for the Bulk vertices -- i.e.\ those in $V_1$ -- will for now just be the collection of all six Hilbert spaces associated with of the  by grouping together all input legs which line up ``vertically'', so that each code in the stack contributes one input leg per vertex.  For the boundary system, do the same but with output legs and $V_0$ vertices. Bulk $V_0$ vertices are again assigned trivial Hilbert spaces.

To introduce nontrivial Hilbert spaces for each edge, we now shuffle some of these assignments around as follows.
For each Bulk vertex $v \in V_1$, and for each of its neighboring Bulk vertices $v' \in V_1$, one leg from each is combined to form an input leg associated with the corresponding edge $(v,v')$ or $(v',v)$.
This is done by contracting the two legs with the so-called ``copying tensor'' -- the conjugate of the isometry $|i\rangle\rightarrow |i\rangle|i\rangle, \ i\in \{+,-\}$, see \cref{fig:LOTE-construction}. %The subscript $X$ denotes that this is done in the $X$ basis.  
For Bulk vertices that do not connect to any dangling edges, only one of the six input legs remains after this procedure has assigned the other five to all neighboring edges.
How the legs are assigned is arbitrary -- for each vertex, the legs from any of the six layers can be arbitrarily distributed to the six assignments, namely five connections to adjacent copying tensors and one assignment to remain as a vertex leg. %\footnote{This results in significant non-uniqueness of the construction%, with the number of configurations scaling loosely as $ \sim (6!)^N$, with $N$ the number of Bulk vertices, and ignoring boundary effects.
%}.
%There will always be the same number of uncontracted input legs and dangling edges associated with such a vertex.

\begin{figure}
     \centering
     \begin{subfigure}[t]{0.49\textwidth}
         \centering
\includegraphics[width=0.9\textwidth]{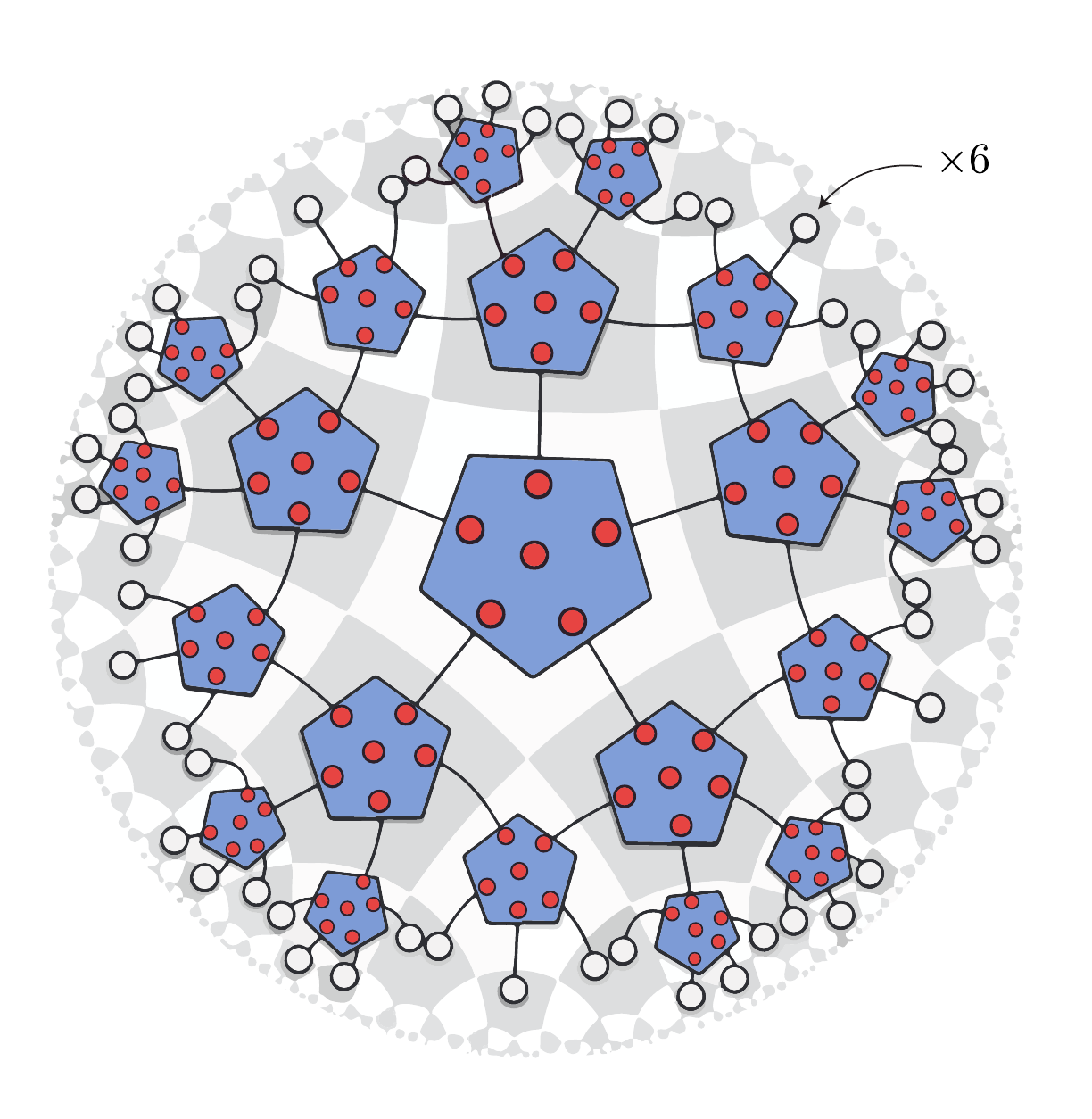}
         \caption{A stack of six HaPPY codes viewed from the top.
}
         \label{fig:stacked-HaPPY}
     \end{subfigure}
     \hfill
     \begin{subfigure}[t]{0.49\textwidth}
         \centering
         \includegraphics[width=\textwidth]{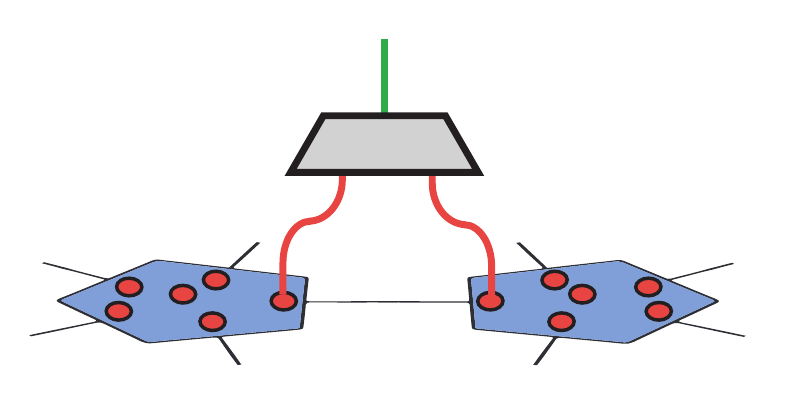}
         \caption{Logical legs of adjacent tensors are contracted by the ``copy tensor''.}
         \label{fig:copy-tensor}
     \end{subfigure}
     \hfill
     \begin{subfigure}[t]{0.49\textwidth}
         \centering
         \includegraphics[width=0.9\textwidth]{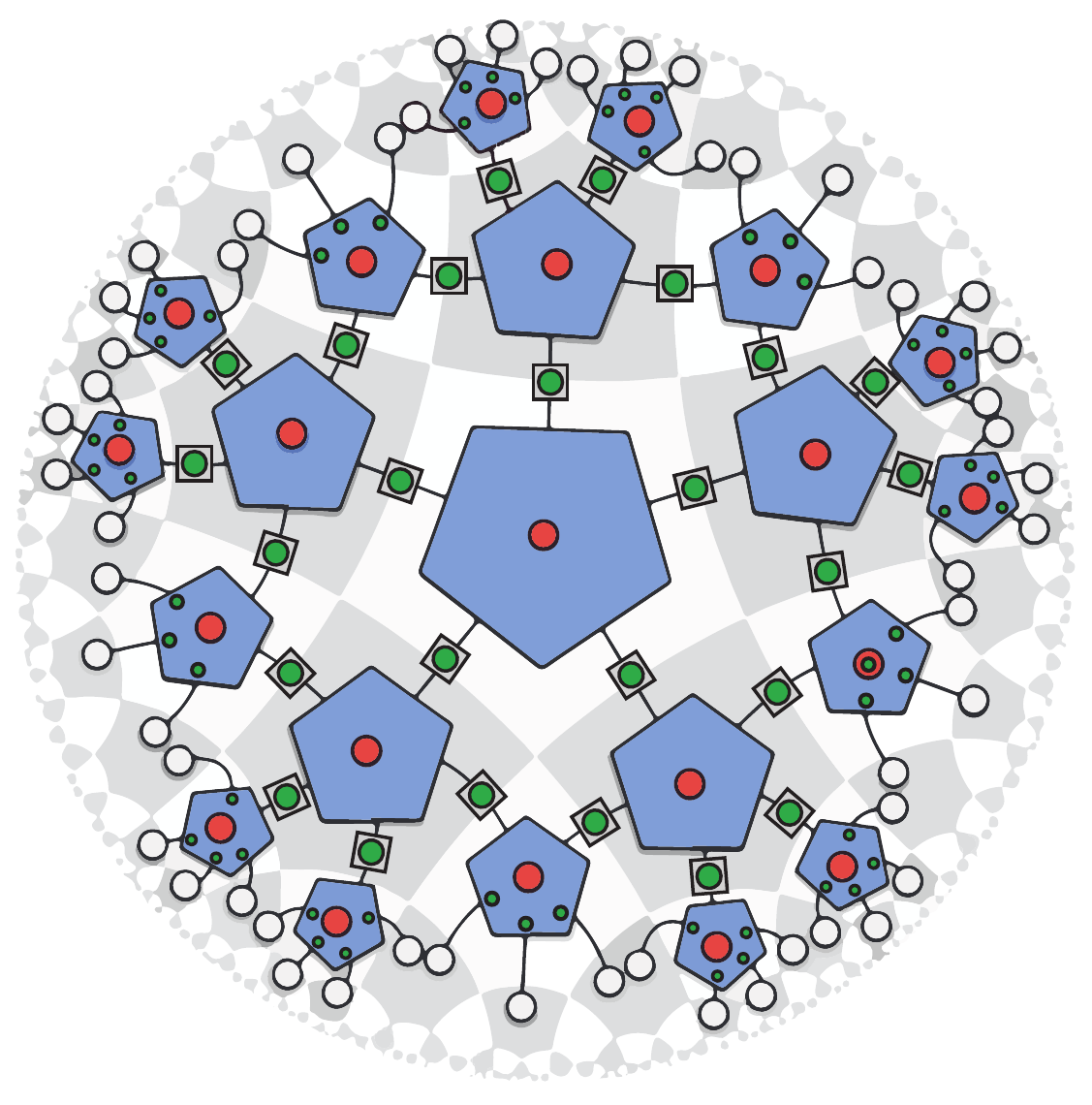}
         \caption{The LOTE code.}
         \label{fig:LOTE-final}
     \end{subfigure}
        \caption{Steps in the construction of the LOTE code. (a) Six HaPPY codes are stacked. The input legs are arranged to clarify how they are contracted.
        (b) Each pair of input legs from two adjacent tensors is contracted with the copy tensor to create the edge degree of freedom. (c) Final result: The LOTE code. Edge degrees of freedom are represented by green dots. The ones inside grey boxes came from the contractions. The small ones on the boundary tensors correspond to the dangling edge degrees of freedom, and come from  uncontracted input legs.}
        
        \label{fig:LOTE-construction}
\end{figure}

%``vertically'' on top of one another. Group virtual legs that align ``vertically'' into one leg. For 
% Our definition of holographic code requires that the entanglement wedge of a single vertex in $V_0$ should include the degree of freedom associated with the unique dangling edge in which it participates, and that all physical operators acting on degrees of freedom in the entanglement wedge are reconstructable.

The procedure is slightly more complicated for Bulk vertices $v\in V_1$ connecting to dangling edges. 
In HaPPY, such vertices have only one or two neighboring Bulk vertices -- see \cref{fig:HaPPY-code}.
After assigning one leg to be the input leg for $\mathcal{H}_v$, this leaves either three or four unassigned legs.
Instead of using the copying tensor, we associate these legs directly with the Hilbert spaces of the adjacent dangling edges, $\mathcal{H}_e$ with $e\in E_0$.
We now deviate from the construction in Ref.~\cite{Marolf} in a small and superficial way. In order to make the entanglement wedge the same as the one used for HaPPY, we want that a Bulk dangling edge $\mathcal{H}_e$ be reconstructable on the Boundary on $\mathcal{H}_{v_0}$ alone, with $v_0$ being the Boundary vertex connected to the dangling edge $e$.
However, in the code's current state, reconstruction of $\mathcal{H}_e$ on the Boundary must use the reconstruction properties of the tensor at the adjacent Bulk vertex to implement it on \emph{three} Boundary vertices (i.e.\ $v_0$ and two more).
This can be fixed using a minor reshuffling of the assignments of output legs to vertices in the Boundary system.
Specifically, for each vertex $v_0\in V_0$, consider the adjacent vertex $v\in V_1$ to which it is connected by the dangling edge $e = (v_0,v)$.
The input leg associated to $\mathcal{H}_e$ belongs to a unique layer in the stack, in that it was originally $\mathcal{H}_{v}$ for that layer of HaPPY.
Thus $\mathcal{H}_e$ can be reconstructed on the set of all output legs neighboring $v$ in that layer.
Thus, these legs should all be reassigned to $\tilde{\mathcal{H}}_{v_0}$, so that they can be used to reconstruct the input leg at $\mathcal{H}_e$ as we desire.
By doing this for all vertices in $V_0$ one can guarantee the desired reconstruction properties.
We stress that this reshuffling of labels is unimportant to the central points of the paper; if we did not do so then each dangling edge leg would just be reconstructed on $\mathcal{O}(1)$ Boundary vertices instead of a single one.
We choose to include it in order to align as closely as possible the algebra associated to an entanglement wedge, $\mathcal{A}_{\mathcal{E}(R)}$ with the algebra of physical operators supported on that wedge, $\mathcal{A}_{\mathcal{T}} (\mathcal{E}(R))$.
%In fact, in this case, the algebra $\mathcal{A}_{\mathcal{E}(R)}$ contains all operators supported on the \emph{interior} of the entanglement wedge, in addition to the binary 

% We choose to include it to simplify this aspect of our construction, but one could reach similar conclusions without it%
% \footnote{Another way to see the effects of this reshuffling is that it allows the \emph{gauged} LOTE code, to be introduced shortly, to satisfy the definition of holographic codes as given in \cref{def:holographic-code}.
% Without reshuffling, regions $R$ consisting of a single boundary vertex would violate the second part of requirement 2, as the incident edges would not be able to be included in $\mathcal{E}(R)$.
% One could relax our definition of holographic codes to allow for such a case, but this would complicate the construction.
% }.

This defines the ungauged/original LOTE code.
%\footnote{Note that this code does not technically satisfy our definition of a holographic code due to the superficial issues with the entanglement wedge reconstruction map discussed in footnote \ref{footnote: more-general-holocode-def}, but we only need the gauged version we are about to define, which does satisfy the definition.}
We now argue that the Bulk of the ungauged LOTE code has the structure of a $\mathds{Z}_2$ pre-gauging system as defined in \cref{def:pre-gauging}. It has a graph, $\Lambda_{\textrm{HaPPY}}$, Hilbert spaces associated with vertex degrees of freedom (via the input legs), and Hilbert spaces associated to each edge that are isomorphic to $L^2(\mathds{Z}_2)$ (i.e.\ they have dimension 2).
% For each dangling edge, the associated Hilbert space $\mathcal{H}_e$ is one of the qubits associated with the vertex that it connects to, which has not been merged into a copying tensor (i.e.\ the outermost red legs in \cref{fig:pre-lote})\ssc{Fix this once figure is fixed}.
% For each vertex $v\in V$, the Hilbert space $\mathcal{H}_v$ consists of the single qubit that remains unconnected to any of the three-legged copying tensors described above, and is not already associated with one of the dangling edges.
% For interior vertices, this is just a single qubit, but it can be multiple qubits for vertices connected to dangling edges.
% For each edge $e\in E$, the Hilbert space $\mathcal{H}_e$ is the single qubit associated with the tensor leg attached to that edge (the green legs in \cref{fig:pre-lote}).
% Since we want to treat the system as pre-gauged, the edge systems must be isomorphic to $L^2(G)$ for some $G$. Since we are dealing with qubits, we must have $G =\mathbb{Z}_2$.
All that remains is to equip the vertex degrees of freedom with a global symmetry transformation. This can be any unitary representation of $\mathds{Z}_2$, i.e. $U_v({0\in \mathbb{Z}_2}) = \id$, and $U_v({1 \in \mathbb{Z}_2})$ any unitary squaring to $\id$.

Having constructed the ungauged LOTE code and shown that its Bulk system has the structure of a $\mathds{Z}_2$ pre-gauging Hilbert space, we are ready to define the \emph{gauged} LOTE code. 
For the logical constrained system we take $\mathcal{T}=(\mathcal{U},\Pi_{GI})$ with graph $\Lambda_{\textrm{HaPPY}}$, with $\mathcal{U}$ the Bulk system of the ungauged LOTE code described above, and with $\Pi_{GI}$ the projector onto the gauge-invariant Hilbert space as defined in \cref{def: gauging-definition}.
The construction of that projector requires an oriented graph, whereas our definition of holographic codes did not require directions to be assigned to each edge; this orientation can be chosen arbitrarily so long as dangling edges $e\in E_0$ point from the Boundary to the Bulk, i.e.\ take the form $e=(v_0,v_1)$ with $v_0\in V_0 $, $v_1 \in V_1$.
This gauge constraint is essentially the only difference from the original LOTE code: we use the same Boundary system; we use the same isometry but with the domain restricted to gauge-invariant states; and we use the same entanglement wedge map as defined for the HaPPY code in the previous section.

% It is the constrained system with projector given by $\Pi_{GI} 
% %=\prod_{v\in V_{1}}\Pi_{v}$
% =\prod_{v\in V_{1}} \int dg\ A_v(g)$, with each $A_v(g)$ defined as in \cref{eq:Avg}. \kd{Say in particular that, for example the new encoding isometry is the encoding isometry from LOTE and adding the GI projector to the input.}

%As far as defining the code, it remains only to define the constraint $\Pi$ on the Bulk system $\mathcal{T}=(\mathcal{U},\Pi)$.

Since the entanglement wedge map is unchanged, it is still near-boundary probing.
%We must now determine the entanglement wedge algebras. 
%As such algebras consists of physical operators, they must all be gauge-invariant, and they may now have some support on the edge degrees of freedom unlike in HaPPY.
For now, we choose the entanglement wedge algebra to be generated by the entire algebra of physical (i.e.\ gauge-invariant) operators supported on the interior of the entanglement wedge (as it must be), as well as the \emph{central flux operators} on the boundary, that is, those of the form $U_e^L(g)$ with $g$ in the centre of $G$ and $e$ on the boundary of $\mathcal{E}(R)$.
Note that this includes any gauge-invariant operators supported on $\mathcal{E}(R)$ that can be generated from Wilson loops, NGC-to-NGC Wilson lines, NGC to charge Wilson lines, and central flux operators.
We also conjecture that the entanglement wedge algebra can be chosen instead to simply be the entire algebra of physical operators supported on $\mathcal{E}(R)$.

First, we outline the reconstruction of arbitrary operators in the interior of $\mathcal{E}(R)$, before discussing central flux operators on its boundary.
Asymptotic gauge transformations are clearly recoverable just on the $V_0$ whose edge they act on, because of the superficial degree of freedom shuffling we did above.
For larger boundary regions, reconstruction of $E_1$ edge degrees of freedom makes use of the fact that an isometry such as the copy tensor can always implement an operator $\mathcal{O}$ on its input system via an operator $V\mathcal{O}V^\dagger$ on its output system.
Thus components of gauge-invariant operators supported on inputs to copying tensors can be reconstructed as operators on the output legs.
Then these output legs connect to the tensors located at the two adjacent vertices, so the operators can again be reconstructed using properties of perfect tensors.
If the initial edge leg and both its adjacent vertices are in the entanglement wedge, then this allows for reconstruction of the operator on $R$.
With this reasoning we can conclude that the the entanglement wedge algebras contain all gauge invariant operators supported on the interior of the entanglement wedge.

For operators supported on the boundary of $\mathcal{E}(R)$, this procedure would in general result in an implementation with support on a vertex just outside the entanglement wedge.
However, the only obvious non-trivial gauge-invariant operator supported on these boundary edges can also be reconstructed, namely the central flux operators $U^L_e(1\in \mathds{Z}_2)=U^R_e(1\in \mathds{Z}_2)=X$. 
This is because although the implementation defined by $\mathcal{O} \to V\mathcal{O}V^{\dagger}$ generally has support on both output legs, in special cases one can exploit redundancy in this reconstruction to implement an input operator via an output operator localized to one subsystem.
Because of the way that we constructed the copy tensor, for example, an $X$ operator on the input can be implemented via either $X\otimes \id$ or $\id \otimes X$ on the output.
Thus the $X$ operator defined above can always be chosen to be implemented on the neighboring vertex that lies \emph{inside} of the entanglement wedge.

As a final note, we would like to stress that the gauge symmetry exhibited by the gauged LOTE code is in no way related to the global/global symmetry dualities exhibited by HaPPY. As we see in the following subsection, this code does indeed display a $\mathds{Z}^2$ gauge/global duality, but this is a coincidence; at no point do we rely on the global/global duality in HaPPY. 
In fact, we could have constructed the gauge/global duality using a different idempotent unitary operator to emphasize this fact, but we chose the $X$ operator (and associated copying tensor) for simplicity.
Later, we also show how to generalize this construction to build Bulk systems gauged with respect to arbitrary finite groups, and these are certainly unrelated to the HaPPY $\mathds{Z}_2$ global/global duality. 

\subsection{Full gauge-invariant Bulk \texorpdfstring{$\implies$}{→} gauge/global duality}
\label{subsec: full_gauged_Bulk_to_gauge/global}

%In \cref{thm:global_global_to_gauge_global} and \cref{thm:gauge_global_to_global_global} it was necessary to assume that the given holographic code already exhibits a global/global or gauge/global duality.
We now show that if the Bulk system of a holographic code has a full gauge-invariant Hilbert space, i.e.\ $\Pi=\Pi_{GI}$\footnote{In contrast to a holographic code constrained further to a fixed flux sector, such as the one we construct in \cref{subsec: gauging-holocodes}.}, then that code automatically exhibits a gauge/global duality.
We only require the bare-bones assumption that local asymptotic gauge transformations are recoverable, i.e.\ that the entanglement wedge algebra $\mathcal{A}_{\mathcal{E}(R)}$ of a single-site Boundary region $R = \{v_0\}$ includes local operators on the adjacent edge $e \ni v_0$.
This proof is adapted from a similar proof given in \cite{harlow-ooguri} for AdS/CFT.
We give an intuitive sketch of the proof, and then give a formal statement and proof.

\begin{figure}[t]
  \centering
    \includegraphics[width=0.6\textwidth]{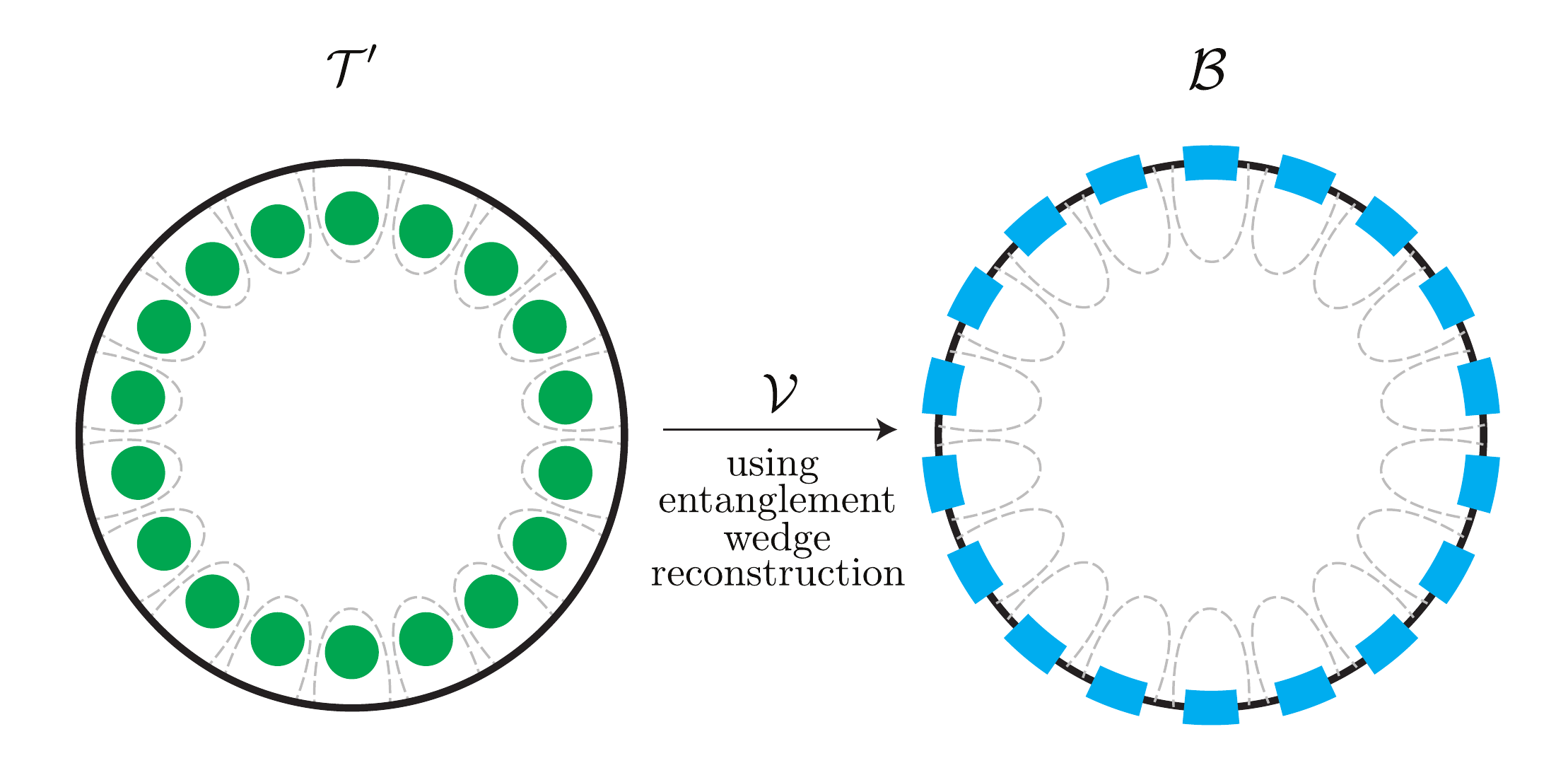}
    \caption{Depiction of the main part of the argument that holographic codes with full gauge-invariant Bulks automatically exhibit a gauge/global duality. The green dots represent the operators constituting the asymptotic symmetry transformation on the $\mathcal{T}$ system. They are each in an entanglement wedge of a small Boundary region and thus can be reconstructed as a product of operators along the $\mathcal{T}'$ system, depicted by the small blue rectangles, forming a global symmetry transformation. This part of the proof appears also in the proof of the same duality in AdS/CFT as shown in Ref.~\cite{harlow-ooguri}.}
  \label{fig:full-gauge-to-global-global}
\end{figure}

The sketch of the proof proceeds as follows (see also \cref{fig:full-gauge-to-global-global}).
Consider an asymptotic symmetry transformation in the Bulk.
It consists of a product of local asymptotic gauge transformations, which act only all dangling edges, as there are no Bulk degrees of freedom associated with vertices in $V_0$.
Because the Bulk has a full gauge-invariant Hilbert space, each local asymptotic gauge transformation in the product is also a physical operator, as it commutes with each local gauge constraint\footnote{This would not be the case if the Bulk system were further constrained to a fixed-flux sector, since restricted asymptotic symmetry transformations generically mix flux sectors}.
We now need to make the mild assumption that each of these operators can be reconstructed on the Boundary vertex to which the edge it lives on is attached.
Reconstructing all of these operators in this way gives a product of operators along single site regions of the Boundary system. It can be shown that these operators are unitary and form a representation of the same group, and thus the code exhibits a gauge/global duality. This is an adaptation of the proof strategy employed in Ref.~\cite{harlow-ooguri} to show that in AdS/CFT a Bulk long-range gauge symmetry implies a Boundary global symmetry. Other than explicit realization in a toy model, the difference in the proofs is how we show that the reconstructed operators can be chosen to be unitary representations.

\begin{theorem}
\label{thm: full_gauge_to_gauge_global}
Suppose there exists a holographic code $(\mathcal{T},\Lambda,\alpha,\mathcal{B},\mathcal{E},\{\mathcal{A}_{\mathcal{E}(R)}\},\mathcal{V})$ such that $\mathcal{T}$ is a constrained system obtained by gauging a system with global symmetry $(\mathcal{S},G,\{U_v\})$, and $ A_{v_0}(g) \in \mathcal{A}_{ \mathcal{E}(\{v_0\})}$.
It is possible to equip the Boundary system $\mathcal{B}$ with a global symmetry $(\mathcal{B},G,\{\tilde{U}_{v\in V_0}\})$ such that $\mathcal{V}$ exhibits a gauge/global duality between $\mathcal{T}$ and $(\mathcal{B},G,\{\tilde{U}_{v\in V_0}\})$. 
\end{theorem}
\begin{proof}
Recall that $A_{V_0}(g)\equiv\prod_{v_0\in V_0} A_{v_0}(g)$ and that $A_{v_0}(g)\in \mathcal{A}_{\mathcal{T}}(\{v_0,e\})$\footnote{Notice this would not be true if $\mathcal{T}$ were only a flux-free sector since an asymptotic symmetry transformation restricted to act on only a subset of $V_0$ generically changes the flux sector.}, with $e$ the unique dangling edge incident to $v_0$.
By the definition of a holographic code, since $ A_{v_0}(g) \in \mathcal{A}_{ \mathcal{E}(\{v_0\})}$, there exists a codespace-preserving operator $\tilde{U}_{v_0}(g)\in\mathcal{A}_\mathcal{B}(v_0)$ such that $\tilde{U}_{v_0}(g)\mathcal{V}=\mathcal{V}A_{v_0}(g)$.
By \cref{lem:logical-unitary-to-physical-unitary}, $\tilde{U}_{v_0}(g)$ can be chosen to be unitary and by \cref{lem: logical-representation to physical representation} it can also be chosen to be a representation.
Thus $\{\tilde{U}_{v_0}(g)\}$ is a global symmetry transformation of the type whose existence was to be shown. 
\end{proof}

\section{The gauging isometry applied to holographic codes}
\label{sec:applications-to-holography}
Having defined and examined the gauging isometry, we turn to its application in constructing, and mapping between, holographic codes with various symmetry dualities. We show with a simple application of the gauging isometry that any holographic code with a global/global duality can be made into one with a gauge/global duality, and vice versa. We then use these results to make several observations about holographic codes and AdS/CFT. First, we discuss connections between this work and Ref.~\cite{harlow-ooguri}, in particular we use our results to gain a better understanding of why holographic codes seem to violate their no global symmetries result, and point out a sufficient condition for holographic codes to exhibit gauge symmetry. We then argue that our toy models can be seen as a rudimentary model of time evolution, and can be used to illustrate the relationship between metric fluctuations and approximate error correction in AdS/CFT. Finally, we explicitly construct a holographic code with a global/global duality for an arbitrary finite symmetry group, as well as an approximate version for any compact Lie groups.

\subsection{Gauging: global/global \texorpdfstring{$\implies$}{→} fixed sector gauge/global}
\label{subsec: gauging-holocodes}

We now show that given a holographic code with a global/global symmetry duality, one may ``gauge'' its Bulk via the gauging map to obtain a holographic code with a fixed sector gauge/global symmetry duality with the same symmetry group. We first give an intuitive sketch of the argument in the next paragraph and \cref{fig:gauging-collective}. We then give a formal proof in \cref{thm:global_global_to_gauge_global}. Finally we give an example using the HaPPY code in \cref{example: gauging-HaPPY}.

The construction of the gauge/global code is accomplished by composing the inverse of the gauging map with the encoding isometry of the given holographic code, see \cref{fig:gauging-flan}. The resulting map inherits the entanglement wedge recovery properties of the original encoding isometry because all ungauged operators in an entanglement wedge may be implemented by gauged ones supported strictly within same entanglement wedge, see \cref{fig:gauging-ewedge-preserved}. Lastly, the resulting map exhibits a gauge/global symmetry duality because the inverse of the gauging map ``converts'' a Bulk global symmetry into an asymptotic symmetry transformation using its own gauge/global duality, see \cref{fig:gauging-symmetry-conversion}.
We present this more formally now.

\begin{figure}
     \centering
     \begin{subfigure}[t]{0.46\textwidth}
         \centering
\includegraphics[width=\textwidth]{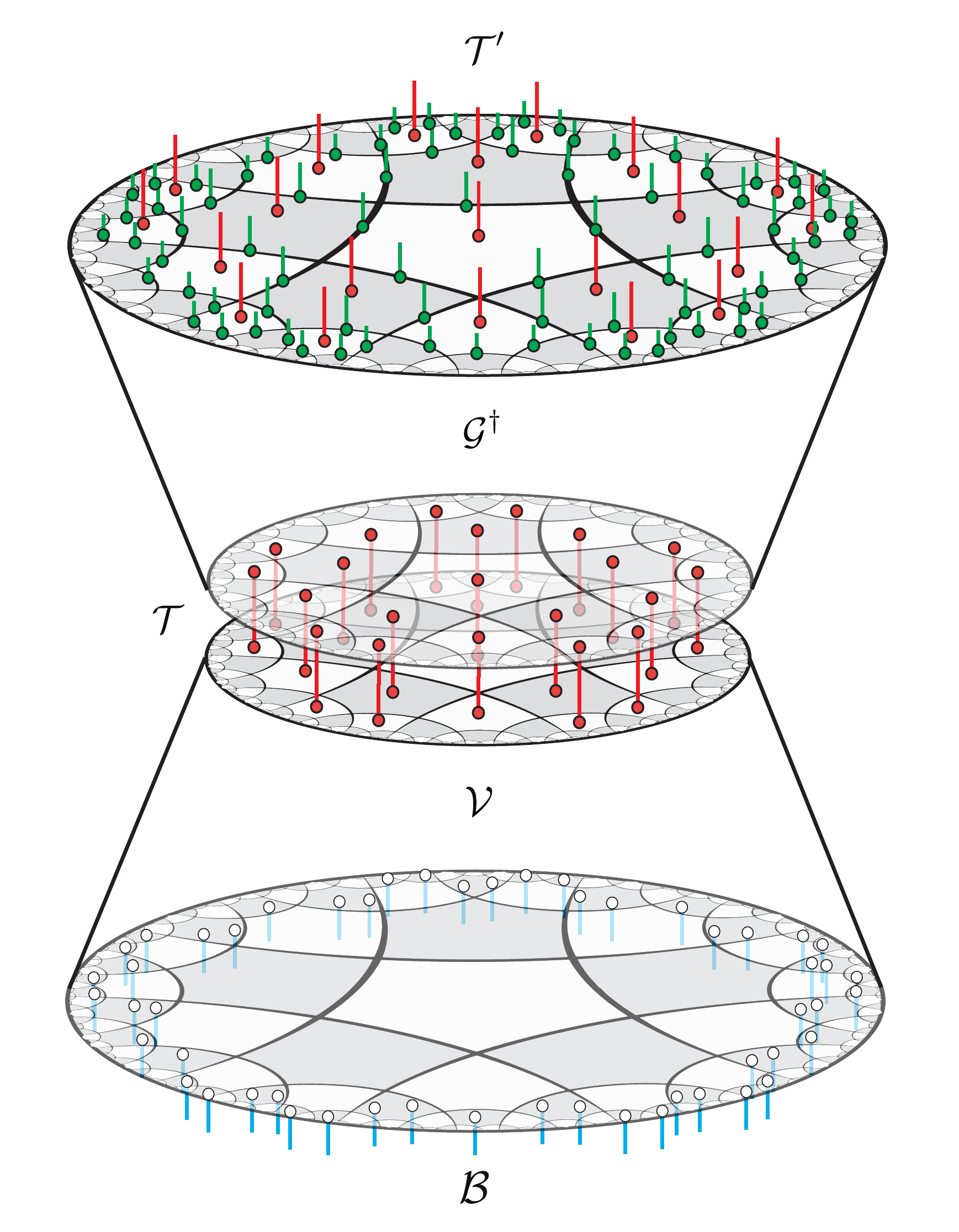}
         \caption{Depiction of $\mathcal{V} \circ \mathcal{G}^\dagger $, the encoding isometry of the code given by \cref{thm:global_global_to_gauge_global}.}
         \label{fig:gauging-flan}1
     \end{subfigure}
     \hfill
     \begin{subfigure}[t]{0.23\textwidth}
         \centering
         \includegraphics[width=\textwidth]{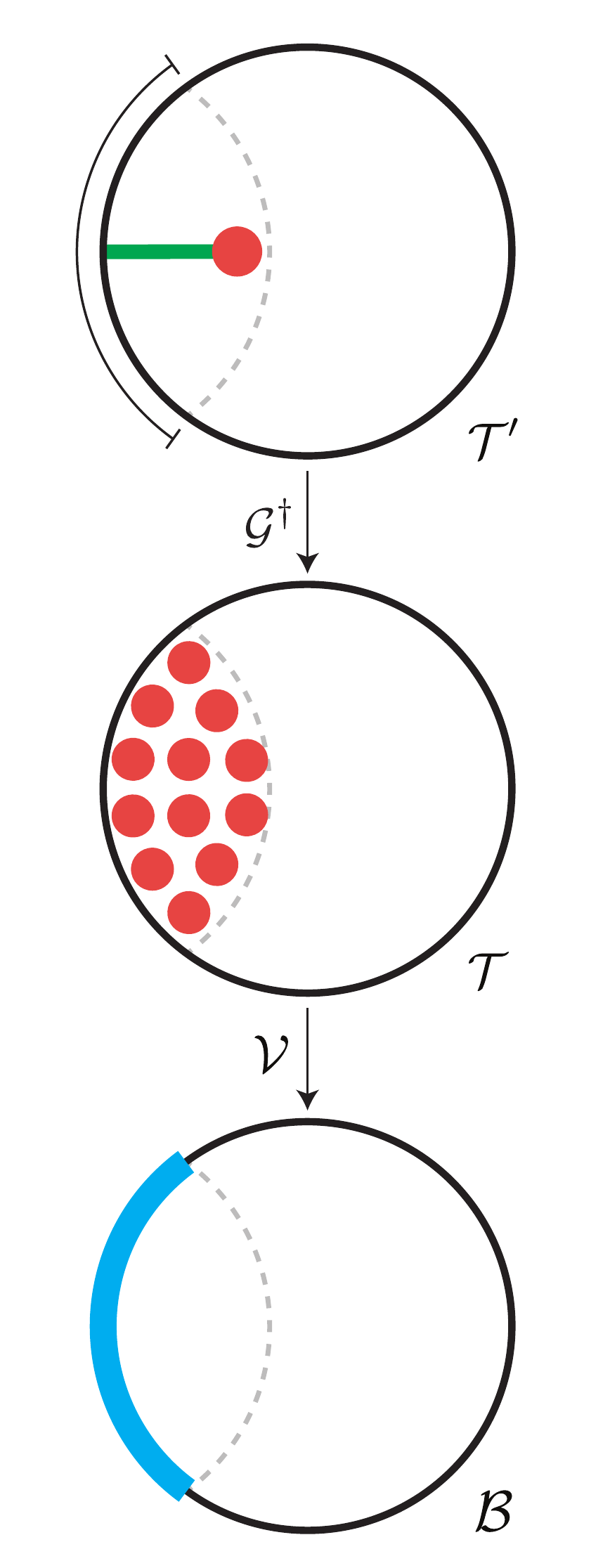}
         \caption{$\mathcal{V} \circ \mathcal{G}^\dagger $ inherits holographic reconstruction from $\mathcal{V}$.}
         \label{fig:gauging-ewedge-preserved}
     \end{subfigure}
     \hfill
     \begin{subfigure}[t]{0.23\textwidth}
         \centering
         \includegraphics[width=\textwidth]{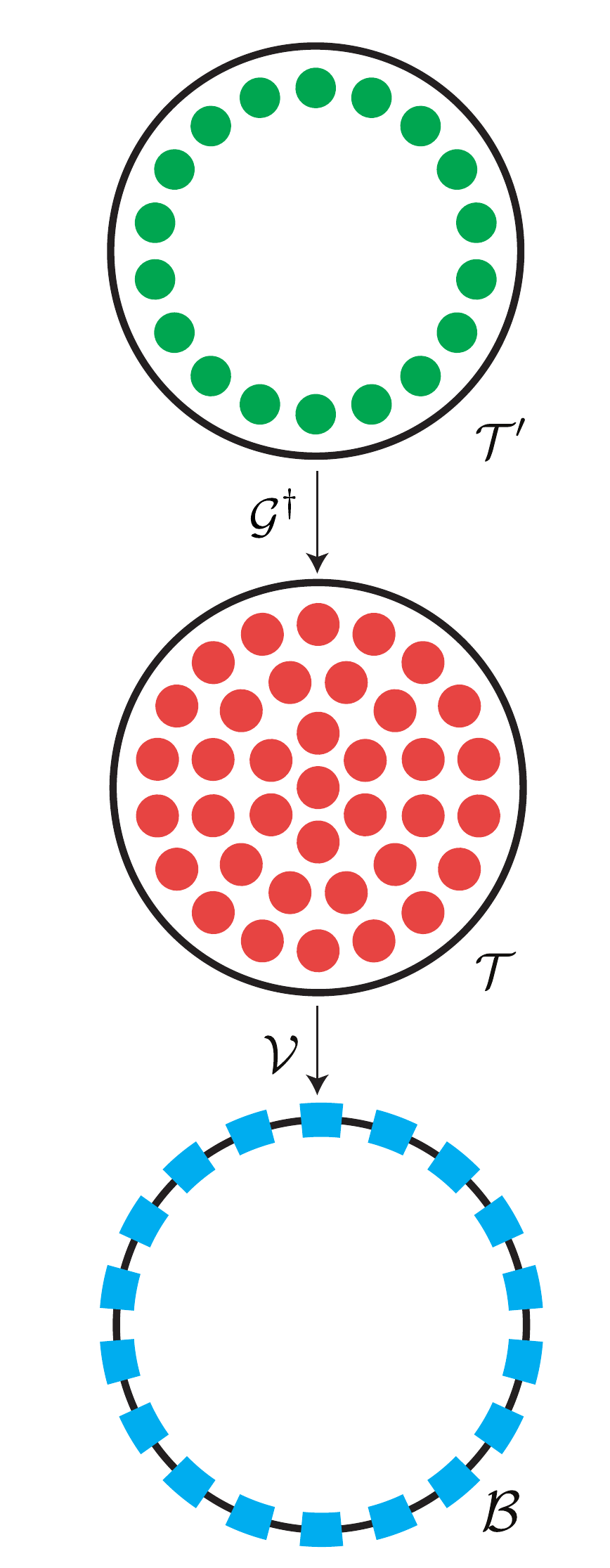}
         \caption{$\mathcal{V} \circ \mathcal{G}^\dagger $ exhibits a gauge/global duality.}
         \label{fig:gauging-symmetry-conversion}
     \end{subfigure}
        \caption{Basic components of the proof that any holographic code with a global/global symmetry duality can be converted into one with a gauge/global symmetry duality via \cref{thm:global_global_to_gauge_global}. (a) The encoding isometry of the new code is constructed by composing the the Hermitian conjugate of the gauging map with the encoding isometry of the given code, i.e. $\mathcal{V}'=\mathcal{G}^\dagger\circ\mathcal{V}$. Each map is depicted in tensor network notation with their legs arranged to reflect the locality of the degrees of freedom they represent. The composition is depicted by the contraction of legs in the $\mathcal{T}$ system. Red lines represent Bulk vertex degrees of freedom, green lines represent Bulk edge degrees of freedom, and blue lines represent the Boundary degrees of freedom. (b) $\mathcal{V}\circ\mathcal{G}^\dagger$ inherits the code properties of $\mathcal{V}$ because $\mathcal{T}'$ operators (red dot with green line) in the entanglement wedge of a region can be implemented by $\mathcal{T}$ operators (many red dots left of dashed line) which have support in the same wedge, which in turn can be reconstructed on the Boundary region (blue arc) using the original code. The thin black line shows the Boundary region and the dashed line its entanglement wedge. (c) $\mathcal{V}\circ\mathcal{G}^\dagger$ exhibits a gauge/global duality because an asymptotic symmetry transformation (green dots) on $\mathcal{T}'$ is implemented as a global symmetry transformation (red dots) on $\mathcal{T}$, which is in turn implemented by a global symmetry transformation on the Boundary (blue dots).} 
        \label{fig:gauging-collective}
\end{figure}

\begin{theorem}
\label{thm:global_global_to_gauge_global}

\item Consider a holographic code $(\mathcal{T},\Lambda,\alpha,\mathcal{B},\mathcal{E},\{\mathcal{A}_{\mathcal{E}(R)}\},\mathcal{V})$ as per \cref{def:holographic-code}, with unconstrained logical Bulk system and trivial edge degrees of freedom -- that is, with $\mathcal{T}=(\mathcal{U},\mathds{1})$ and $\mathcal{U}=(V\cup E,\{\mathcal{H}_{v\in V}\}\cup\{\mathcal{H}_{e\in E}=\mathds{C}\})$.
Suppose that we can equip $\mathcal{T}$ and $\mathcal{B}$ with global symmetry transformations $(\mathcal{T}, G, U_{v\in V})$ and $(\mathcal{B}, G, \tilde{U}_{v\in V_0})$, as per \cref{def:global-symmetry-transformation}, such that with respect to these transformations $\mathcal{V}$ exhibits a global/global duality as in \cref{def:global-global}.

Then there exists a holographic code $(\mathcal{T}',\Lambda,\alpha,\mathcal{B},\mathcal{E},\{\mathcal{A}'_{\mathcal{E}(R)}\},\mathcal{V}')$ whose encoding isometry $\mathcal{V}'$ exhibits a gauge/global duality, and with logical Bulk system $\mathcal{T}'=(\mathcal{U}',\Pi_{FF})$ such that $\mathcal{T}'$ is the flux-free sector of the system obtained by gauging $(\mathcal{T}, G, U_{v\in V})$.
\end{theorem}

\begin{proof}
The proof is carried out by constructing the promised holographic code from the given one. 
Let $\Pi_{GI}\mathcal{H}'_\Lambda$ denote the gauge-invariant Hilbert space of the system obtained by gauging $(\mathcal{T}, G, U_{v\in V})$ as per \cref{def: gauging-definition}, and let $\mathcal{G}:\mathcal{H}_\Lambda\rightarrow\Pi_{GI}\mathcal{H}'_\Lambda$ be the gauging map from \cref{def: gauging-map}.
Since $\mathcal{G}$ is invertible on the flux-free sector via \cref{thm: gauging-map-is-isometry,thm: gauging-map-image-is-flux-free}, the map $\mathcal{V}'\equiv\mathcal{V}\circ\mathcal{G}^{\dagger}$ isometrically encodes the Hilbert space of $\mathcal{T}'$ into that of $\mathcal{B}$.
Furthermore, this map exhibits a global/gauge duality as per \cref{def:gauge-global} since $\mathcal{V}\mathcal{G}^{\dagger}A_{V_0}(g)=\mathcal{V}U_V(g)\mathcal{G}^{\dagger}=\tilde{U}_{V_0}(g)\mathcal{V}\mathcal{G}^{\dagger}$.
It remains only to show that we can choose $\mathcal{A}'_{\mathcal{E}(R)}$ to include $\mathcal{A}_{\mathcal{T}'}(\interior{\mathcal{E}(R)})$, i.e.\ any operator $O'\in\mathcal{A}_{\mathcal{T}'}(\interior{\mathcal{E}(R)})$ can be reconstructed on $R\subseteq  V_0$.
Since $O'$ is supported on $\interior{\mathcal{E}(R)}$, we can use \cref{thm: undressing-local-ops} with $\Gamma = \interior{\mathcal{E}(R)}$, to conclude that there exists an operator $O$ supported on the vertices of $\mathcal{E}(R)$ such that $\mathcal{G}^{\dagger}O'=O\mathcal{G}^{\dagger}$. 
The fact that $\Gamma$ satisfies the assumptions of \cref{thm: undressing-local-ops} follows from our requirement that $R \mapsto \mathcal{E}(R^c)^c$ is near-boundary probing, applied to $R^c$; namely, it guarantees that after removal of $\mathcal{E}(R)$, all Bulk vertices remain path-connected to sites in $R^c$.
Now, by definition of a holographic code there must exist $\tilde{O}\in\mathcal{A}_{\mathcal{B}}(R)$ that preserves the image of $\mathcal{V}$ such that $\mathcal{V}O=\tilde{O}\mathcal{V}$.
Thus $\mathcal{V}\circ\mathcal{G}^{\dagger}O'=\tilde{O}\mathcal{V}\circ\mathcal{G}^{\dagger}$.
Finally, it preserves the image of $\mathcal{V}^\prime$ since $\mathcal{G}^\dagger$ is surjective.

% It remains only to show that any operator in $O'\in\mathcal{A}_{\mathcal{T}'}(\mathcal{E}(R))$ can be reconstructed on $R\subseteq  V_0$. Consider the operator $O\in\mathcal{A}_\mathcal{T}(\Lambda)$ defined as $O\equiv\mathcal{G}^{\dagger}O'\mathcal{G}$. Suppose, for the sake of contradiction, that $O\not\in\mathcal{A}_\mathcal{T}(\mathcal{E}(R))$, i.e. it has some support outside $\mathcal{E}(R)$. Consider an operator $Q\in\mathcal{A}_\mathcal{T}(\Lambda\setminus\mathcal{E}(R))$ i.e. whose support is completely outside $\mathcal{E}(R)$, and which does not commute with $O$. Such an operator always exists since $Q$ may have overlapping support with $O$ and $\mathcal{T}$ is unconstrained. By \cref{thm: operator-dressing-theorem}, there exists an operator $Q' \in \mathcal{A}_{\mathcal{T}'}(\Lambda)$ such that $\mathcal{G}^{\dagger}Q'=Q\mathcal{G}^{\dagger}$ which has the same support as $Q$ plus an arbitrary dressing to the boundary. Since this dressing commutes with $O'$ \kd{need to add proof of this} it follows that that $O'$ and $Q'$ commute, in contradiction to the fact that $O$ and $Q$ don't commute. Thus $O$ is supported completely in $\mathcal{E}(R)$. Now, by definition of a holographic code there must exist $\tilde{O}\in\mathcal{A}_{\mathcal{B}}(R)$ such that $\mathcal{V}O=\tilde{O}\mathcal{V}$. Thus $\mathcal{V}\circ\mathcal{G}^{\dagger}O'=\tilde{O}\mathcal{V}\circ\mathcal{G}^{\dagger}$.

\end{proof}
We can show an example of this construction in the context of a HaPPY code.
\begin{example}
\label{example: gauging-HaPPY}
Consider the HaPPY code. It is a holographic code with an unconstrained Bulk system and trivial edge degrees of freedom.
Furthermore, as discussed in \cref{subsubsec: holocodes-example-happy} it exhibits a global/global symmetry duality with group $G=\mathbb{Z}_2$. 
Thus we can apply the above reasoning to gauge HaPPY with respect to this $\mathbb{Z}_2$ symmetry -- that is, to produce a holographic code with a $\mathbb{Z}_2$ fixed sector gauge/global duality. 
\end{example}

\subsubsection{Comment on entanglement entropy}

A primary motivation for including gauge-like degrees of freedom ``living on the edges'' in the LOTE code from Ref.~\cite{Marolf} was to capture higher order corrections to the entanglement entropy than was possible in HaPPY. 
In this subsection we examine the entanglement entropy of gauged HaPPY and find that no similar correction term linear in the bulk state appears, as the contribution to entanglement entropy from edge degrees of freedom on the gauge invariant Hilbert space is stripped off as the bulk state passes through the ungauging map. 

The entanglement entropy of a holographic code 
%with a global/global duality that has been gauged into a fixed flux sector to produce an encoding isometry $\mathcal{V}'$
constructed from the above theorem is most easily understood by viewing the gauged code as a conventional holographic code with encoding isometry $\mathcal{V}$ with the ungauged state $\mathcal{G}^{\dagger}\ket{\psi_{\text{Bulk}}}$ in the Bulk, where $\ket{\psi_{\text{Bulk}}}$ is a gauge invariant Bulk state within the zero flux sector. 
For simplicity consider a decomposition of the Boundary into $R,R^{c}$ with respective Bulk entanglement wedges $\mathcal{E}(R),\mathcal{E}(R)^c$, i.e.\ satisfying complementary recovery. This is realized by certain regions for example in the HaPPY code using the greedy entanglement wedge~\cite{HaPPY,noncliffords}. 
The entanglement entropy of the boundary region $R$ is  
\begin{align}
\label{eq:HoloCodeRT}
    % &S_R(\mathcal{V} \mathcal{G}^\dagger \rho_{\text{Bulk}} \mathcal{G} \mathcal{V}^\dagger ) = S([\mathcal{G}^\dagger \rho_{\text{Bulk}} \mathcal{G}]_{\mathcal{E}(R)}) + |\gamma_R|\ln \chi  \, , 
    % %\qquad \text{where } \rho' := G^\dagger \rho_{\text{Bulk}} \text{ is the ungauged Bulk state,}
    % \\
    % &\text{or}
    % \\
    % &S([\mathcal{V}\mathcal{G}^\dagger \rho_{\text{Bulk}} \mathcal{G} \mathcal{V}^\dagger]_R ) \ ?
    % \\
    % S([\mathcal{V} \mathcal{G}^\dagger\ket{\psi_{\text{Bulk}}}]_R) = S([\mathcal{G}^\dagger\ket{\psi_{\text{Bulk}}}]_{\mathcal{E}(R)}) +  |\gamma_R|\ln \chi  \, , 
    % \\
    S([\mathcal{V} \rho \mathcal{V}^\dagger]_R) = S(\rho_{\mathcal{E}(R)})  +  |\gamma_R|\ln \chi  \, , 
\end{align}
%where $\rho_{\text{Bulk}}=\ket{\psi_{\text{Bulk}}}\bra{\psi_{\text{Bulk}}}$  
where $\rho = \mathcal{G}^\dagger \ket{\psi_{\text{Bulk}}}\bra{\psi_{\text{Bulk}}} \mathcal{G}$ is the ungauged Bulk state, and $\gamma_R$ is the RT surface for $R$, counted in units of how many legs with bond dimension $\chi$ are cut in the holographic tensor network. From this formula we see that the Bulk gauge fields do not directly contribute to the boundary entropy as only the ungauged Bulk state appears in Eq.~\eqref{eq:HoloCodeRT}. 
More generally when complementary recovery is not strictly satisfied, as is ubiquitous in exact holographic codes, there may be some residual region not contained within the entanglement wedge of $R$ or $R^c$ which leads to a further correction. 
%For simplicity we present our discussion of the boundary entropy formula for the demonstrative case where complementary recovery is satisfied.

%Although Ref.~\cite{VanAcoleyen2016} (Appendix C) provides some tools for studying the entropy of the gauging map with no NGC vertices, their techniques do not easily generalize to our more general family of gauging maps.
%Thus, it appears difficult to restate this formula directly in terms of the entropy of the gauged state, $S(\Tr_{\mathcal{E}(R)^c} \ketbra{\psi_{\mathrm{Bulk}}} )$.
%Thus, it is difficult to make a direct comparison with the entropy decomposition found in Ref.~\cite{Marolf}.

In Appendix C of Ref.~\cite{VanAcoleyen2016} a formula was derived for the entanglement entropy of a gauged state in terms of the ungauged state and a state independent perimeter term (due to the gauge degrees of freedom), assuming the gauging map is abelian and has no NGC vertices.
The explicit formula is 
\begin{align}
\label{eq:GaugedS}
    S( \rho^\mathcal{G}_A ) = (|\partial A|-1) \ln |G| + S(\rho_A)
\end{align}
where $\rho_A = \tr_{A^c} (\ket{\psi} \bra{\psi} )$, $\rho^\mathcal{G}_A = \tr_{A^c}(\mathcal{G} \ket{\psi} \bra{\psi} \mathcal{G}^\dagger)$, and $|\partial A|$ counts the number of edges between $A$ and $A^c$. 
In this case the entropy of the ungauged state clearly differs from the entropy of the gauged state by a term that is proportional to the boundary. This boundary term is of a different nature to the boundary term appearing in the calculation of entanglement entropy in Ref.~\cite{Marolf} as it does not depend on the Bulk state. 
There is an additional complication, as the simple formula above does not straightforwardly generalize to the family of gauging maps considered here due to the presence of NGC vertices. 
Thus, we do not draw an explicit comparison with the entropy decomposition found in Ref.~\cite{Marolf} for the gauging maps we consider.

\subsection{Ungauging: gauge/global \texorpdfstring{$\implies$}{→} global/global }
We can also show a result that mirrors \cref{thm:global_global_to_gauge_global}, by constructing a code with a global/global duality given a theory with either a fixed sector or full gauge/global duality. The intuitive picture is similar, and is shown in \cref{fig:ungauging-collective}. We prove the statement formally in theorem~\ref{thm:gauge_global_to_global_global}, and exemplify it with example~\ref{example: ungauging-HaPPY}. 

\begin{figure}
     \centering
     \begin{subfigure}[t]{0.46\textwidth}
         \centering
\includegraphics[width=1.1\textwidth]{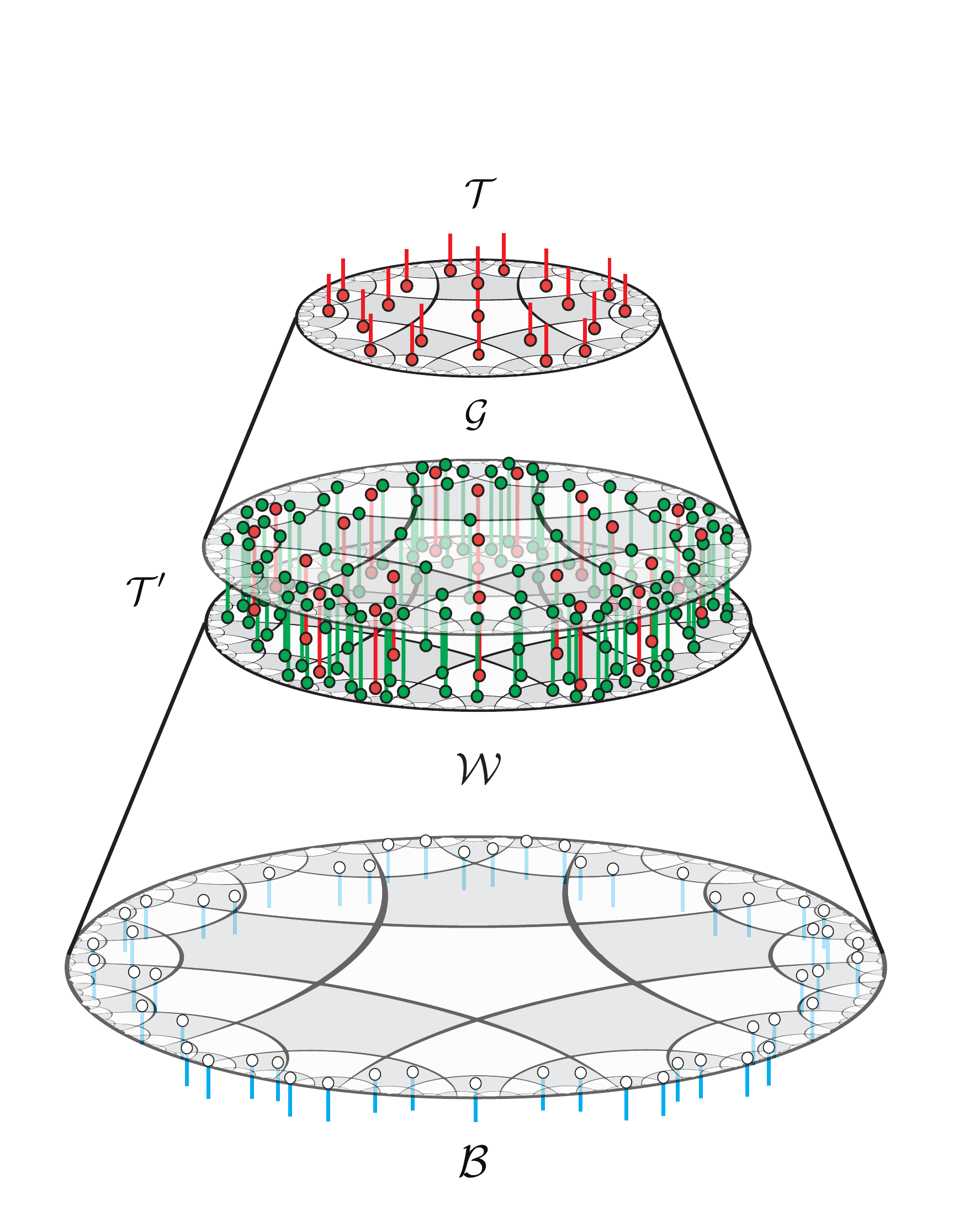}
         \caption{Depiction of $\mathcal{W} \circ \mathcal{G}$, the encoding isometry of the code given by \cref{thm:gauge_global_to_global_global}.
}
         \label{fig:}
     \end{subfigure}
     \hfill
     \begin{subfigure}[t]{0.23\textwidth}
         \centering
         \includegraphics[width=\textwidth]{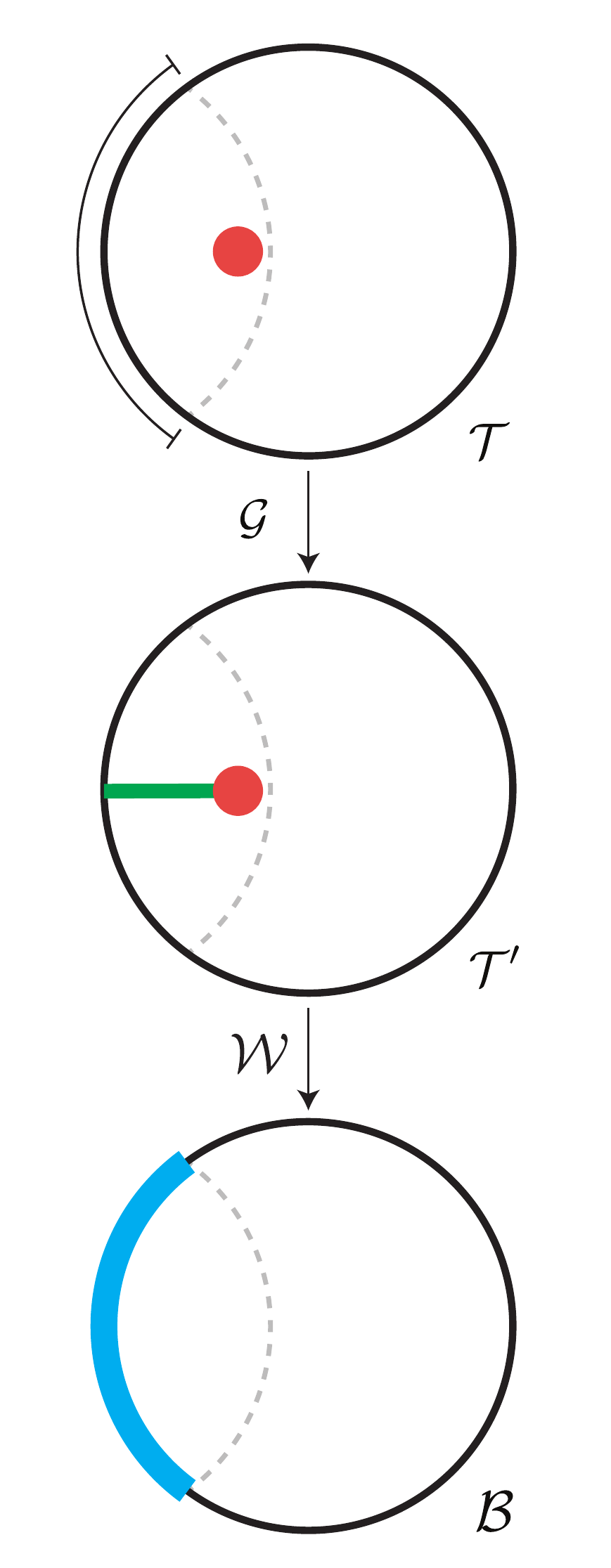}
         \caption{$\mathcal{W} \circ \mathcal{G}$ inherits holographic reconstruction from $\mathcal{W}$.}
         \label{fig:ungauging-flan}
     \end{subfigure}
     \hfill
     \begin{subfigure}[t]{0.23\textwidth}
         \centering
         \includegraphics[width=\textwidth]{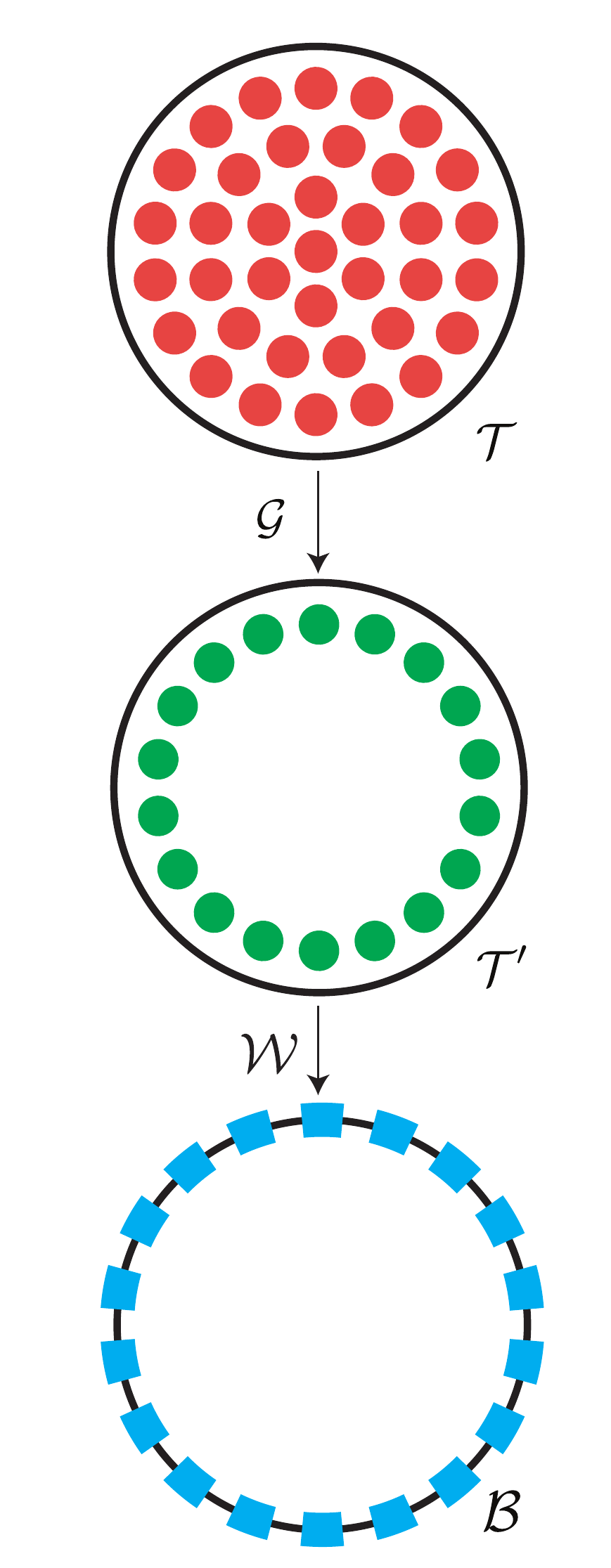}
         \caption{$\mathcal{W} \circ \mathcal{G} $ exhibits a global/global duality.}
         \label{fig:ungauging-ewedge-preserved}
     \end{subfigure}
        \caption{Basic components of the proof that any holographic code with a gauge/global symmetry duality can be converted into one with a global/global symmetry duality via \cref{thm:gauge_global_to_global_global}. The structure is essentially the same as in  \cref{thm:global_global_to_gauge_global} except that the $\mathcal{T}$' may now have a full rather than fixed sector gauge-invariance.}
        
        \label{fig:ungauging-collective}
\end{figure}

\begin{theorem}
\label{thm:gauge_global_to_global_global}
Suppose there exists a holographic code $(\mathcal{T}',\Lambda,\alpha,\mathcal{B},\mathcal{E},\{\mathcal{A}_{\mathcal{E}(R)}\},\mathcal{W})$ such that $\mathcal{T}'=(\mathcal{U}',\Pi)$ is a constrained system obtained by gauging a system with global symmetry $(\mathcal{T},G,\{U_v\})$\footnote{Not to be confused with a system obtained by applying the gauging map $\mathcal{G}$; see discussion after \cref{def: gauging-map}.}, potentially further constrained to its flux-free sector, i.e. $\Pi=\Pi_{GI}$ or $\Pi=\Pi_{FF}$. Suppose $\mathcal{W}$ exhibits a gauge/global duality with some bulk global symmetry $(\mathcal{B},G,\{\tilde U_v\}_{v\in V_0})$. Then there exists a holographic code $(\mathcal{T},\Lambda,\alpha,\mathcal{B},\mathcal{E},\{\mathcal{A}'_{\mathcal{E}(R)}\},\mathcal{W}')$ with $\mathcal{T}$ unconstrained and with trivial edge degrees of freedom which exhibits a global/global symmetry duality.
\end{theorem}

\begin{proof}
Let $\mathcal{G}:\mathcal{H}_\Lambda\rightarrow\Pi\mathcal{H}'_\Lambda$ be the gauging map for the Hilbert space $\mathcal{H}_\Lambda$ associated with system $\mathcal{T}$.
The map $\mathcal{W}'=\mathcal{W}\circ\mathcal{G}$ isometrically encodes the Hilbert space of $\mathcal{T}$ into that of $\mathcal{B}$.
$\mathcal{W}'$ has a global/global duality since $\mathcal{W}\mathcal{G}U_{V}(g)=\mathcal{W}A_{V_0}(g)\mathcal{G}=\tilde{U}_{V_0}(g)\mathcal{W}\mathcal{G} $, and by assumption, $\tilde{U}_{V_0}(g)$ preserves the image of $\mathcal{W}$.
In the last equation, we used that if $\Pi = \Pi_{GI} $ or $\Pi_{FF}$, the asymptotic symmetry transformation $A_{V_0}$ is a physical operator and thus can be reconstructed on the Boundary.

It remains to show that this construction is a valid holographic code, i.e.\ for any Boundary region $R\subseteq  V_0$ any operator $O\in\mathcal{A}_{\mathcal{T}}(\interior{\mathcal{E}(R)})$ can be reconstructed on $R$.
In this case, the interior restriction is irrelevant as there are no edge degrees of freedom.
By \cref{thm: operator-dressing-theorem} there exists an operator $O'$ such that $\mathcal{G}O=O'\mathcal{G}$ with the same support as $O$ plus an arbitrary dressing along a path to the Boundary.
By definition of a holographic code (namely the path-connected requirement for a near-boundary probing map), there exists such a path purely contained in $\mathcal{E}(R)$.
Furthermore, it can be chosen to be in $\interior{\mathcal{E}(R)}$ because the support of $O$ only contains vertices.
Choosing the dressing to have support along this path, we have that $O'$ can be reconstructed on $R$, i.e.\ there exists an operator $\tilde{O}\in\mathcal{A}_\mathcal{B}(R)$ such that $\mathcal{W}O'=\tilde{O}\mathcal{W}$.
Thus $\mathcal{W}\mathcal{G}O=\tilde{O}\mathcal{W}\mathcal{G}$. Finally, one can also show that $\mathcal{W}O'^\dagger=\tilde{O}^\dagger\mathcal{W}$ and thus $\tilde{O}$ preserves the image of $\mathcal{W}$.
\end{proof}

\begin{example}
\label{example: ungauging-HaPPY}
The $gauge/global$ code obtained by gauging the HaPPY code, as in example \ref{example: gauging-HaPPY}, can be ungauged to re-obtain HaPPY and its $\mathds{Z}_2$ symmetry.
\end{example}

\begin{example}
\label{example: ungauging-LOTE}
Consider the gauged LOTE code.
It is a holographic code with a constrained Bulk system, whose Hilbert space is one obtained by gauging a system with global symmetry, and whose encoding isometry, by \cref{thm: full_gauge_to_gauge_global}, exhibits a gauge/global duality.
Thus it can be ungauged to obtain a system with a global/global duality.
\end{example}

\subsubsection{Comment on entanglement entropy}
Let us briefly discuss what can be said about the entanglement entropy of boundary regions in a code constructed via \cref{thm:gauge_global_to_global_global}.
If one gauges and ungauges in succession, i.e.\ applies \cref{thm:global_global_to_gauge_global} and then \cref{thm:gauge_global_to_global_global}, the entropy simply has the same bulk/boundary relationship as in the original holographic code
%In a holographic code with global/global duality that is obtained by ungauging a holographic code that was gauged to obtain gauge/global duality in a fixed flux sector the entanglement entropy simply reduces to that of the original holographic code before gauging 
(note with the boundary conditions we pick there is no projector onto the symmetric sector due to gauging and ungauging i.e. $\mathcal{G}^\dagger \mathcal{G}=\mathbbm{1}$).

For a less trivial example, consider a holographic code with gauge/global duality produced by projecting a Bulk pre-gauged Hilbert space onto the gauge invariant subspace, such as gauged LOTE. 
The holographic encoding isometry of the LOTE code~\cite{Marolf} is given by composing six copies of the HaPPY code isometry, $\mathcal{V}$, with an isometry associated to each edge of the Bulk, $W^{\otimes |E|}$ illustrated in Fig.~\ref{fig:LOTE-construction} (b), to form a combined encoding isometry $\mathcal{V}^{\otimes 6}W^{\otimes |E|}$. In what follows we allow each $\mathcal{V}$ to be a stack of multiple HaPPY codes to achieve the bond dimension required to support the group representations of a finite group $G$ we consider. 

For simplicity of presentation, in this subsection we choose the edge isometry to be the group multiplication tensor
\begin{align}
    W = \frac{1}{\sqrt{|G|}} \sum_{g,h\in G} \ket{g,h} \bra{ g h^{-1} } \, ,
\end{align}
which satisfies 
\begin{align}
\label{eq:WULRG}
    \big(U^L(g)\otimes \mathbbm{1}\big) W = W U^{L}(g) \, ,
    &&
    \big(\mathbbm{1} \otimes U^L(g)\big) W = W U^{R}(g) \, ,
    &&
    W \ket{1} = \ket{\Phi} \, ,
\end{align}
where 
\begin{align}
    \ket{\Phi} =  \frac{1}{\sqrt{|G|}} \sum_{g \in G} \ket{g,g} \, ,
\end{align}
is a maximally entangled state. 
The tensor $W$ on edge $e$ is oriented such that left multiplication moves through to act on the vertex $e^-$ and right multiplication on $e^+$.

Projecting the bulk of LOTE onto the gauge invariant subspace produces a holographic code with encoding map $\mathcal{V}^{\otimes 6}W^{\otimes |E|}\Pi_{GI}$ that exhibits gauge/global duality. 
Ungauging the bulk produces a new code with encoding map $\mathcal{V}^{\otimes 6}W^{\otimes |E|}\mathcal{G}$ that exhibits global/global duality. 

The entropy of a boundary region $R$ that satisfies complementary recovery in the ungauged LOTE code is similar to Eq.~\eqref{eq:HoloCodeRT},
\begin{align}
\label{eq:HoloCodeRT3}
    S([\rho_{\text{Bdry}}]_R) &= S([ \rho'_{\text{Bulk}'} ]_{\mathcal{E}( R)}) + |\gamma_R| \ln \chi \, 
    \\
    &= S([\rho_{\text{Bulk}}]_{\mathcal{E}(R)}) + |\gamma_{R}| \ln |G| + |\gamma_R| \ln \chi
\end{align}
where $\chi$ measures the total bond dimension of the stack of HaPPY tensor networks and 
\begin{align}
    \rho_{\text{Bdry}} &= |G|^{|V|} \, \mathcal{V}^{\otimes 6}W^{\otimes |E|}\mathcal{G}\, \rho_{\text{Bulk}} \mathcal{G}^\dagger W^{\dagger \otimes |E|}   \mathcal{V}^{\dagger \otimes 6} \, ,
    \\
    \rho'_{\text{Bulk}'} &= |G|^{|V|}\,  W^{\otimes |E|}\mathcal{G}\, \rho_{\text{Bulk}} \mathcal{G}^\dagger W^{\dagger \otimes |E|} \, , 
\end{align}
for $\rho_{\text{Bulk}}$ the state on the ungauged Bulk Hilbert space. The prefactors $|G|^{|V|}$ are included to correctly normalize the state. 
We point out that the $|\gamma_{R}| \ln |G|$ term in the entropy can be absorbed into the RT term by making the coefficient $(\ln \chi+\ln |G|)$, which reflects an increase in the bond dimension of the encoding isometry due to the gauging map $\mathcal{G}$. 

To explain the decomposition of the $S([ \rho'_{\text{Bulk}'}]_{\mathcal{E} (R)})$ term on the second line of Eq.~\eqref{eq:HoloCodeRT3} it is useful to first rewrite the $W^{\otimes |E|}\mathcal{G}$ part of the encoding map by making use of the properties of $W$ introduced above in Eq.~\eqref{eq:WULRG}
\begin{align}
    W^{\otimes |E|}\mathcal{G}  = W^{\otimes |E|} \Pi_{GI} \bigotimes_{e} \ket{1}_e 
    = \Pi_{GI}' \bigotimes_{e} W \ket{1}_e 
    = \Pi_{GI}' \bigotimes_{e} \ket{\Phi}_e \, ,
    \label{eq:WGIPhi}
\end{align}
for 
\begin{align}
    \Pi_{GI}' = \prod_{v} \Pi_v' \, ,
    &&
    \Pi_v' = \int dg\,  U_v(g) \otimes U^L(g)^{\otimes 5} \, ,
\end{align}
where $\Pi_v'$ acts on the six legs of the HaPPY copies at vertex $v$ in the Bulk. 
Next, we note
\begin{align}
    [ \rho'_{\text{Bulk}'}]_{\mathcal{E} (R)} &= |G|^{|V|}
    \Tr_{\mathcal{E}(R)^c} \big( \Pi_{GI}' 
     \, \rho_{\text{Bulk}} \, \bigotimes_{e} \ket{\Phi}\bra{\Phi}_e \Pi_{GI}' \big) \, , \\
     &= |G|^{|\mathcal{E}(R)|} \Pi'_{\mathcal{E}(R)} \, [\rho_{\text{Bulk}}]_{\mathcal{E}(R)} \, \bigotimes_{e\in \mathcal{E} (R)} \ket{\Phi}\bra{\Phi}_e 
     \bigotimes_{e \in \gamma_R} \frac{1}{|G|} \mathbbm{1}_{v_e} \,
     \Pi'_{\mathcal{E}(R)}
     \, ,
\end{align}
where $v_e$ denotes the vertex in $\mathcal{E}(R)$ that is adjacent to $e$, and $\Pi'_{\mathcal{E}(R)} = \prod_{v\in \mathcal{E}(R)} \Pi'_v$.
We have used \cref{eq:WGIPhi} and successive applications of $\bra{\Phi} U^{L/R}(g) \otimes \mathbbm{1} \ket{\Phi}=\delta(g)$ starting from the rough boundary conditions to fix all the group elements from $\Pi_{GI}'$ acting on vertices within $\mathcal{E}(R)^c$ to identity. 

We now make use of the equality
\begin{align}
    \Tr([ \rho'_{\text{Bulk}'}]_{\mathcal{E} (R)}^n) = 
    \Tr ([\rho_{\text{Bulk}}]^n_{\mathcal{E}(R)}\bigotimes_{e \in \gamma_R} \frac{1}{|G|^n} \mathbbm{1}_{v_e})
\end{align}
which follows from the observation above that all group elements from $\Pi'_{\mathcal{E}(R)}$s appearing on the LHS are fixed to the identity by $\ket{\Phi}$ states. 
To finally arrive at the desired decomposition of the bulk entropy term 
\begin{align}
    S( [ \rho'_{\text{Bulk}'}]_{\mathcal{E} (R)}) 
    &= S([\rho_{\text{Bulk}}]_{\mathcal{E}(R)}\bigotimes_{e \in \gamma_R} \frac{1}{|G|} \mathbbm{1}_{v_e}) \, ,
    \\
    &= S([\rho_{\text{Bulk}}]_{\mathcal{E}(R)}) + |\gamma_{R}| \ln |G| \, .
\end{align}

\subsection{Connection to Harlow-Ooguri}
\label{subsec:connection-to-harlow-ooguri}

So far we have frequently discussed holographic codes with global/global symmetry dualities. This notion may cause alarm to those familiar with the results of Ref.~\cite{harlow-ooguri}, which show that AdS/CFT cannot exhibit such dualities. Adding further to the confusion is that these results appear to rely only on precisely those aspects of AdS/CFT that holographic codes, which certainly can exhibit such dualities, manage to capture. In Ref.~\cite{Faist_2020} the subtle but key distinction between the two was pointed out: for holographic codes, the action of Boundary global symmetries on restricted regions need not preserve the code subspace. We now argue that this distinction is potentially superficial, and propose a more complete criteria by which to judge toy models which deserves further exploration. 

We first clarify an additional subtlety in the distinction between AdS/CFT and holographic codes with global/global dualities. We then use the gauged LOTE code to point out that a holographic code displaying a global/global duality may just mean that its bulk is overly restricted. We proceed by arguing that any global/global holographic code is in fact overly restricted, and that there ought to be a way to lift its Bulk to a full gauge-like symmetry so that a restricted global symmetry keeps it in the code space. We discuss a way to do this for the HaPPY code, proposed in Ref.~\cite{Marolf}, that unfortunately destroys the entanglement wedge reconstruction property. This leads us to pose the question: when is it possible to lift the Bulk of a holographic code to a full gauge invariant Hilbert space while maintaining its holographic properties? Finally, we give a sufficient condition for holographic codes to both display Bulk gauge symmetry and to have restrictions of the Boundary global symmetry preserve the code space.

\subsubsection{Global/global dualities in holographic codes are benign} 

The additional subtlety mentioned above is that the restricted action is only defined up to local operators supported on a neighborhood of the domain walls, that cancel out when restricted actions are multiplied to form a global symmetry action. To state this more precisely we need to introduce the notion of symmetry domain walls for holographic codes.
\begin{definition}
A collection of \textbf{$g$-domain walls} of a global symmetry on the Boundary of a holographic code with encoding isometry $\mathcal V$ are created at the location $\partial R$ by a partial symmetry action $\prod_{v\in R}U_v(g)$ on a subregion $R$ of the Boundary. This defines a new holographic encoding map $\big(\prod_{v\in \mathcal R}U_v(g) \big) \mathcal V$ which contains $g$-domain walls on $\partial R$. 
\end{definition}

\begin{definition}
    The \textbf{superselection sector} within a sufficiently small neighborhood $\mathcal N$ of a $g$-domain wall at $\partial R$ is the equivalence class
    \begin{align}
        \{ \mathcal{D}_{ \mathcal N} \big(\prod_{v\in  R}U_v(g) \big) \mathcal{V} \ |\  \mathcal{D}_{ \mathcal N} \text{ \emph{a unitary operator supported on}  }\mathcal N    \}
        \, .
    \end{align}
\end{definition} 

\begin{remark}
\label{rem:BdryDW}
If a decomposition of the Boundary into disjoint regions $R_i$ is fixed, the restricted symmetry actions $U_{R_i}(g)=\prod_{v \in R_i} U_v(g)$ can be redefined up to any boundary unitaries $\mathcal{D}_{\partial R_i}(g)$ that preserve the global symmetry condition
\begin{align}
    \prod_{i} \big( \mathcal{D}_{\partial R_i}(g) U_{R_i}(g) \big)= U_V(g) \, .
\end{align}
Such redefinitions clearly preserve the superselection sectors of the domain walls created by restricted symmetry actions $\mathcal{D}_{\partial R_i}(g) U_{R_i}(g)$.
\end{remark}

The freedom in redefining the domain wall was effectively used in Ref.~\cite{harlow-ooguri} to argue that the restricted action of any continuous symmetry (or a finite group symmetry embedded into a continuous group) can be made code space preserving because the domain wall can be smeared out so as to become an arbitrarily low energy excitation, which, in AdS/CFT means it can be added to the code space if not already contained. In holographic codes with global/global dualities such as HaPPY, the symmetry domain walls lie in nontrivial superselection sectors, i.e. no redefinition of the domain wall by multiplication with a local operator can return the domain wall to the codespace.
This is guaranteed by the arguments in Ref.~\cite{harlow-ooguri} -- if such a redefinition were possible, such that the restricted symmetry was codespace-preserving, then the arguments presented there would apply and rule out the possibility of a global/global duality, which is inconsistent with the explicit presence of such a duality in HaPPY~\cite{noncliffords}.

The global/global code we have constructed from the gauged LOTE code evades this theorem for precisely the same reason -- its restricted symmetry actions are not codespace-preserving.
However, by looking ``under the hood'' of this model we see that this is only true in a subspace, and that the full ambient Hilbert space does not violate the Harlow-Ooguri theorem. More precisely, the global/global symmetry appears only when we restrict the Bulk system to a fixed-flux sector. This causes the action of a restricted global symmetry transformation to leave the codespace because such an action is mapped to restricted asymptotic symmetry transformations in the Bulk, and these generically change the flux sector from flux-free to an NGC-flux sector (see \cref{subsec: NGC-flux-maps}). As we have seen, we can ungauge these flux sectors using a gauging map tailored to them, and recover a different holographic code. This is consistent with the fact that a restricted boundary global symmetry action can be viewed as a change of basis, and thus should not change the entanglement wedge map. If we return to the full gauge theory description of the bulk (i.e.\ one with constraint $\Pi_{GI}$ rather than $\Pi_{FF}$), then a restricted boundary global symmetry action preserves the codespace because it now includes all flux sectors. 

Interestingly, we can conclude from this that it is impossible to isometrically ungauge a full gauge invariant Hilbert space using a map with the same properties as the original gauging map -- namely, being an isometry, exhibiting a gauge/global duality, and preserving the locality of vertex operators.
Such a map could be composed with HaPPY to create a holographic code with a global/global duality that \textit{does} exhibit all the properties necessary to use the Harlow-Ooguri theorem, meaning it cannot have a global/global duality, which is a contradiction.
In fact, this allows us to draw a general conclusion about gauge theory that is independent of holographic codes: there is no general map that isometrically maps a system with gauge symmetry into an unconstrained system such that the map exhibits a gauge/global duality and preserves the locality of operators supported on the vertices. 
%Notice that the nonexistence of such a map is a statement purely about lattice gauge theory, even though we used properties of holographic codes to derive it.\footnote{At least such a map can't exist for any graph, since the ones that can be placed in the bulk of a holographic code are counterexamples by this argument.}

It is natural now to wonder whether the global/global duality of the HaPPY code really just comes from an over-restriction of the Bulk. After all, as we discuss in the next section, even in AdS/CFT there emerges a global/global duality at sufficiently low energies.\footnote{That is, at low energies the emergence of bulk locality makes time evolution look like a spacetime version of a global/global duality (which still violates Harlow-Ooguri), albeit only approximate due to metric fluctuations}
For any global/global holographic code, it still holds that the action of restricted global symmetries on the boundary takes the code out of the codespace, resulting in another holographic code with the same entanglement wedge map as the original. So inside of the code there is secretly something resembling multiple flux sectors. In the case of HaPPY, it is possible to lift the Bulk to a gauge invariant like system so that such an action preserves the code space. This is done by inserting the copying tensor described in \cref{subsubsec: holocodes-example-lote} directly in between the virtual edges of the tensor network, as explained Ref.~\cite{Marolf}. However, in this lifted HaPPY code entanglement wedge reconstruction no longer holds outside of the fixed flux sectors, in contrast with the gauged LOTE code and AdS/CFT. 

We see then that before dismissing global/global holographic codes as irredeemably inaccurate models of AdS/CFT, it can be constructive to first ask whether or not it is possible to lift their Bulk to a gauge-like theory which maintains entanglement wedge reconstruction and has boundary domain walls preserving the codespace. We leave the characterization of conditions under which this is possible for future work. 

\subsubsection{A sufficient condition for gauge symmetry}

For the remainder of this subsection we focus on domain walls created by partial global symmetry actions within the Bulk and on the Boundary, and the consequences for the holographic encoding map when they can be created and removed by local operators. 

Recall that the arguments of Ref.~\cite{harlow-ooguri} imply a contradiction if we have a holographic isometry map $\mathcal{V}$ with an unconstrained Bulk Hilbert space satisfying both the global/global duality condition $\tilde{U}_{\tilde{V}}(g) \mathcal{V} = \mathcal{V} U_V (g)$ and that the restricted symmetry action $\tilde{U}_{\tilde{R}}$ on a region $\tilde{R}$ preserves the codespace, potentially after multiplication with a domain wall operator $\tilde{\mathcal{D}}_{\partial \tilde{R}}$. 
It is further argued in Ref.~\cite{harlow-ooguri} that in AdS/CFT with these assumptions the bulk must support a gauge symmetry. It would certainly be interesting to establish a similar result in the setting of holographic codes satisfying the above two symmetry conditions. However, it is unclear that giving up the unconstrained Bulk assumption is the only way to resolve the tension, it could also be interesting to look into whether abandoning the symmetry condition $\tilde{U}_{\tilde{V}}(g) \mathcal{V} = \mathcal{V} U_V (g)$ while maintaining the holographic structure of the code is possible. We leave these questions to future work and instead identify a natural sufficient condition on a holographic map, that we call \textit{locally removable domain walls}, that implies the two symmetry conditions introduced above. We go on to show that this condition automatically ensures the bulk is constrained to be invariant under local symmetries that we suggest can be interpreted as gauge symmetries. 

We first define symmetry domain walls in the Bulk of a holographic map analogously to the Boundary definition above. 
\begin{definition}
    A collection of $g$\textbf{-domain walls} of a global symmetry in the Bulk of a holographic code with encoding isometry $\mathcal{V}$ are created at the location $\partial {R}$ by a partial symmetry action $\prod_{v\in{R}} {U}_v(g)$ on a subregion ${R}$ of the Bulk. This defines a new holographic encoding map $\mathcal{V} \prod_{v\in{R}} {U}_v(g)$ which contains $g$-domain walls on $\partial {R}$.
\end{definition}

\begin{remark}
Part, or all, of $\partial {R}$ may coincide with the Boundary, we decompose it as ${\partial {R}= (\partial {R})_\partial \sqcup (\partial {R})_\circ}$ for $(\partial {R})_\partial $ the portion of $\partial {R}$ on the Boundary and $(\partial {R})_\circ$ the portion in the Bulk. See Fig.~\ref{fig:LRDW} for an illustration. 
\end{remark}

\begin{remark}
In HaPPY for the $\mathbb{Z}_2$ symmetry generated by $X$ operators the domain walls can be identified explicitly as $X$ operators acting on the virtual and/or Boundary legs of the tensors on $\partial {R}$. A similar property holds for tensor network holographic codes satisfying a local symmetry condition such as in Ref.~\cite{1367-2630-12-2-025010,Singh2013,williamson2014matrix,Bridgeman2017,NewSETPaper2017}, but our definition above does not require a tensor network structure. 
\end{remark}

Domain walls of the Bulk symmetry that coincide with the Boundary can be removed by acting on the boundary. In particular the domain wall of the global action ${U}_{{V}}$ in the bulk covers the whole Boundary and can be removed by the global action $\tilde{U}_{\tilde{V}}^\dagger$ on the Boundary. To see this notice that creating the domain wall from the bulk and then removing it on the boundary preserves the original encoding map $\tilde{U}_{\tilde{V}}^\dagger \mathcal{V} {U}_{{V}} = \mathcal{V}$. 
Hence the Boundary symmetry appears as a domain wall of the bulk symmetry, and consequently domain walls of the Boundary symmetry correspond to twist defects of the Bulk symmetry, i.e. locations where the bulk domain wall terminates, see Refs.~\cite{williamson2014matrix,Bridgeman2017,NewSETPaper2017} for a discussion of such twist defects in tensor networks. 
For a restricted Bulk symmetry action $\prod_{v\in{R}} {U}_v(g)$, any component of its domain wall that lies on the Boundary can be removed by the restricted Boundary symmetry action $\prod_{v\in(\partial {R})_\partial} \tilde{U}_v(g)$, up to some operator in the neighborhood of the twist defect at $\partial (\partial {R})_\partial$ on the Boundary, see Remark~\ref{rem:BdryDW}. %Note that the twist defect may not be removable

Next we explore the consequences of a local symmetry condition that requires any domain wall in a holographic map to be removable by some operator acting within a neighborhood of the domain wall, this is analogous to the domain walls lying in trivial superselection sectors. 
\begin{definition}
    A holographic map $\mathcal{V}$ satisfying $ \mathcal{V} {U}_{{V}} = \tilde{U}_{\tilde{V}} \mathcal{V}$ has \textbf{locally removable domain walls} if for any restricted symmetry action $\prod_{v\in{R}} {U}_v(g)$ there exists an operator $\mathcal{D}_\circ(g)$ supported in a neighborhood of the bulk domain wall $(\partial{{R}})_\circ$ and an operator $\mathcal{D}_\partial(g)$ supported in a neighborhood of the boundary twist defect $\partial (\partial{{R}})_\partial $ such that 
    \begin{align}
    \label{eq:RemovableDW}
       \mathcal{D}_\partial(g) \prod_{v\in (\partial {R})_\partial} \tilde{U}_v(g) \,  \mathcal{V} \, \prod_{v\in{R}} {U}_v(g) \mathcal{D}_\circ(g) = \mathcal{V} \, .
    \end{align}
\end{definition}

\begin{figure}[t]
  \centering
    \includegraphics[width=.8\textwidth]{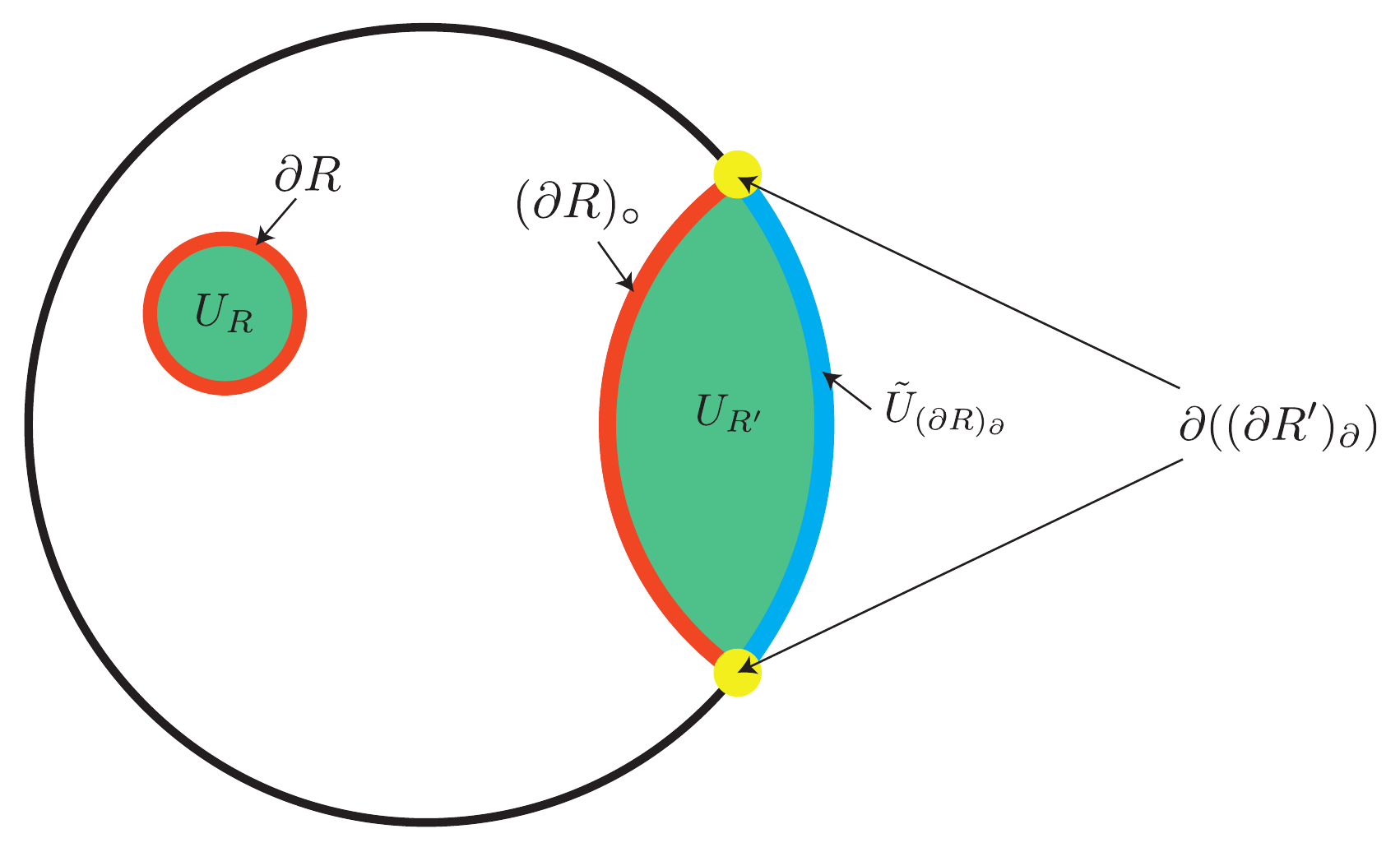}
    \caption{The symmetry action $U_R$ within a Bulk region $R$ introduces a domain wall along a neighborhood of its boundary. A symmetry action $U_{R'}$ in a Bulk region $R'$ along with a Boundary symmetry action $\tilde{U}_{(\partial R)_\partial}$ where $R'$ coincides with the Boundary introduces a domain wall along a neighborhood of the boundary $(\partial R')_\circ$ and a Bulk defect (or Boundary domain wall) in a neighborhood of $\partial ((\partial R')_\partial)$.  }
  \label{fig:LRDW}
\end{figure}

\begin{remark}
For the connected Bulk geometries we consider, it is sufficient to require the above condition to hold only for connected Bulk subregions that overlap nontrivially with the Boundary. By considering multiple Bulk subregions with the same Boundary overlap one can generate all purely Bulk constraints covered by the above definition. 
\end{remark}

\begin{remark}
The condition introduced in the definition above is sufficient to ensure that the domain walls of the Boundary symmetry lie in the trivial superselection sector and hence the arguments of Ref.~\cite{harlow-ooguri} apply. It further implies that the bulk is constrained to be invariant under local symmetries that have a similar form to gauge symmetries. To see this consider Eq.~\ref{eq:RemovableDW} for Bulk regions that consist of single sites.
\end{remark}

\begin{example}
The LOTE code after projection onto a gauge invariant Hilbert space satisfies the above definition, as does the modification of HaPPY via the introduction of edge copy tensors also considered in Ref.~\cite{Marolf}. 
\end{example}

\begin{remark}
A non-example that fails the locally removable domain walls condition is the original HaPPY code. Although the domain walls for the global symmetry generated by $X$ are simply represented by $X$ operators on the virtual indices, there is no way to locally remove them in the bulk, as this would lead to further constraints that are not present in HaPPY. More generally, the above condition must fail for any nontrivial holographic isometry map with an unconstrained Bulk. 
\end{remark}

\subsection{Analogies with gravity}

We now discuss some analogies that can be drawn between the application of the gauging map to make the gauged LOTE code global/global, and two different features of AdS/CFT: time evolution, and the approximate appearance of Bulk locality. In \cref{subsec: building_holocodes_with_arbitrary_dualities} we show how to generalize the gauged LOTE code to any group $G$. We use this here to improve our analogies. 

\subsubsection{Time evolution}
Arguably the most significant challenge holographic codes face in accurately modelling AdS/CFT is the inclusion of dynamics; specifically a \emph{dynamical duality} in which the time evolution in the bulk is equivalent to time evolution on the boundary.
Although proposals for such constructions exist~\cite{Kohler2019,Apel2021}, they fall short of obtaining desired features like having the appropriate speed of information propagation.
Since AdS/CFT is a model of quantum gravity, its Bulk geometry can change with time.
For low enough energy dynamics, the Bulk geometry can remain fixed but matter degrees of freedom may still undergo dynamical evolution. 

We now argue that the global/global code constructed by applying the gauging map to a $G$ gauged LOTE code can be thought of as a model of a simple type of time evolution where all degrees of freedom are stuck in place. Choosing $G=U(1)$, we can write the Boundary global symmetry transformation as $\prod_{v\in V_0}\tilde{U}_v(g)=e^{-i(\sum_v \tilde{H}_v)t}=e^{-iH_{\text{bdry}}t}$ by parametrizing the elements of $U(1)$ with $t$. We can do the same for the Bulk transformation to obtain a Bulk Hamiltonian composed of a sum of single site terms. Thus we have a local boundary Hamiltonian implementing a local Bulk Hamiltonian via a symmetry duality. Since nearby sites are decoupled, the degrees of freedom can not interact and are ``stuck in place''. Thus far this construction is not particularly new, as there already existed codes with a $U(1)$ global/global duality (see Ref.~\cite{Faist_2020}). What is new is that the time evolution comes about from an ambient gauge theory. Before applying the gauging map to the gauged LOTE code, a Boundary global symmetry transformation would implement a Bulk asymptotic symmetry transformation. The latter acts entirely on the boundary of the Bulk system. It is only by restricting the Bulk to a fixed flux sector that the transformation can act in the deep Bulk. The same thing occurs for time evolution in AdS/CFT; it only appears to act deep within the Bulk when we describe the Bulk using some local coordinates in a fixed geometry subspace (i.e.\ we have the appearance of Bulk locality using a ``gravitational'' analog of the gauging map). 

\subsubsection{Noise from metric fluctuations}

% As discussed in Ref.~\cite{Marolf}, because of the similarity between Yang-Mills and linearized gravity, the LOTE code is a reasonable model for Bulk gravity with small fluctuations around a fixed geometry, with fluxes playing the role of perturbations from the background metric. Since all the gauged LOTE does differently is enforce gauge-invariance, we expect that in this regime it will accurately capture aspects of the gauge-redundancy appearing in gravity. 

Another aspect of AdS/CFT that our model illuminates is why insisting on Bulk locality results in AdS/CFT functioning as only an approximate error correcting code.
In that context, the codespace consists of the Bulk ground state and states obtained by acting on it with a finite number of local operators~\cite{almheiriBulkLocalityQuantum2015}.
%However, the ground state does not have a fixed value for the metric, but rather includes small fluctuations around one.
% This can be seen from a stationary phase approximation in Bulk path integral, which fixes the metric only to leading order.
In the ground state, a certain fixed geometry dominates in the Bulk path integral, which fixes the metric to leading order.
However, a full description of the ground state includes contributions from fluctuations about this fixed geometry.
When seeking a local description of the system -- that is, a locally generated algebra of observables -- we approximate the ground state and its excitations in the codespace by ignoring these fluctuations and considering just the fixed-geometry background.
It is this approximation which results in the approximate nature of the associated code.

We can construct a conceptual toy model for this from a generalized gauged LOTE code with code isometry $\mathcal{V}_{GL}$ as follows.
First, we artificially introduce fluctuations by constraining the codespace from the full gauge-invariant Hilbert space to a subspace $\Pi_s$ that is slightly perturbed from the flux-free sector, i.e.\ $\Pi_s \approx \Pi_{FF}$, but is still gauge-invariant. 
This is done in such a way that all of states in the image of $\Pi_s$ are close to, but not equal to, a state in the flux-free sector (i.e.\ no states are left unperturbed).
In a gauge constrained system, the algebra of observables is not locally generated, but for a flux-free sector this can be resolved by ungauging the Bulk.
This is done by replacing the code isometry $\mathcal{V}_{GL}$ with $\mathcal{V}_{GL} \circ \mathcal{G}$, with $\mathcal{G}$ the gauging map, so that the input is now an unconstrained Bulk system.
For the perturbed system, we consider instead the coding map $\mathcal{V}=\mathcal{V}_{GL}\circ \Pi_s \circ \mathcal{G}$.
This is only an approximate isometry, as $\Pi_s\approx \Pi_{FF}$ and thus $\Pi_s \circ \mathcal{G}\approx \mathcal{G}$. That it approximately preserves entanglement wedge reconstruction follows from the fact that $\Pi_s \approx \Pi_{FF}$, and since $\Pi_{FF}$ exactly commutes with gauge-invariant operators, $\Pi_s$ does so approximately\footnote{See \cite{flammia} for details on relating approximate reconstruction of observables to traditional definitions of approximate error correcting codes.}.
It can be converted into an exact isometry with approximately the same reconstruction properties. For example, one way to do this is to notice that since $P=\mathcal{V}^\dagger{\mathcal{V}}$ is close to the identity, and $P$ is Hermitian, $P$ should have positive eigenvalues for a close enough approximation. Thus one can use the positive square root $\sqrt{P}$ to create the exact isometry $\mathcal{V}\circ(\sqrt{P})^{-1}$.

\subsection{Building holographic codes with arbitrary symmetry dualities}
\label{subsec: building_holocodes_with_arbitrary_dualities}

Using theorems \ref{thm:gauge_global_to_global_global} and \ref{thm: full_gauge_to_gauge_global}, we now explicitly construct a holographic code with a global/global duality for an arbitrary finite group. We then show how to modify this construction to include all compact Lie groups at the cost of making the code approximate, as it must be due to the results of Ref.~\cite{Faist_2020}.

Besides giving a concrete example of the theorems described above, this construction is interesting because it provides a new method for building approximate codes with $U(2^n)$ transversal gate sets, the first such holographic codes for $n>0$.

\subsubsection{The construction}

\label{subsubsec: the-construction}

We now construct, for any finite group $G$, a holographic code with global/global duality with respect to $G$. 
The construction is identical to that of the LOTE case described in \cref{subsubsec: holocodes-example-lote} except in two regards. First, all vertex legs, edge legs, and legs which are to be combined on edges in the LOTE code have dimension 2. In general, we need to increase the dimension of the edge legs to be $|G|$, which requires increasing the dimension of the legs to be combined to some number that depends on $G$. The dimension of the vertex legs equals that of whatever unitary representations we choose for them. Second, we would like to generalize the copy tensor so that we retain reconstruction of central flux operators living on the boundary of an entanglement wedge. This is not necessary to build the global/global code, but we would like that the gauged version of this generalized LOTE code still allows of the reconstruction of as many gauge-invariant operators as possible (again we conjecture that our construction includes all gauge-invariant operators with support on the exterior of the entanglement wedge). 

We begin by addressing the second concern. The natural choice for the generalized copy tensor is what we refer to here as the ``grand-orthogonality tensor,'' described in detail in Ref.~\cite{Tagliacozzo_2014}, from which we heavily base our review of it below.

Our construction of the tensor proceeds as follows: using the grand-orthogonality theorem, we can construct a unitary map between $L^2(G)$ and a direct sum over representations.
We then embed the latter space into a larger space, that can be written as a product of two subsystems of equal dimension.
This is necessary to obtain the desired reconstruction properties, namely that a central operator on the input leg can be reconstructed on either output leg.
These systems are then identified with the output legs of the copying tensor.

The grand-orthogonality theorem states that for any group $G$, 
\begin{align}
    \frac{\sqrt{n_rn_{r'}}}{|G|}\sum_{g\in G} \pi^r_{ij}(g)^*\pi^{r'}_{kl}(g)=\delta_{rr'}\delta_{ik}\delta_{jl} \label{eq:GOT}
\end{align}
where $r$ and $r'$ label unitary representations of $G$, with dimensions $n_r$ and $n_{r'}$ respectively; $\pi^r_{ij}(g)$ is the $i,j$ component of the representation matrix of $g$ for the representation $r$ in an arbitrary but fixed basis, and $\pi^r_{ij}(g)^*$ is the complex conjugate of $\pi^r_{ij}(g)$. We can build from $\pi^r_{ij}(g)$ an isometry from the Hilbert space $\mathcal{H}_{r}\otimes\mathcal{H}_{\bar{r}}$, with both components having dimension $n_r$, to $L^2(G)$ via the map
\begin{align*}
    \ket{i}_r\ket{j}_{\bar{r}}\mapsto\sqrt{\frac{n_r}{|G|}}\sum_{g\in G}\pi^r_{ij}(g)\ket{g}.
\end{align*}
If we take a direct sum of these maps, so that the domain is $\oplus_r (\mathcal{H}_{r}\otimes\mathcal{H}_{\bar{r}})$, then using the grand-orthogonality theorem \cref{eq:GOT}, as well as the representation theory result that $\sum_r n_r^2=|G|$, implies that it is a unitary map.

Now, we take the conjugate of the unitary map and extend its target to $(\oplus_r \mathcal{H}_{r})\otimes(\oplus_{r'}\mathcal{H}_{\bar{r}}')$ -- the product of two equal-dimension subsystems mentioned above.
This gives us the map
\begin{align*}
    \mathcal{W}:\ket{g}\mapsto \sqrt{\frac{n_r}{|G|}}\sum_{r,i_r,j_r}\pi^r_{ij}(g)^* \ket{i_r}_r\ket{j_r}_{\bar{r}},
\end{align*}
which is an isometry of the type we need, namely from $L^2(G)$, the edge degree of freedom, to two factors of equal dimension. 

To show that central flux operators can still be reconstructed properly, we use a special property of this tensor: that $U^L(g)$ can be reconstructed on the left factor, and $U^R(g)$ on the right one. This can be seen as follows:
\begin{align*}
    U^L(h)\ket{g}=\ket{hg}&\mapsto \sqrt{\frac{n_r}{|G|}}\sum_{r,i,j}\pi^r_{ij}(hg)^* \ket{i}_r\ket{j}_{\bar{r}}\\
    &= 
    \left(\sum_{r',k,l}\pi^r_{kl}(h)^*\ketbra{k}{l}_r\otimes\mathds{1}_{\bar{r}}\right)
    \left(\sqrt{\frac{n_r}{|G|}}
    \sum_{r,i,j}\pi^{r'}_{ij}(g)^* \ket{i}_r\ket{j}_{\bar{r}}\right)
\end{align*}
with a similar argument for $U^R$. Since the central flux operators $U^R(g)=U^L(g)$ for $g$ in the center of the group, we have shown that they can be reconstructed on either side. 

We may now turn our attention to the first concern.
From the above argument, the legs in the stack that are to be contracted need to be of dimension $\sum_r n_r$.
This can be achieved by making five codes in the stack themselves stacks of $\lceil \log_{2}(\sum_r n_r) \rceil$ and removing any extra states from the codespace.
Since the six codes in the stack may now have different dimensions, one must be slightly more careful than in the construction of the LOTE code; legs from the sixth code in the stack must always be assigned to the vertex degrees of freedom, and those from the other five codes (which are themselves stacks of codes) must be assigned to the edges, rather than these assignments being arbitrary as before.

Using these two alterations to the LOTE construction we obtain a holographic code with a $G$-gauge-constrained Bulk. By \cref{thm: full_gauge_to_gauge_global} this code has a gauge/global duality and by \cref{thm:gauge_global_to_global_global} we can convert it to a code with a global/global duality. Thus we have constructed a holographic code with a global/global duality for arbitrary finite groups, or in more information-theoretic terminology, transversal gates for an arbitrary finite group.

\subsubsection{Compact Lie groups}
\label{subsubsec: continuous-groups}

In order to implement our ideas for compact Lie groups, we use a truncation scheme which we argue, in principle, implements the full isometry in the appropriate limit. 
The tensor networks still only requires legs associated with finite dimensional Hilbert spaces. 
Throughout this section we are somewhat cavalier about assuming potential subtleties that appear when regularizing the $\delta$ distribution evaluated at group identity elements work out.

In this section we construct an isometry $ \mathcal{W}: L^2[G] \rightarrow \bigoplus_{r\in \Omega_G} \mathcal{H}_{r} \otimes \mathcal{H}_{\bar{r}} $\footnote{
Recall that for a continuous group, $L^2[G]$ is a Hilbert space corresponding to Lebesgue-square-integrable real functions on $G$.
A basis for such functions is given by Dirac delta functions for each group element, just as with real-space quantum wavefunctions spanned by states $\{\ket{x}\}$.%
}, where $\Omega_G$ is the set of all irreducible representations of $G$ \footnote{Strictly speaking we include one representative of each equivalence class of irreducible representations of the group $G$.}.
This isometry is used to construct the appropriate analog of the copying tensor, which is rendered finite by simply truncating the sum over representations.
Specifically we define $\mathcal{W}$ as a sum over maps to each representation, $\mathcal{W} = \bigoplus_{r \in \Omega_G} \mathcal{W}_r$, with $ \mathcal{W}_{r}: L^2[G]\rightarrow V_{r}\otimes V_{\bar{r}} $ defined as
\[ \mathcal{W}_{r}\ket{g}=\sqrt{n_r}\sum_{i,j=1}^{n_{r}} \pi_{ij}^{r}(g) \ket{i}_{r}\otimes \ket{j}_{\bar{r}}, \]
where $\pi_{ij}^r(g)$ are the coefficients of $g$ in representation $g$ as in the previous section.
The adjoint is simply
\[  \mathcal{W}_{r'}^{\dagger}\ket{i}_{r}\otimes\ket{j}_{\bar{r}}  = \delta_{r,r'} \sqrt{n_{r}} \int_{G}\!\dd{g}\pi_{ij}^{r}(g)^{*}\ket{g}.\]

To verify that $\mathcal{W}$ is an isometry, we need to introduce two results first -- the Peter-Weyl theorem, and the grand-orthogonality theorem for continuous groups.
The latter is a straightforward generalization of \cref{eq:GOT}, namely
\begin{equation} \label{eq:GOT-cont}
 \int_{G}\!\dd{g}\pi^r_{ij}(g)^{*}\pi^{r'}_{kl}(g) = \frac{\delta_{r,r'}\delta_{ik}\delta_{jl}}{\sqrt{n_r n_{r'}}}.
\end{equation}
The former is needed in place of our argument in the previous section that relied on equality of Hilbert space dimensions, which does not work for infinite dimensional Hilbert spaces.
Specifically, the Peter-Weyl theorem (see e.g. Ref.~\cite{bump}) serves as a generalization of the Fourier decomposition beyond the abelian case of $U(1)$ when $ G $ is a compact Lie group.
Specifically, for every state $ \ket{\Psi}\in L^{2}[G] $ there is a family of $ n_{r}\times n_{r} $ matrices $ \Psi_{r} $ indexed by $r\in\Omega_G$ (analogous to a set of Fourier coefficients) such that
\begin{align}
    \ket{\Psi}  = \int_{G}\!\dd{g} \sum_{r\in \Omega_{G}} \sqrt{n_r}\Tr[\pi^{r}(g)^{\dagger} \Psi_{r}] \ket{g} \label{eq:PWthm} ,
    % \braket{g}{\Psi}  = \sum_{r\in \Omega_{G}} \sqrt{n_r}\Tr[\pi^{r}(g)^{\dagger} \Psi_{r}]  \label{eq:PWthm} ,
\end{align}
%where $ \Omega_{G} $ is the set of irreducible representations of $ G $.
% 
%Using the later theorem, and applying it to the decomposition of a generic state implied by the Peter-Weyl theorem, \cref{eq:PWthm}, we get 
Applying both of these results to the definition of $\mathcal{W}_r$ gives
\[ {\mathcal{W}}_{r}\ket{\Psi} = \sum_{i,j=1}^{n_r}\left(\Psi_{r}\right)_{ij} \ket{i}_{r}\otimes \ket{j}_{\bar{r}}.\]
Then $\mathcal{W}$ is an isometry, because
\begin{align}
 \mathcal{W}^\dagger \mathcal{W} \ket{\Psi} &= \sum_{r,r' \in \Omega_G} {\mathcal{W}}_{r'}^\dagger {\mathcal{W}}_{r}\ket{\Psi} \\
 &= \sum_{r \in \Omega_G} \sum_{i,j=1}^{n_r} \sqrt{n_{r}} \int_{G}\!\dd{g}\pi_{ij}^{r}(g)^{*} \left(\Psi_{r}\right)_{ij} \ket{g} \\
 &= \ket\Psi, \ \text{ by \cref{eq:PWthm}.}
\end{align}
\newcommand{\Wmc}{\mathcal{W}}
We can then define a (truncation) operation $ \mathcal{W}_{\Sigma}: L^{2}[G]\rightarrow \oplus_{r\in\Sigma} \left(V_{r} \otimes V_{\bar{r}} \right) $, where $ \Sigma \subseteq  \Omega_{G}$, as follows
\[  \Wmc_{\Sigma} = \bigoplus_{r\in\Sigma} \Wmc_{r}. \]
By the definition of $\mathcal{W}_r$ and its adjoint, this operator satisfies 
% \[ \Wmc^{\dagger}_{r}\Wmc^{}_{r}= \int_{G}\!\dd{g}\int_{G}\!\dd{h}\sqrt{n_r}\Tr_{r}[\pi(h^{-1}g)]\ket{h}\bra{g}\]
% we can shift $g\rightarrow h\cdot g$ to obtain
% \begin{align*}
%      \Wmc^{\dagger}_{r}\Wmc^{}_{r}&= \int_{G}\!\dd{g}\int_{G}\!\dd{h}\sqrt{n_r}\Tr_{r}[\pi(h^{-1})]\ket{h\cdot g}\bra{g}\\
%      & \int_{G}\!\dd{h} U^{L}(h) \Tr_{r}[]
% \end{align*}
\begin{equation} \label{eq:delta_Sigma}
    \bra{h}\mathcal{W}_{\Sigma}^{\dag}\mathcal{W}^{}_{\Sigma}\ket{g}=\delta_{\Sigma}(h^{-1}\cdot g)
\end{equation}
 where 
\[\delta_{\Sigma}(g)=\sum_{r\in\Sigma}n_{r}\Tr_{r}[\pi^{r}(g)]. \]
Furthermore, $\delta_{\Sigma}(g)$ converges to a Dirac delta distribution when $ \Sigma \rightarrow \Omega_{G}$ because ${\Wmc_{\Sigma} \rightarrow \Wmc }$. 
\cref{eq:delta_Sigma} implies that $\mathcal{W}^{\dagger}_{\Sigma}\mathcal{W}^{}_{\Sigma}  $ commutes with the left and right action of $ G $ because the RHS expression is invariant if we left or right multiply $g$ and $h$ by any group element, i.e.\
\begin{align} \label{eq:ww-commute}
   [\mathcal{W}^{\dagger}_{\Sigma}\mathcal{W}^{}_{\Sigma} , U^{L/R}(g) ] = 0.
\end{align}
Finally, we note that although $\mathcal{W}_\Sigma$ is not an isometry, $\mathcal{W}^{\dagger}_{\Sigma} $ is, which one can verify using the observation that 
\[ \mathcal{W}_{r'} \cdot \mathcal{W}_{r}^{\dagger}= \delta_{r,r'} \id_{V_{r}\otimes V_{\bar{r}}}.
\]
With this truncated isometry in hand, we can define a truncated gauging map for compact Lie group $ G $ and a finite set $\Sigma$ of irreducible representations of $ G $:
\[  \mathcal{G}_{\Sigma} =  \mathcal{W}^{E}_{\Sigma} \mathcal{G} ,\]
where
\begin{align*}
 \mathcal{G} \ket\psi_V =\prod_{v\in V_{1}}\int_{G}\!\dd{g_{v}}A_{v}(g_{v}) \ket{\psi}_{V}\bigotimes_{e\in E} \ket{I}_{e}
\end{align*}
is the standard generalization of the gauging map (which maps to an infinite-dimensional Hilbert space), and
\[\mathcal{W}^{E}_{\Sigma} = \prod_{e\in E} \mathcal{W}^e_{\Sigma} \]
is the truncation acting on every edge degree of freedom to render the target space finite-dimensional.

%We want to show an analogous result to our previous one showing that
Unlike the original gauging map $\mathcal{G}$, this truncated version is not an isometry.
From its definition, one can see that
\[P_{\Sigma}=\mathcal{G}^{\dagger}_{\Sigma}\cdot \mathcal{G}_{\Sigma} =
\prod_{v\in V_{1}}\int_{G}\!\dd{g_{v}}U_{v}(g_{v}) \prod_{e\in E} \delta_{\Sigma}\left( g_{e_{-}}\cdot g_{e_{+}}^{-1}\right),
\]
is a Hermitian operator acting on $\mathcal{H}_V$ that is not necessarily proportional to the identity because the $\delta_{\Sigma}$ does not enforce all $g$ to be $I$. Nevertheless when $\Sigma \rightarrow \Omega_G$, we find that as in the finite $G$ case
\[ P = \lim_{\Sigma\rightarrow \Omega_G} P_{\Sigma} 
=
\left(
\prod_{v\in V_{1}} \int_{G}\!\dd{g_{v}} \prod_{e\in E} \delta_{}\left( g_{e_{-} }\cdot g_{e_{+}}^{-1}\right)\right) \idV.
\]
A caveat in the last expression is that if $\abs{E}>\abs{V_1}$ we are left with positive powers of $\delta(I)$ which are not well-defined numbers. A way to regularize the above is to understand it as the limit 
\[ \delta(I) = \lim_{\Sigma\rightarrow \Omega_G} \delta_{\Sigma}(I)=\lim_{\Sigma\rightarrow \Omega} \sum_{r\in\Sigma}n_r^2. \]

We can construct a regularized version $P_{\Sigma}$ by dividing by an appropriate power of $\delta_{\Sigma}(I)$ so that it converges to $\idV$, i.e. $\tilde{P}_{\Sigma} = \delta_{\Sigma}(I)^{-\alpha}P_{\Sigma}$ satisfies $\lim_{\Sigma\rightarrow \Omega_G}\tilde{P}_{\Sigma} =\idV $, where $\alpha \in \mathbb{R}$. Next notice that $\tilde{P}_\Sigma$ is a Hermitian operator on a finite Hilbert space $\mathcal{H}_{V}$ that converges to the identity. Thus for large enough $\Sigma$, all the eigenvalues of $\tilde{P}_{\Sigma}$ should be close to $1$ which in particular means that they are all positive. Therefore, for some (finite) $\Sigma$ one can define an operator $K_{\Sigma}$ which is the (positive) square root of ${P}_{\Sigma}$, i.e.
\[ K_{\Sigma} = \sqrt{P_{\Sigma}} = \delta^{\alpha/2}_\Sigma(I) \sqrt{\tilde{P}_{\Sigma}}.\]

What is more, as $ P_{\Sigma} $ commutes with $ U_{V}(h) $, its square root still commutes with $ U_{V}(h) $! Thus for an appropiate $\Sigma$, we can construct an isometry from $\mathcal{G}_\Sigma$ by 
\[ \tilde{\mathcal{G}}_{\Sigma} = \mathcal{G}_{\Sigma} \cdot K_{\Sigma}^{-1}.\]
% The drawback of this new map is that is not quite local because of $ K_{\Sigma} $. 

Finally, we would like to show that this map implements a gauge/global symmetry, i.e.\
\begin{align}
\label{eq:isometry-covariant}
\tilde {\mathcal{G}}_{\Sigma} \cdot U_{V}(h) =A_{V_{0}}(h) \cdot \tilde{\mathcal{G}}_{\Sigma},
\end{align}
with $A_{V_0}(h)$ defined to act on the truncated, finite-dimensional space.
First, we show that this holds for the (non-isometric) $\mathcal{G}_\Sigma$, i.e.\ 
\begin{align}
\label{eq:truncated-covariant}    
 {\mathcal{G}}_{\Sigma} \cdot U_{V}(h) =A_{V_{0}}(h) \cdot {\mathcal{G}}_{\Sigma}.
 \end{align}
This holds as follows,
\begin{align*}
    \mathcal{W}_{\Sigma}^{E\dagger} \mathcal{G}_{\Sigma} \cdot U_{V}(h) &= \mathcal{W}_{\Sigma}^{E\dagger}\mathcal{W}_{\Sigma}^{E}  \cdot \mathcal{G} \cdot U_{V}(h) \\
    &= \mathcal{W}_{\Sigma}^{E\dagger}\mathcal{W}_{\Sigma}^{E}  \cdot A_{V_{0}}(h) \cdot \mathcal{G} \text{, using \cref{thm:gauge-map-duality}%
    \footnote{One may object that we only proved this for a finite-dimensional setting; however all the reasoning follows through to the continuous case.} %
    ,} \\
    &= A_{V_{0}}(h) \cdot\mathcal{W}_{\Sigma}^{E\dagger}\mathcal{W}_{\Sigma}^{E}  \cdot  \mathcal{G} \text{, using \cref{eq:ww-commute},} \\
     \mathcal{G}_{\Sigma} \cdot U_{V}(h) &=  
\left(\mathcal{W}_{\Sigma}^{E}\cdot A_{V_{0}}(h)
\cdot \mathcal{W}_{\Sigma}^{E\dagger}\right)
\cdot \mathcal{G}_{\Sigma}\text{, since $\mathcal{W}_\Sigma^\dagger $ and thus $\mathcal{W}_\Sigma^{E\dagger}$ is an isometry}. 
\end{align*}
%To show this, note $ \mathcal{W}_{\Sigma}^{E\dagger} \mathcal{G}_{\Sigma} $ acts as the untruncated gauging map followed by $ \mathcal{W}_{\Sigma}^{E\dagger}\mathcal{W}_{\Sigma}^{E} $.
 %to pull $ U_{V}(h) $ pass the first gauging map to get $  $%
%.
%Because it acts on the edges as right- or left-multiplication by group elements, this commutes with $ \mathcal{W}_{\Sigma}^{E\dagger}\mathcal{W}_{\Sigma}^{E}  $.
%We can then hit the left hand side with $ \mathcal{W}_{\Sigma}^{E} $ to obtain
Note that $ h \mapsto A_{V_{0}}^{(\Sigma)}(h) =
\left(\mathcal{W}_{\Sigma}^{E}\cdot A_{V_{0}}(h)
\cdot \mathcal{W}_{\Sigma}^{E\dagger}\right)$ is a representation of the group $ G $ due to the fact that  $ \mathcal{W}_{\Sigma}^{\dagger}\mathcal{W}_{\Sigma}^{}  $ commutes with the left and right actions, which is how $A_{V_0}(h)$ acts on the edge degrees of freedom.
Thus this is our ``finite-dimensional version'' of $A_{V_0}(h)$.

To expand to the proper isometry, $\tilde{\mathcal{G}}_\Sigma$, we show that $P_\Sigma$ commutes with the action of the group.
Consider the square of \cref{eq:truncated-covariant}, namely
\begin{align}
 U_{V}(h)^\dagger \cdot {\mathcal{G}}_{\Sigma}^\dagger \cdot  {\mathcal{G}}_{\Sigma} \cdot U_{V}(h) &= {\mathcal{G}}_{\Sigma}^\dagger \cdot A^\dagger_{V_0}(h) A_{V_{0}}(h) \cdot {\mathcal{G}}_{\Sigma}.\\
 U_{V}(h)^\dagger \cdot P_\Sigma \cdot U_{V}(h) &=P_\Sigma
\end{align}
Thus $P_\Sigma$ commutes with the group action, which implies that $K_\Sigma^{-1}$ commutes with it, implying \cref{eq:isometry-covariant}.

One may now wonder why using $\tilde{\mathcal{G}}_\Sigma$ on a generalized gauged LOTE code does not violate the Eastin-Knill theorem~\cite{eastinRestrictionsTransversalEncoded2009}, which says that codes capable of correcting erasure errors can not have a continuous transversal gate set. The reason is that by truncating we have lost the property that locality of operators is preserved up to dressing (see \Cref{subsec: dressing-local-operators}). This was necessary to argue that the resulting code kept the same entanglement wedge map. After truncation, the reconstruction properties holds only approximately, in which case the code can have a continuous transversal gate set, as shown in Ref.~\cite{Faist_2020}.

\section{Discussion}
\label{sec:conclusion}
While toy models of AdS/CFT remain far from perfect, there is much to be learned, in both holography and quantum information, by attempting to improve these models. In this work we took a step in this direction by constructing and analyzing codes with Bulk gauge symmetries that were dual to Boundary global symmetries. 

We began by introducing the gauging map, which maps a system with global symmetry to another with gauge symmetry. We proved several key features of the map: that it can be isometric, that it preserves locality of vertex operators up to a dressing to the Boundary which can be chosen to go along an arbitrary path, and its global/gauge duality. 
We proceeded by introducing holographic codes. We gave a working definition of a holographic code, which allowed the Bulk to be a constrained system, meaning that only constraint-preserving operators need be reconstructable. We also required %only 
a %weaker 
version of entanglement wedge reconstruction that  foreshadowed the compatibility of holographic codes and the gauging map: that an entanglement wedge must be ``near boundary probing'', that is there must always be a path from a point in the entanglement wedge to one on its corresponding Boundary region. We then gave two examples of our definition: the HaPPY code, whose Bulk is unconstrained and which exhibits a global/global duality, and the gauged LOTE code, whose Bulk's Hilbert space was constrained to a full gauge-invariant Hilbert space, but which did not a priori have a gauge/global duality. We then proved, however, that for holographic codes with such a gauge-invariant Bulk, there must always be a dual Boundary global symmetry. Thus we have constructed the first toy model of a holographic code with a gauge/global duality, but also, shown that establishing the presence of such a duality (given a Bulk gauge symmetry) requires only the error correcting properties of AdS/CFT, not any additional structure. 

In the last section of the paper, we used the gauging map to analyze holographic codes. We began by establishing some general Bulk gauging and ungauging theorems: a code with a global/global duality can be converted into one with a gauge/global duality, albeit with all fluxes fixed to a particular value. A code with gauge/global duality with Bulk containing at least one flux sector can be converted to one with global/global duality by restricting the Bulk to that sector, i.e. by fixing the gauge field configuration to a classical value. This latter fact inspired the various remaining observations in the paper. First, it offers an explanation for why the HaPPY code fails to model AdS/CFT by displaying a global/global symmetry: its Bulk is already restricted to a fixed-flux sector. Second, it gives us a model of a somewhat trivial time evolution, in which nearby sites are completely decoupled from one another.
Thirdly, it gave a concrete model in which one can understand how metric fluctuations contribute to AdS/CFT being only an approximate code.
Finally, when combined with the generalized gauged LOTE code it gives us a new mechanism for creating approximate codes that are transversal with respect to a universal gate set. 

We expect it to be straightforward to generalize this work to higher dimensions and to higher-form global symmetries. The most straightforward future direction is quantifying the code approximation factor of the codes with global/global continuous symmetry dualities as a function of the number of representations included in the truncation. Specifically, it would be interesting to know whether or not it saturates the upper bound given in Ref~\cite{Faist_2020}. 
It would also be interesting to explore such bounds further in the context of toy models for AdS/CFT.  
Our gauging and ungauging approach to produce gauge/global or global/global holographic codes could also be applied to random (stabilizer) tensor network constructions~\cite{haydenHolographicDualityRandom2016,Nezami2020} to avoid any ``accidental'' symmetries. 

The discussion in this work has focused on toy models for holography close to a fixed background geometry.
Recently there has been progress on establishing no-go results that constrain possible symmetries in a theory of quantum gravity that involves summing over topologies~\cite{Harlow2021,Belin2020,Hsin2020,Chen2021}. It would be interesting to extend the holographic code toy models to this setting and to study the possibilities for symmetry there.
This would appear to require some new ideas, as here we have relied upon finding a path to the boundary in our dressing of operators to make them gauge-invariant, which may not always be possible.
The resolution to this potential pitfall may be the incorporation of wormholes into the toy models. 

The kind of dynamical duality that we construct here is only a fairly trivial one, in which subsystems are fully decoupled (i.e.\ the speed of light vanishes).
However, a similar construction may be capable of implementing a proper dynamical duality between $k$-local Hamiltonians in the bulk and boundary.
This could involve a map similar to the gauging map, but whose gauge/global duality involves symmetries that are \emph{not} simply generated by sums of local operators, but rather sums of $k$-local terms. 
%(or the corresponding locality-preserving unitaries).
One approach would be to start by studying such a map in a general (non-holographic) setting, as we did in \cref{sec:generalized-gauging-map}, before applying it to a code like LOTE as we did in \cref{sec:applications-to-holography} to obtain a holographic code with global/global duality.
%We also speculate that time evolution may appear by a similar mechanism, namely that there is a Bulk with constraints in it and that in some subspace of this Bulk there is an isometry that preserves holographic reconstruction and has a spacetime gauge/global duality.
%It would be interesting to try and look for such a Bulk system and map without reference to a holographic code. If we can find such a system and map then we can use the strategy used to build the gauged LOTE code to artificially build the constrained system in the Bulk. 
Conversely, if we find holographic codes with a proper dynamical duality (i.e.\ allowing for time evolution around a fixed Bulk geometry), it may be possible to construct from it a natural ambient constrained system, which may be of interest to the study of quantum gravity. 

\acknowledgments

We thank Patrick Hayden for providing useful comments on a draft of this work. 
We thank Beni Yoshida for permitting us to adapt several figures from Ref.~\cite{HaPPY} for use in this work. 
DW acknowledges support from the Simons Foundation. KD is supported by the Center for Science of Information (CSoI), an NSF Science
and Technology Center, under grant agreement CCF-0939370.

\appendix

\section{Error correction}

\label{sec:qecc}

We define correctability of regions as follows.
\begin{definition}
	Consider a code with encoding isometry $\mathcal{V}:\mathcal{H}_L \to \mathcal{H}_P$ and codespace projector $P=\mathcal{V} \mathcal{V}^\dagger$.
	A region $R$ is said to be \emph{correctable} with respect to an algebra of logical operators $\mathcal{A}_L$ if the erasure channel $\mathcal{N}_R (\cdot)\defi \tr_R(\cdot)$ is correctable; i.e.\ for every $\mathcal{O}_L \in \mathcal{A}_L $ there exists a recovery channel $\mathcal{R}$ such that
	\begin{align}
		P (\mathcal{R} \circ \mathcal{N}_R )^\dagger(\mathcal{V}\mathcal{O}_L\mathcal{V}^\dagger) P = \mathcal{V} \mathcal{O}_L \mathcal{V}^\dagger
	\end{align}
	%any logical unitary $U_L\in \mathcal{A}_L$ there will be a unitary codespace-preserving operator $U$ supported on $R^c$ that implements $U_L$; i.e.\ $V^{\dagger}UV = U_L$. 
	\label{def:correctable}
\end{definition}

	The most useful property of a correctable region in a holographic code is a result of the \emph{reconstruction} or \emph{cleaning} lemma, which shows that this is equivalent to

\begin{lemma}
	A region $R$ of a code with encoding isometry $\mathcal{V}$ and codespace projector $P$ is correctable if and only if for any logical operator $\mathcal{O}_L \in \mathcal{A}_L$, there is a physical operator $\mathcal{O}_{R^c}$ supported on $R^c$ which preserves the codespace, i.e.\ $[\mathcal{O}_{R^c},P] = 0$, and implements $\mathcal{O}_L$, i.e.\ $\mathcal{V}^\dagger \mathcal{O}_{R^c}\mathcal{V} = \mathcal{O}_L$.
	\label{lem:cleaning-lemma}
\end{lemma}

The forward direction of this lemma is partially proven in Ref.~\cite{preskillpastawski}, but without the condition that $[\mathcal{O}_{R^c},P] = 0$.
We prove this result here using results from \cite{BKK,harlow}.
\begin{proof}
    By \cite{BKK}, correctability is equivalent to the following condition:
    \begin{align}
		\forall Y_L \in \mathcal{A}_L, \forall \mathcal{O}_R, [P\mathcal{O}_RP,\mathcal{V}Y_L \mathcal{V}^\dagger] = 0, \label{eq:bkk}
    \end{align}
	where $\mathcal{O}_R$ is any operator acting on the region $R$.
    
    $(\impliedby)$: Suppose that \cref{eq:bkk} does not hold.
	Then $\exists \mathcal{O}_R, Y_L \in \mathcal{A}_L$ such that $[P\mathcal{O}_R P , \mathcal{V}Y_L \mathcal{V}^\dagger] \neq 0$.
    But by assumption, $\exists Y_{R^c} $ s.t.\ $PY_{R^c}P = \mathcal{V}Y_L \mathcal{V}^\dagger$ and $[Y_{R^c},P] = 0$.
	Therefore $[P\mathcal{O}_RP, \mathcal{V}Y_L\mathcal{V}^\dagger] = P [\mathcal{O}_R , Y_{R^c}] P = 0$, a contradiction.

    $(\implies)$: Shown in Ref.~\cite{harlow}.
\end{proof}

\begin{lemma}
	Consider a code with encoding isometry $\mathcal{V}:\mathcal{H}\rightarrow\mathcal{H}_L$.
	Suppose $\mathcal{H}=\mathcal{H}_R\otimes \mathcal{H}_{R^c}$ and that $R$ is correctable with respect to an algebra of logical operators $\mathcal{A}_L$. Then for any logical unitary $U_L\in \mathcal{A}_L$ there is a unitary codespace-preserving operator $U$ supported on $R^c$ that implements $U_L$; i.e.\ $U\mathcal{V} = \mathcal{V}U_L$. 
	\label{lem:logical-unitary-to-physical-unitary}
\end{lemma}
\begin{proof}
	Define $H_L \defi -i \log(U_L)\in \mathcal{A}_L$ (the algebra is closed under the logarithm function which amounts to addition and multiplication via its power series).
	By the cleaning lemma, \cref{lem:cleaning-lemma}, there exists a codespace-preserving operator $H$ supported on $R^c$ such that $H\mathcal{V} = \mathcal{V}H_L$.
	We claim that the operator $U = \exp(\frac{i}{2}( H + H^\dagger))$ is unitary and reconstructs $U_L$ on the region $R^c$.

	First, we show unitarity.
	$H_L$ is Hermitian by definition.
	$H$ need not be Hermitian, but one can construct $\tilde{H} \defi \frac{1}{2} (H + H^\dagger)$ which is.
	Furthermore, $\tilde{H}$ implements $H_L$ because $\mathcal{V} \tilde{H}^\dagger \mathcal{V}^\dagger = H_L^\dagger = H_L$.
	It is also supported on $R^c$ by assumption.
	Thus, $U$ is unitary and also supported on $R^c$.

	Second, consider the operator implemented by $U$.
	One can expand $\mathcal{V}^{\dagger}\exp(iH) \mathcal{V} $ as a power series to conclude that $\mathcal{V}^{\dagger} U \mathcal{V} = U_L$.
	
	Finally, $U$ is also codespace-preserving, which again can be seen from the power series and $[H,P]=0$.

\end{proof}

\begin{lemma}
	For any correctable region $R$ with respect to an algebra of logical operators $\mathcal{A}_L$, and any logical unitary representation $U_L: G \to \mathcal{A}_L$, there is a unitary representation of $G$ supported on $R^c$ that implements $U_L$.
	\label{lem: logical-representation to physical representation}
\end{lemma}

\begin{proof}
From \cref{lem:logical-unitary-to-physical-unitary}, we have a set of unitaries $\left\{U(g)\right\}_{g \in G}$ that are codespace-preserving, supported on $R^c$, and implement the respective unitaries $U_L(g)$.

%Now, $U(g_1) U(g_2) V = U(g_1) V U_L(g_2) = V U_L(g_1) U_L(g_2) = U_L(g_1 \cdot g_2) = U(g_1 \cdot g_2) V $
Using the fact that the $U(g) V = V U_L(g)$ and the representation property of $U_L$, we have $U(g_1) U(g_2) V =  U(g_1 \cdot g_2) V $.
From this,
\begin{align}
    U(g_1) U(g_2) V &=  U(g_1 \cdot g_2) V  \\
    U(g_1) U(g_2) \Pi_c &=  U(g_1 \cdot g_2) \Pi_c  \\
    U(g_1) U(g_2) \Tr_R \Pi_c &=  U(g_1 \cdot g_2) \Tr_R \Pi_c ,
\end{align}
where we have multiplied by $V^\dagger$ on the right and then traced out $R$.
We can then multiply on the right by $\Tr_R(\Pi_c)^+$, the generalized inverse of $\Tr_R(\Pi_c)$, and use $\Tr_R(\Pi_c) \Tr_R(\Pi_c)^+ =: \Pi_{S}$, the projector onto $S$, which we define as the support of the partial trace of the codespace projector.
We obtain
\begin{align}
    U(g_1) U(g_2) \Pi_{S} &=  U(g_1 \cdot g_2)  \Pi_{S}.
\end{align}
Thus we have a representation on this subspace.
We now just need to consider the orthogonal complement $S^\perp$.

We claim the following: (1) $U(g)$ is block-diagonal with respect to the decomposition of $R^c = S \oplus S^\perp$, and (2) the block acting on $S^\perp$ can be chosen freely without affecting the codespace-preserving property of $U$ or its action on the codespace; in particular we can choose it to be $\id_{S^\perp}$ and thus act as a trivial representation of $G$.

Proof of (1):

We know that $U$ is codespace-preserving; thus
\begin{align}
(U \otimes \id_R ) \Pi_c &= \Pi_c (U \otimes \id_R )\\
    U \Tr_R(\Pi_c) &= \Tr_R(\Pi_c) U \\
    U \Pi_S &= \Tr_R(\Pi_c) U\Tr_R(\Pi_c)^+ \\
    \Pi_S^\perp U \Pi_S &= \Pi_S^\perp \Tr_R(\Pi_c) U\Tr_R(\Pi_c)^+ \\
    &= 0, 
\end{align}
where we traced out the $R$ subsystem on line two, and on the last line used $ \Tr_R(\Pi_c) = \Pi_S  \Tr_R(\Pi_c) $ and $\Pi_S ^\perp \Pi_S = 0$.
Similarly one can show $\Pi_S U \Pi_S^\perp = 0$;  thus $U$ is block-diagonal as claimed.

Proof of (2):

Replace $U$ with $U' = \Pi_S U \Pi_S + \Pi_S^\perp$, where we use $\Pi_S^{(\perp)}$ as shorthand for $\Pi_S^{(\perp)} \otimes \id _R $.
First, we show that this is still codespace preserving by demonstrating $\Pi_S \Pi_c = \Pi_c$.

Given positive semidefinite operators $M$, $N$, it is known that $M^{1/2} N M^{1/2}$ is also positive semidefinite; thus $ \Pi_S^\perp \Pi_c \Pi_S^\perp  \succeq 0$.
However, it is also traceless:
\begin{align}
\Tr( \Pi_S^\perp \Pi_c \Pi_S^\perp )
    &=  \Tr_{R^c}( \Pi_S^\perp \Tr_R\Pi_c )\\
    &= 0,
\end{align}
But the trace of a positive semidefinite operator is only zero if the matrix equals zero, so $\Pi_S^\perp \Pi_c \Pi_S^\perp = |\Pi_S^\perp \Pi_c | ^2 = 0$; thus $\Pi_S^\perp  $ and $\Pi_c$ are orthogonal and $\Pi_S \Pi_c = \Pi_c$.
Finally,
\begin{align}
    U' \Pi_c 
    &= U' \Pi_S \Pi_c \\
    %&=  U \Pi_S \Pi_c \\
    &=  U \Pi_c \\
    &=  \Pi_c U \\
    &=  \Pi_c \Pi_S U \\
    &=  \Pi_c U' ,
\end{align}
so $U'$, like $U$, is codespace-preserving.

Finally, the action on the codespace is the same:  $\Pi_c U' \Pi_c = \Pi_c \Pi_S U' \Pi_S \Pi_c = \Pi_c U \Pi_c$.
Unitarity also follows straightforwardly from the unitarity and block-diagonality of $U$, and thus we have shown the result.

\end{proof}

\section{Gauging isometry proofs} \label{app:gauge-proofs}

%REFERRED TO IN SEC 2
\subsection{Gauge-invariance of Wilson loops}
\begin{theorem}\label{app:ProofWilsonLoopsI}
All Wilson loops are gauge invariant and commute with each other.
\end{theorem}

Note that 
\begin{align*}
    U_{e}^{R}(h)^{\dagger} \cdot W_{ij}^{e} \cdot U_{e}^{R}(h) &= \sum_{k} W_{ik}^{e} r_{kj}(h^{-1}),\\
    U_{e}^{L}(h)^{\dagger} \cdot W_{ij}^{e} \cdot U_{e}^{L}(h) &= \sum_{k}  r_{ik}(h)W_{kj}^{e}.
\end{align*}
So that after conjugation transpose and flipping indices, we get
\begin{align*}
    U_{e}^{R}(h)^{\dagger} \cdot W_{ij}^{\bar{e}} \cdot U_{e}^{R}(h) &= \sum_{k} r_{ik}(h)W_{kj}^{\bar{e}} \\
    U_{e}^{L}(h)^{\dagger} \cdot W_{ij}^{\bar{e}} \cdot U_{e}^{L}(h) &= \sum_{k} W_{ik}^{\bar{e}}r_{kj}(h^{-1})
\end{align*}

In order to proof the gauge-invariance, we need to show that every projector $\Pi_{v}$ for $v\in V_{1}$ commutes with the Wilson loops. If $v$ is not in the Wilson loop, then they commute trivially. Otherwise, there is an incoming $e^{-} \in E^{-}(v)$ and a outcoming $e^{+}\in E^{+}(v)$ that are part of the Wilson loop (note that $v$ is in $V_1$). Then we only need to show that 
\[ [\Pi_{v}, \left(W^{\overline{e^{-}}}W^{\overline{e^{+}}}\right)_{ij}]=0  \]
which follows from
\[ [ U^{R}_{e^{+}}(g^{})U^{L}_{e^{-}}(g^{}), \left(W^{\overline{e^{-}}}W^{\overline{e^{+}}}\right)_{ij}]=0\]
as the Wilson loop trivially commutes with the other parts of $A_v(g)$.

Consider
\newcommand{\Wp}{W^{\overline{e^+}}}
\newcommand{\Wm}{W^{\overline{e^-}}}
\newcommand{\URp}{U^{R}_{e^{+}}}
\newcommand{\URm}{U^{R}_{e^{-}}}
\newcommand{\ULp}{U^{L}_{e^{+}}}
\newcommand{\ULm}{U^{L}_{e^{-}}}

\begin{align*}
    \left(\URp(g^{})\ULm(g^{})\right)\left(\Wm\Wp \right)_{ij} \left(\URp(g^{})\ULm(g^{})\right)^{\dagger} &= 
    \sum_{k} \left( \ULm(g^{})\cdot \Wm_{ik} \cdot \ULm(g^{-1})   \right) \\
    &\quad \quad\quad \left( \URp(g^{})\cdot \Wp_{kj} \cdot \URp(g^{-1})   \right) \\
    &= \sum_{klm}\left( \Wm_{im}r_{mk}(g^{})\right)\left( r_{kl}(g^{-1})\Wp_{lj}\right) \\
    &= \sum_{lm}\left( \Wm_{im}r_{ml}(g\cdot g^{-1})\Wp_{lj}\right) \\
    &= \sum_{lm}\left( \Wm_{im}\delta_{lm}\Wp_{lj}\right) \\
    &= \left(\Wm\Wp\right)_{ij}
\end{align*}

\subsection{Flux free sector}
\begin{lemma}[Restatement of \cref{lemma:subgraph-gauging-lemma}] \label{proof:subgraphGaugingLemma}
\subgraphGaugingLemma
%Consider a system $\mathcal{T}$ obtained by gauging a system transforming under a global symmetry, $(V, \{\mathcal{H}_v\}, G, \{U_v\})$ gauged with respect a graph $\Lambda=(V,E)$. Let $\mathcal{S}'$ be the pre-gauged system. Consider a subgraph $\Gamma=(V_\Gamma,E_\Gamma)$ with $V_{\Gamma}\subseteq  V_{1}$. For any $v\in V_\Gamma$, define the $\mathcal{T}$-operator
%\begin{align*}
%    A^\Gamma_{v}(g)=U_v(g)\prod_{e\in E_\Gamma^+(v)}U_{e}^{R}(g_v) \prod_{f\in E_\Gamma^-(v)}U_{e}^{L}(g_v).
%\end{align*}
%For any $Q_\Gamma\in \mathcal{A}_{\mathcal{S}'}(V_\Gamma)$, the operator
%\begin{align*}
%O_\Gamma=\DgvG \quad \left (\prod_{v\in V_\Gamma}A^\Gamma_v(g_{v})\right)
%Q_\Gamma
%\left (\prod_{v\in V_\Gamma}A^\Gamma_v(g_{v})\right)^\dagger    
%\end{align*}
%is an element of $\mathcal{A}_\mathcal{T}(\Gamma).$
\end{lemma}

\begin{proof}\label{proof:subgraph-gauging-lemma}
Clearly $O_{\Gamma}$ has support only on $\Gamma$. We need to show that $[O_{\Gamma},\Pi_{GI}]=0$. For $v\in V\setminus V_{\Gamma}$, $A_{v}(g_v)$ commutes with $O_{\Gamma}$ since the only possible overlap is on edges but all edges of $E_{\Gamma}$ avoid vertices on $V\setminus V_{\Gamma}$. Thus, we only need to keep track of the $v\in V_{\Gamma}$, i.e. we need to show
\begin{align}
\label{eq: gamma-commutator}
   [O_\Gamma,\prod_{v\in V_\Gamma}\int dg_v A_v(g_v)]=0.
\end{align}
Now notice that 
\begin{align}
    A_{v}(g_v)A^{\Gamma}_{v}(h_v) &= 
    A_{v}^{\Gamma}(g_v h_v)  \bigotimes_{e \in E^{+}(v) \setminus E_{\Gamma}} U_{e}^{R}(g_{v}^{}) \bigotimes_{f \in E^{-}(v) \setminus E_{\Gamma}} U_{f}^{L}(g_{v}^{}) \label{eq:1}\\
    \left(A^{\Gamma}_{v}(h_v)\right)^{\dagger}A_{v}(g_v) &= \left(A_{v}^{\Gamma}(g_v^{-1} h_v) \right)^{\dagger} \bigotimes_{e \in E^{+}(v) \setminus E_{\Gamma}} U_{e}^{R}(g_{v}^{}) \bigotimes_{f \in E^{-}(v) \setminus E_{\Gamma}} U_{f}^{L}(g_{v}^{}) \label{eq:2}
\end{align}
then 
\begin{align*}
    O_{\Gamma} \DgvG A_v(g_v)&= \DgvG\dd{h_v}
    \left(\prod_{v\in V_{\Gamma}} A^{\Gamma}_{v}(h_v)\right) 
    Q_{}
    \prod_{v\in V_{\Gamma}} \left(A^{\Gamma}_{v}(h_v)^{\dagger} A_{v}(g_v)\right).
\end{align*}
We can use \cref{eq:2}, make a change of variables $h_v \rightarrow g_v \cdot h_v$ and then use \cref{eq:1} to arrive to the desired result. 
\end{proof}

% \begin{proof} \kd{need to rewrite}
% We need only prove that $[O_\Gamma,\Pi]=0$. Recall that
% \begin{align*}
%     \Pi=\prod_{v\in V}\int dg_v A_v(g_v).
% \end{align*}
% Since for $v\not\in V_\Gamma$, $A_v(g_v)$ does not act on degrees of freedom in $\Gamma$, it is sufficient to show that

% \begin{align}
% \label{eq: gamma-commutator}
%   [O_\Gamma,\prod_{v\in V_\Gamma}\int dg_v A_v(g_v)]=0
% \end{align}
% where for simplicity we have, and will for the remainder of the proof, drop extraneous factors of identity. The first term in the commutator is 
% \begin{align*}
%     \int \prod_{v\in V_\Gamma}dg_vdh_v \quad \left (\prod_{v\in V_\Gamma}A^\Gamma_v(g_{v})\right)\left(O_v\bigotimes_{e\in E_\Gamma} \ketbra{I}{I}_e\otimes\ketbra{I}{I}_{E_0 }
%     \right)\left (\prod_{v\in V_\Gamma}A^\dagger_v(h_v)A^\Gamma_v(g_{v})\right)^\dagger.
% \end{align*}
% From here it is possible to confirm \ref{eq: gamma-commutator} by doing the changing of variables $g_v\rightarrow h_vg_v$, and using both that $A^\dagger_v(h_v)A^\Gamma_v(h_v)$ has no support on $\Gamma$ and $[A^\dagger_v(h_v),A^\Gamma_v(h_v)]=0$ to get equality with the second term of the commutator. \kd{reword}\\
% \end{proof}

% \begin{theorem}
% \label{thm: undressing-local-ops}
% Need to show that for any gauged operator supported on a subgraph $\Gamma$, that it can be pushed across $\mathcal{G}$ to an operator supported only on the vertices of $\Gamma$
% \end{theorem}

% }

\begin{theorem}[Restatement of \cref{thm: gauging-map-image-is-flux-free}] \label{proof:gauging-map-image-is-flux-free}
\fluxfreegauge
\end{theorem}

The inclusion 
\[ \image[\mathcal{G}] \subseteq  \image[\Pi_{FF}] \]
was shown in the main text.
To prove the other inclusion, we need Lemma 3.3 of Ref.~\cite{Cui_2020}, which makes reference to group elements defined by Wilson lines defined over cycles in the following way.
Consider a Wilson line along a cycle such that the trace is not taken at the end, distinguishing it from a Wilson loop.
Then its eigenstates are group states of the edges along the cycle that have a well-defined product (or ``holonomy''), which is the group element $g_\gamma$ assigned to the cycle $\gamma$ for that state.
We refer to the set of such cycles that start and end at a given point $v\in V$ as $\mathcal{C}_v$.
Then Lemma 3.3 of Ref.~\cite{Cui_2020} states that two states $\ket{\psi}, \ket{\psi'} \in \mathcal{H}_E$ are related by a gauge transformation $\ket{\psi} =  \prod_{v\in V} \bar{A}_{v}(h_v) \ket{\psi'}$, with $\bar{A_v}$ just the action of $A_v$ restricted to edge degrees of freedom, if and only if for any $v\in V$ there is an element $h\in G$ such that for all $\gamma \in \mathcal{C}_v$, the eigenvalue $g_\gamma$ for $\ket\psi$ and the eigenvalue $g'_\gamma$ for $\ket{\psi'}$ are related by $g_\gamma = h g_\gamma h^{-1}$.

%recall that the flux-free sector is characterized by having all eigenvalues of the Wilson loop and line operators as $\ket{\{I_{e}\}}_{E}$. As the character we are using to define the Wilson operators , the previous condition imply that any holonomy along any path \footnote{By holonomy along a path we mean the product of the group elements along the path.} is in the conjugacy class of $I_{e}$, i.e. all the holonomies are trivial. 
%Then Lemma 3.3 of Ref.~\cite{Cui_2020} implies that 
%Now, consider some state in the flux-free sector. 
For now, we focus on the state of the edges, $\bigotimes_{e\in E} \ket{ g_e } \in \mathcal{H}_E$.
For this state, and any point $v\in V_1$ and any cycle $\gamma \in \mathcal{C}_v$, we have $\tr r(g_\gamma) = \tr r(I)$.
By character theory, and the faithfulness of $r$, this means $g_\gamma = I$.
%But now consider any state with all edge states set to $I$, i.e.\ $\ket{\psi}_V \otimes \ket{\{I_e\}}_E$.
Then by the above lemma, this is related by a gauge transformation to the state with all edge states set to $I$, i.e.\ 
\begin{align}
 \bigotimes_{e\in E} \ket{ g_e }= \prod_{v\in V} \bar{A}_{v}(h_v) \bigotimes_{e\in E} \ket{I_e}  \label{eq:equivalent-up-to-gauge},
\end{align}
where $\bar{A}_v(h)$ is the gauge constraint operator $A_v(h)$ restricted to edges only, i.e.\ in the full Hilbert space we have $ A_{v}(h) = U_{v}(h) \otimes \bar{A}_{v}(h)$.

Our goal now is to show that the state $\Pi_{GI} \ket{\psi}_V \otimes \bigotimes_{e\in E} \ket{ g_e }$ lies in the image of $\mathcal{G}$. Let's start by using \cref{eq:equivalent-up-to-gauge}:
% Now we are ready to show explicitly that our full state $\Pi_{GI} \ket{\psi}_V \otimes \bigotimes_{e\in E} \ket{ g_e }$ lies in the image of $\mathcal{G}$ as follows.
%Returning now to the full state, $\ket{\psi}_{V} \in \mathcal{H}_V$, then 
\begin{align*}
    \Pi_{GI} \ket{\psi}_V \otimes \bigotimes_{e\in E} \ket{ g_e }
    &= \Pi_{GI} 
    \prod_{v\in V} \bar{A}_{v}(h_v)
    \ket{\psi}_V \otimes \bigotimes_{e\in E} \ket{I_e} \text{, using \cref{eq:equivalent-up-to-gauge}}\\
    &= \Pi_{GI}\  
    \prod_{v\in V_1} {A}_{v}(h_v)  {U}_{v}(h_v^{-1})  
    \prod_{v\in V_0} \bar{A}_{v}(h_v) 
    \ket{\psi}_V \otimes  \bigotimes_{e\in E} \ket{I_e} \text{, by def.\ of $\bar{A}$} \\
    &= 
    \Pi_{GI}
    \ \prod_{v\in V_0} \bar{A}_{v}(h_v) 
    \ \prod_{v\in V_1} {U}_{v}(h_v^{-1}) \ket{\psi}_V \otimes \bigotimes_{e\in E} \ket{I_e} \\
    &\hspace{3cm} \text{, using that $\Pi_{GI}A_{v}(g) = \Pi_{GI}$ for $v\in V_1$.  }\\
    % &=  \prod_{v'\in V_0} \bar{A}_{v'}(h) \ \mathcal{G} \prod_{v\in V} {U}_{v}(h_v^{-1}) \ket\psi _V \text{, via \cref{thm:gauge-map-duality}}\\
    &= 
    \Pi_{GI}
    % \ \prod_{v\in V_1} {U}_{v}(h_v^{-1}) \ket{\psi}_V 
    \ \ket{\psi'}_V 
    \otimes \bigotimes_{e\in E_1} \ket{I} 
    \otimes \bigotimes_{e\in E_0} \ket{h_e} 
\end{align*}
where $h_{e} = h_{e^{+}}^{-1} $ and $\ket{\psi'} _V   =\prod_{v\in V_1} {U}_{v}(h_v^{-1}) \ket{\psi}_V$.

As the NGC to NGC Wilson lines from $e$ to $f$, with $e,f\in E_0$, are trivial, this implies that $h_{e}h_{f}^{-1}=I$ whenever there is a NGC to NGC path from $e$ to $f$ which is guaranteed because $(V_1,E_1)$ is connected. Therefore, there exist some $h\in G$ such that for every $e\in E_{0}$ we have $h_e = h$. Note that
\[ \prod_{v \in V_1}\bar{A}_{v}(h) \bigotimes_{e\in E}\ket{I}_e = \bigotimes_{e\in E_1} \ket{I}_e 
    \otimes \bigotimes_{e\in E_0} \ket{h}_e   \]
then
\[ \ket{\psi'}_V 
\otimes \bigotimes_{e\in E_1} \ket{I}_e 
    \otimes \bigotimes_{e\in E_0} \ket{h}_e =
    \prod_{v \in V_1}{A}_{v}(h) \
    \prod_{v \in V_1}{U}_{v}(h)^{\dagger}
    \ket{\psi'}_V 
\otimes \bigotimes_{e\in E}\ket{I_e}_e \] 
so that 
\[\Pi_{GI} \ket{\psi}_V \otimes \bigotimes_{e\in E} \ket{ g_e }_e = 
\Pi_{GI} \ 
    \ket{\psi''}_V
\otimes \bigotimes_{e\in E}\ket{I_e} = \mathcal{G}\ket{\psi''}_V\]
with $\ket{\psi''}_V = \prod_{v \in V_1}{U}_{v}(h)^{\dagger}
    \ket{\psi'}_V $.
    
Thus we have shown
\[ \image[\mathcal{G}] = \image[\Pi_{FF}]. \]

% \begin{align*}
% \prod_{v\in V_1} {A}_{v}(h_v) \left(\ket{\psi}_{V}\otimes \ket{\{I_e\}}_E\right) &=  
% \prod_{v\in V_1} {U}_{v}(h_v) \otimes
% \prod_{v\in V_1} \bar{A}_{v}(h_v)
% \left(\ket{\psi}_V \otimes \ket{\{I_e\}}_E\right)      \\
% &= \left(\prod_{v\in V_1} {U}_{v}(h_v)\ket{\psi}_{V} \right)\otimes \ket{\{g_e\}}_{E}.
% \end{align*}
% %%
% We can then act with $\Pi_{GI}$ in the above expression and use the fact that for any vertex $v$ and group element $h$, $\Pi_{GI}A_{v}(g) = \Pi_{GI}$ to get to 
% \[ \mathcal{G}\ket{\psi}_{V} = \Pi_{GI} \left(\prod_{v\in V_1} {U}_{v}(h_v)\ket{\psi}_{V}\right) \otimes \ket{\{g_e\}}_{E}. \]
%The last equation shows that if $\ket{\{g_e\}}_{E}$ is flux-free, then there exists a unitary operator $\mathcal{U}_{\{g_e\}}$ such that 
% \[ \mathcal{G} \left(\mathcal{U}_{\{g_e\}}\ket{\psi}_{V}\right) = \Pi_{GI} \ket{\psi}_{V} \otimes \ket{\{g_e\}}_{E}. \]
% where $\mathcal{U}_{\{g_e\}} = \prod_{v\in V_1} {U}_{v}(h_v^{-1})$.

\begin{theorem} \label{thm: undressing-local-ops}
Consider a subgraph $\Gamma= (V_{\Gamma}, E_{\Gamma})$ such that for every point $v\in V_1\setminus V_\Gamma$ there is a path between $v$ and $V_0$ fully contained in $\Lambda\setminus\Gamma$. Let $O_{\Gamma} $ be an %gauge invariant
operator supported on $\Gamma$ that commutes with $\Pi_{FF}=\Gmc\Gmc^{\dagger}$. Then there exists an operator $O$ acting on $\mathcal{H}_V$ with support on $V_{\Gamma}$ such that 
\[ O_{\Gamma}\mathcal{G} = \mathcal{G}O.\]
\end{theorem}

\begin{proof}
Consider the operator $O = \mathcal{G}^{\dagger} \, O_{\Gamma} \, \mathcal{G}$. Then 
\[
\mathcal{G}O= \Gmc^{}\Gmc^{\dagger} O_{\Gamma} \Gmc =  O_{\Gamma}\Gmc^{}\Gmc^{\dagger} \Gmc =  O_{\Gamma}\Gmc^{}.
\]
Next consider any operator $T_v$ with support on $v \in V_1\setminus V_{\Gamma}$. Then one can use \cref{thm: operator-dressing-theorem} with the subgraph $\mathcal{P}_v$ being the path from $v$ to $V_0$ that avoids ${\Gamma}$ so that $\Gmc T_{v} = T_{\mathcal{P}_v} \Gmc$. Then we can write
\begin{align*}
    O_{\Gamma}\, T_{\mathcal{P}_v} \,\Gmc & = \Gmc O \, T_{v}, \\
    T_{\mathcal{P}_v} \,O_{\Gamma}\,\Gmc & = \Gmc T_{v} \, O.
\end{align*}
Note that as $\tilde{O}$ and $O_{\mathcal{P}_v}$ have non-overlapping support they trivially commute, thus
\[
[O, T_v]  = \Gmc^{\dagger}\Gmc [O, T_v] =  \Gmc^{\dagger}[O_{\Gamma}, T_{\mathcal{P}_v}]\Gmc = 0 .
\]
which means that $O$ commutes with all $T_v$ for $v\in V_{1}\setminus V_{\Gamma}$ therefore $O$ acts as $\id$ on $\mathcal{H}_v$. 

Then one can see that $O$ acts trivially on $V_0 \setminus V_{\Gamma}$ by explicitly expanding its definition:
\[
O = 
\left(\prod_{v\in V_1}\int_{G}\!\dd{g}_v\int_G\dd{h}_v \right)
\left(\prod_{v\in V_1} U_{v}^{\dagger}(g_v)\right)
\left(\Tr_{E}[O_{\Gamma} \prod_{e} \ketbra{h_{e^{-}}h_{e^{+}}^{-1}}{g_{e^{-}}g_{e^{+}}^{-1}}_e] \right)
\left(\prod_{v\in V_1} U_{v}^{}(h_v)\right).
\]
The partial trace over $E$ of the operator with support on $\Gamma$ has support on $V_{\Gamma}=\Gamma \setminus E$. Then as the unitaries multiplying the partial trace has support on $V_1$, this implies that the full operator has support on $V_{\Gamma}\cup V_1$.

The above two results imply that $O$ has support on 
\begin{align*}
    \textrm{support}(O)&\subseteq \left(V\setminus \left(V_1\setminus V_{\Gamma} \right)\right) \cap \left(V_{\Gamma}\cup V_1\right)\\
    & \subseteq\left(V_0 \cup V_{\Gamma} \right) \cap \left(V_{\Gamma}\cup V_1\right) = V_{\Gamma}.
\end{align*}
\end{proof}

\subsection{Other flux sectors}
\label{subsec: other-flux-sectors-proofs}

%REFERRED TO IN 3.6
\begin{theorem}
\label{thm: general-twisted-map-is-an-isometry}
If $V_0\neq\varnothing$ and $\Lambda$ is connected, then $\mathcal{G}_{\phi}$ is an isometry up to an overall normalization factor.
\end{theorem}
\begin{proof}
Similar to before, consider 
\[ \Pmc_{\phi} = \Gmc^{\dagger}_{\phi} \Gmc^{}_{\phi}. \]
We want to show that the above map is still proportional to the identity.

Let us unpack $\mathcal{G}|\psi\rangle_V$ by plugging in the explicit form of $\Pi_{GI}$ (see definition \ref{def: gauging-definition}).
\begin{align*}
    \mathcal{G}_{\phi}|\psi\rangle_V =
    \Dgv
    \prod_{v\in V_{1}} U_v(g_v)\ket{\psi}_V\bigotimes_{
    e\in E}|g^{}_{e^-}\phi_{e}g^{-1}_{e^+}\rangle_e
    % \bigotimes_{E_0 \in E_0 }|g^{}_{E_0 ^-}\rangle_{E_0 }
\end{align*}

From where we can read 
\begin{align*}
    \Pmc_{\phi}=
    \Dgv
    U_v(g_v)
    \prod_{e\in E}\delta(\phi_{e}^{-1}g^{}_{e^-}\phi_{e}g^{-1}_{e^+})
\end{align*}

The effect of the delta distributions is to ensure that the only terms contributing in the integrations are those with $\phi_{e}g_{e^{+}}=g_{e^{-}}\phi_{e}^{}$ for every vertex $E$. As usual if $e^{+}\in V_0$, we have $g_{e^{+}}=I$. From the last observation, we see that actually all $g_{v}$ are set to one. Then we get the same constant as the untwisted case $\phi_{e} = I$.

\end{proof}

%REFERRED TO IN 3.6
\begin{theorem}
\label{thm: general-twisted-gauge-global-duality}
Consider a system transforming under a global symmetry, $(\mathcal{S}, G, \{U_v\})$  with $\mathcal{S}=(V, \{\mathcal{H}_v\})$ with $\Lambda=(V,E)$ and labeled by $\alpha$ such that $V_0 \neq \varnothing $. Let $\mathcal{T}$ be obtained from gauging $\mathcal{S}$. Let $\mathcal{G}$ be the corresponding gauging isometry. We have 
% \textbf{Recall what $A_{V_0}$ is }
\begin{align*}
    \mathcal{G}_{\phi}U_{V}(h) = A_{V_0}(h)\mathcal{G}_{^{h}\!\phi}.
\end{align*}
where $^{h}\!\phi: e \longmapsto  h^{-1}\cdot \phi_{e} \cdot h^{}$.

\end{theorem}

\begin{proof}
% In order to prove the equality, we will show that for every $\ket{\psi}_V\in \mathcal{H}_V$ we have 
% \[ \mathcal{G}_{\phi}U_{V}(h)\ket{\psi}_V = A_{V_0}(h)\mathcal{G}_{^{h}\!\phi}\ket{\psi}_V.\]
% where for any $W\subseteq  V$ we define 
% \[ U_{W}(h) = \prod_{v\in W} U_{v}(h).\]
% 
Expanding the left hand side we get
\begin{align*}
    LHS=\int \prod_{v\in V_{1}}\dd{g_v} \left(\prod_{v\in V_1}U_v(g_vh)\right)
    U_{V_0}(h)
    \bigotimes_{e\in E_1 }|g^{}_{e^-}\phi_{e}g^{-1}_{e^+}\rangle_e\bigotimes_{e\in E_0 }|g_{e^-}\phi_{e}\rangle_{e}
\end{align*}
where $E_0  = \{e \in E | e^{+}\in V_0\}$ and $E_1  = E\setminus E_0 $ and $U_{V_0}(h) = \prod_{v\in V_{0}} U_{v}(h)$.
Next, we perform the change of variables $g'_v=g_vh$ to get
\begin{align*}
    LHS&=\int \prod_{v\in V_{1}}\dd{g'_v} \left(\prod_{v\in V_1}U_v(g_v')\right)
    U_{V_0}(h)
    |\psi\rangle_V\bigotimes_{e\in E_1 }|(g')^{}_{e^-}\left(h^{-1}\cdot \phi_{e}\cdot h^{}\right)(g')^{-1}_{e^+}\rangle_e\\
    &\hspace{10cm}\bigotimes_{e\in E_0 }U^{R}_e(h)\ket{g'_{e^-}h^{-1}\phi_{e}h^{}}_{e} \\
    &=U_{V_0}(h) \prod_{e\in E_0 } U^R_{e}(h) \mathcal{G}_{^{h}\!\phi}\ket{\psi}_V = A_{V_0}(h)\mathcal{G}_{^{h}\!\phi}\ket{\psi}_{V}.
\end{align*}
\end{proof}

\begin{remark}
\label{remark: adding-overlapping-lines}
Note that if $p = (u,v_1,\dots,v_{N},w)$ and $q=(w,v_{N},\dots, v_{N+M},u')$ then 
\[ 
\sum_{j} W_{ij}^{p} W_{jk}^q = W_{ik}^{c}
\]
where $c$ is the path $c=(u, v_1,\dots, v_{N}, \dots, v_{N+M},u')$ which is obtained by deleting $w$ in $p$ and $q$ and then concatenating them. This is true because $\sum_{j}W_{ij}^{\overline{(u,w)}}W_{jk}^{\overline{(w,u)}}=
\sum_{j}W_{ji}^{{(w,u)}}W_{kj}^{{(u,w)}}$ and thus acting on an state $\ket{g}_{(u,w)}$ gives $\sum_{j} r_{ji}(g)r_{kj}(g^{-1}) = r_{ki}(I)=\delta_{ji}$.
\end{remark}

\begin{lemma} \label{lem:same-sector-gauging-non-trivial}
Consider a system transforming under a global symmetry, $(\mathcal{S}, G, \{U_v\})$, with $\mathcal{S}=(V, \{\mathcal{H}_v\})$, and a connected graph $\Lambda=(V,E)$ with $\alpha$ labeling such that $\Lambda\setminus(V_0\cup E_0)$ is still connected. Let $\mathcal{T}$ be obtained from gauging $\mathcal{S}$. Let $\ket{\{g_e\}},\ket{\{h_e\}}\in\mathcal{H}_E$ be two states with the same NGC to NGC Wilson line eigenvalues, trivial Wilson loop eigenvalues, and trivial NGC to NGC line eigenvalues when the line starts and ends at the same point. Then there exist states $\ket{\{g^0_e\}},\ket{\{h^0_e\}}$ with trivial NGC to NGC Wilson lines and trivial Wilson loops, and a map $f:V_0\rightarrow G$ such that
\begin{align*}
    \ket{\{g_e\}}&=\prod_{v\in V_0}\bar{A}_v(f(v))\ket{\{g^0_e\}}\\
    \ket{\{h_e\}}&=\prod_{v\in V_0}\bar{A}_v(f(v))\ket{\{h^0_e\}}
\end{align*}
\end{lemma}

\begin{proof}
Note that the NGC to NGC Wilson lines do measure the product along the path they follow because we are using a faithful representation and we do not take a trace.
Next, notice that as the Wilson loops, are all trivial, then the product of all the elements around a close loop in $E_1$ gives the identity.
Together, these last two facts imply that the eigenvalue for an NGC to NGC Wilson line between two fixed edges in $E_0$ is path-independent.
Furthermore, if we choose any one fixed edge $e_0\in E_0$, the eigenvalues of all the NGC-NGC Wilson lines are determined by knowing the eigenvalues for the line operators starting at $e_0$ and ending at every other point in $E_0 \setminus \{e_0\}$.
This is because when we matrix multiply two NGC-NGC Wilson line operators (the second beginning where the first one ends) we obtain another NGC-NGC Wilson line that follows that concatenation for the initial paths but skipping the repeated edge (see \cref{remark: adding-overlapping-lines}).   

Let $p_k$ be a path that goes from $e_0$ to $e_k$ and $l_{e_{k}}$ be the group element such that 
\[ W_{ij}^{p_k} \ket{\{g_e\}}_E = r_{ij}(l_{e_k}) \ket{\{g_e\}}_E\]
where $k=1,\dots,\abs{E_0}-1$ and $e_{k}$ are the different points of $ E_0\setminus\{e_0\}$. Thus the state configuration defined as
\[
\ket{\{g^0_e\}}_E \equiv 
\left( \prod_{e\in E_{0}\setminus\{e_0\} } U^{R}_{e}(l_{e}^{-1})\right) \ket{\{g_e\}}_E\]
has trivial eigenvalues for NGC-to-NGC Wilson lines.
To see this, note that each $W^{p_k}_{ij}$ now measures the product $ (g^0_{{e_k}})^{-1}\dots g^0_{e_0}  = (g_{{e_k}} l_{e_k})^{-1}\dots  g_{e_0} = l_{e_k}^{-1} l_{e_k} = I$.

Note that $\ket{\{h_e\}}_E$ have the same $l_{e_k}$'s as they have the same NGC-to-NGC Wilson lines, so 
\[
\ket{\{ h^0_e\}}_E= 
 \prod_{e\in E_{0}\setminus\{e_0\} } U^{R}_{e}(l_{e}^{-1})   \ket{\{h_e\}}_E\]
 also have trivial NGC Wilson loop eigenvalues. 

The last step is to rewrite $\left( \prod_{e\in E_{0}\setminus\{e_0\} } U^{R}_{e}(l_{e}^{-1})  \right)$ as some product of $\bar{A}_v$ with $v\in V_0$. For $v$ with only one edge this is clear. If $v$ participates in at least two edges, say $e$ and $e'$, i.e. $e,e'\in E_0$ such that $e^+=e'^+$, we need that $l_e=l_{e'}$. This is true because, since $\Lambda\setminus(V_0\cup E_0)$ is connected, we can consider a loop that leaves $v$ through $e$ and returns through $e'$. Since by assumption this loop has trivial eigenvalues, we must have $l_e=l_{e'}$.
Then we can rewrite the product over $E_0\setminus{e_0}$ as a product over $V_0$ of $V_0$ gauge transformations with
\[
%\left( \prod_{e\in E_{0}\setminus\{e_0\} } U^{R}_{e^{+}}(l_{e^{+}}^{-1})  \right)= \prod_{v_k} \bar{A}_{v_k}(f(v_k)) 
f(v_k) =  l_{e_k}^{-1} , 
\]
%where $f(v) = l_{e}$ such that $e^{+}=v$.
and $f(v_0) = I$, such that the lemma holds.
\end{proof}

% From the proof above we could relax the constraint by a condition that would let us rewrite the product of $U^R$ as $\bar{A}$'s. One such condition can be as follows. First, we extend the definition of Wilson loop operators to contain at most one vertex in $V_0$. Then, the condition is that the eigenvalue for these operators are all trivial. 

%REFERRED TO IN 3.6
\begin{theorem}
\label{thm: NGC-map-image-is-full-sector-sometimes}
Consider a system transforming under a global symmetry, $(\mathcal{S}, G, \{U_v\})$, with $\mathcal{S}=(V, \{\mathcal{H}_v\})$, and a connected graph $\Lambda=(V,E)$ with $\alpha$ labeling such that $\Lambda\setminus(V_0\cup E_0)$ is still connected.
Let $\mathcal{T}$ be obtained from gauging $\mathcal{S}$.
 Let $\mathcal{G}_\phi$ be an NGC flux map, i.e.\ $\mathcal{G}_\phi=\prod_{v\in V_0}A_v(h_v)\mathcal{G}$.
 Then the image of $\mathcal{G}_\phi$ is an entire flux sector of the gauge-invariant Hilbert space.
 Specifically, its image is characterized as having the same eigenvalues for Wilson loops and NGC to NGC Wilson lines as some state $\ket{\{\phi_e\}}=\prod_{v\in V_0}\bar{A}_v(g^0_v)\ket{\{I\}}_E\equiv \prod_{v\in V_0}A_v(g^0_v)\ket{\{I\}}_E \in \mathcal{H}_E$.
\end{theorem}

\begin{proof}
As in \cref{thm: gauging-map-image-is-flux-free}, it is sufficient to show that for any $\ket{\{g_e\}}\in\mathcal{H}_E$ with the same eigenvalues for Wilson loops and NGC to NGC Wilson lines as $\ket{\{\phi_e\}}$, and for any $\ket{\psi}_V\in \mathcal{H}_V$ that 
    $\Pi_{GI}\ket{\psi}_V\otimes\ket{\{g_e\}}$
is in the image of $\mathcal{G}_\phi$. The NGC to NGC loops starting at the same vertex are trivial, because conjugation doesn't change the identity. Thus we can use \cref{lem:same-sector-gauging-non-trivial} to the states $\ket{\{\phi_e\}}$ and $\ket{\{g_e\}}$ to get the promised objects $\ket{\{\phi^0_e\}},\ket{\{g^0_e\}},$ and $f$. Using the same arguments as in \cref{thm: gauging-map-image-is-flux-free}, we can write 
\begin{align*}
    \ket{\{g^0_e\}}=\bar{A}_{V_0}(w)\prod_{u\in V_1}\bar{A}_u(h_u)\ket{\{\phi^0_e\}}
\end{align*}
with $w\in G$.
Thus 
\begin{align*}
    &\Pi_{GI}\ket{\psi}_V\otimes\ket{\{g_e\}}\\
    &= \Pi_{GI}\ket{\psi}_V\otimes\left(\prod_{v\in V_0}\bar{A}_v(f(e))\bar{A}_{V_0}(w)\prod_{u\in V_1}\bar{A}_u(h_u)\ket{\{\phi^0_e\}}\right)\\
    &=\prod_{v\in V_0}\bar{A}_v(f(e))\bar{A}_{V_0}(w)\Pi_{GI}\ket{\psi}_V\otimes\left(\prod_{u\in V_1}\bar{A}_u(h_u)\ket{\{\phi^0_e\}}\right)\\
    &=\prod_{v\in V_0}\bar{A}_v(f(e))\bar{A}_{V_0}(w)\Pi_{GI}\ket{\psi'}_V\otimes\ket{\{\phi^0_e\}}\\
    &=\prod_{v\in V_0}\bar{A}_v(f(e))\Pi_{GI}\ket{\psi''}_V\otimes\ket{\{\phi^0_e\}}\\
    &=\Pi_{GI}\ket{\psi''}_V\otimes\ket{\{\phi_e\}},
\end{align*}
where we absorbed $\prod_{u\in V_1}\bar{A}_u(h_u)$ in $\ket{\psi}_V$ by inserting gauge transformations after $\Pi_{GI},$ and we absorbed $\bar{A}_{V_0}(w)$ using the gauge/global duality of flux free maps. By definition, this is in the image of $\mathcal{G}_\phi$.

\end{proof}

\bibliographystyle{unsrtnat}
\bibliography{biblio}

\end{document}